\documentclass[a4paper,12pt]{article}
\usepackage{epsf, amssymb,amsmath,dsfont}
\usepackage{epsfig}
\usepackage{cite}
\usepackage{hyperref}
\usepackage{xcolor}
\usepackage{graphicx}
\usepackage{tikz}
\usepackage{bbold}
\usepackage{mathtools}
\usepackage{comment}
\usepackage[dvipsnames]{xcolor}
\usepackage{ragged2e}

\counterwithin*{equation}{section}

\hypersetup{hidelinks} 

\makeatletter
\renewcommand\section{\@startsection {section}{1}{\z@}%
                                   {-3.5ex \@plus -1ex \@minus -.2ex}
                                   {2.3ex \@plus.2ex}%
                                   {\normalfont\large\bfseries}}

\renewcommand\subsection{\@startsection{subsection}{2}{\z@}%
                                     {-3.25ex\@plus -1ex \@minus -.2ex}%
                                     {1.5ex \@plus .2ex}%
                                     {\normalfont\bfseries}}

\makeatother

\newcommand{\bea}{\begin{eqnarray}}
\newcommand{\eea}{\end{eqnarray}}
\newcommand{\beq}{\begin{equation}}
\newcommand{\eeq}{\end{equation}}
\newcommand{\bem}{\begin{pmatrix}}
\newcommand{\eem}{\end{pmatrix}}
\newcommand{\bl}{\begin{align}}
\newcommand{\el}{\end{align}}
\newcommand{\bld}{\begin{aligned}}
\newcommand{\eld}{\end{aligned}}

\def\ca{{\cal A}}
\def\cb{{\cal B}}

\def\caBH{{\cal M}}

\def\caR{\ca_R}
\def\caRH{\ca_{RH}}
\def\caB{\ca_B}
\def\caBH{\ca_{BH}}

\def\cbR{\cb_R}
\def\cbBH{\cb_{BH}}
\def\cbB{\cb_{B}}
\def\cbRH{\cb_{RH}}

\newcommand{\Ind}{\text{Ind}}

\newcommand{\ket}[1]{|#1\rangle}

\usepackage{amsmath}
\usepackage{amssymb}
\usepackage{amsthm}
\usepackage{mathtools}

\newtheorem{theorem}{Theorem}[section]
\newtheorem{proposition}[theorem]{Proposition}

\theoremstyle{definition}
\newtheorem{definition}[theorem]{Definition}

\theoremstyle{remark}
\newtheorem{remark}[theorem]{Remark}

\begin{document}

\centering 
${}$
\thispagestyle{empty}

\vskip 1.5cm {\large {\bf Algebraic Complexity and Black Hole Complementarity}}
 
\vskip 1cm    { Aude Corbeel$^1$}, { Jingxin Tu$^1$},  { Pim van den Heuvel$^1$}, \\ { Jeremy van der Heijden$^{2}$}, { Erik Verlinde$^1$}\\
{\footnotesize\vskip 0.5cm  \textit{$^1$Institute for Theoretical Physics, University of Amsterdam\\
Science Park 904,
1090 GL Amsterdam,  
The Netherlands\\
\vskip 0.5cm $^2$Department of Physics and Astronomy, University of British Columbia\\ 6224 Agricultural Road, Vancouver B.C. V6T 1Z1, Canada\\}}
{\vskip 0.5cm \texttt{\href{mailto:a.n.i.corbeel@uva.nl}{a.n.i.corbeel@uva.nl}, \href{mailto:j.tu2@uva.nl}{j.tu2@uva.nl}, \href{mailto: p.vandenheuvel@uva.nl}{p.vandenheuvel@uva.nl}, \href{mailto: jeremy.vanderheijden@ubc.ca}{jeremy.vanderheijden@ubc.ca}, \href{mailto: e.p.verlinde@uva.nl}{e.p.verlinde@uva.nl}}}

\vspace{1cm}

\begin{abstract}
\baselineskip=16pt

We develop an operator algebraic generalization of the Yoshida--Kitaev information recovery protocol that applies to von Neumann algebras of arbitrary type. The construction is based on finite-index inclusions, with the Jones basic construction and canonical endomorphisms providing the central algebraic tools. Unlike the finite-dimensional qubit description, the infinite-dimensional theory exhibits new structural phenomena that play an essential role in information recovery. In particular, the diary information is represented non-locally by a choice of Pimsner--Popa basis associated with the inclusion. Exploiting the Temperley--Lieb relations satisfied by the Jones projections, we introduce an algebraic notion of computational complexity and show that it is naturally measured by the Jones--Kosaki index. For irreducible depth-two inclusions, we demonstrate that the information transfer is implemented by an algebraic Fourier transform. Finally, we discuss a potential spacetime interpretation of the construction, including the emergence of an island algebra in terms of (non-)isometric embeddings and an operator algebraic realization of black hole complementarity.
\end{abstract}
\raggedright
\justifying

\newpage

\setcounter{page}{1}
\tableofcontents

\newpage

\section{Introduction}

In the past decade, important steps have been made towards resolving the black hole information paradox. There is a general consensus that, after the Page time, an observer with access to a sufficient amount of the radiation can in principle decode the quantum information escaping from the black hole \cite{hayden2007blackholes}. The decoding operation involves distilling a set of qubits from the Hawking radiation that are maximally entangled with specific modes inside the black hole, and that are later emitted as Hawking quanta. In practice, this operation is difficult to implement and is associated to a high degree of computational complexity \cite{harlow2013quantum, yoshida2017efficient}.

The entanglement structure required for the information release by the black hole and subsequent recovery from the radiation is at odds with the conventional entanglement structure needed for a smooth experience for an infalling observer. Some have hoped to resolve this tension by invoking the principle of black hole complementarity \cite{susskind1993strechted}, but it remains at the heart of the well-known firewall paradox \cite{susskind1993strechted, almheiri2013black, almheiri2013apologia} and of the related puzzle that the effective bulk description of the black hole interior requires a much larger Hilbert space than is permitted by the microscopic theory \cite{page1993information}.

Subsequent developments, including the discovery of the ``quantum extremal surface (QES) prescription'' \cite{lewkowycz2013generalised, faulkner2013quantum, dong2016deriving, dong2018entropy}, made clear that information recovery and a smooth semiclassical geometry can only be reconciled by involving radiation degrees of freedom in the description of the black hole interior. The QES prescription enables one to reproduce the Page curve from a semiclassical perspective \cite{almheiri2019entropy, almheiri2020page, penington2020entanglement}, and leads to the insight that, after the Page time, the interior structure of black holes contains an ``island region” which becomes part of the entanglement wedge of the radiation \cite{almheiri2020page}.

While these developments represent an important step towards resolving the black hole information paradox, several challenging questions concerning the microscopic description of the black hole interior still persist. The current consensus is that such a microscopic description involves concepts such as computational complexity \cite{stanford2014complexity, brown2020python, engelhardt2021world, engelhardt2022finding} and non-isometric quantum error-correcting codes \cite{akers2022black}. In these approaches, black holes are modeled as ordinary quantum systems with a finite number of degrees of freedom, and the transmission of information is often described using simple qubit models \cite{hayden2007blackholes, yoshida2017efficient, akers2022black}. This means that their operator algebraic structure differs significantly from the one used in the bulk effective field theory. Crucially, the operator algebras associated to spacetime regions in the effective quantum field theoretic description are based on type III$_1$ von Neumann algebras, while the microscopic, finite-dimensional models are formulated in terms of type I algebras. One can argue that one of the deeper challenges that prevents us from completely resolving the information paradox stems from this mismatch. 

Recent studies of the role of von Neumann algebras within holography have shed some further light on this distinction. Starting with the development of holographic quantum error-correcting schemes adapted to infinite-dimensional von Neumann algebras \cite{Beny:2007ewj, Beny:2008, Kang:2018xqy, Kang:2019dfi, Hollands_2021, Faulkner:2020iou, Faulkner:2020hzi, Furuya:2020tzv, Gesteau:2021jzp}, it has been argued that the emergence of certain type III$_1$ boundary algebras in the $N\to \infty$ limit is tantamount to obtaining a bulk dual \cite{leutheusser2023causal, leutheusser2023emergent}. These type III$_1$ boundary algebras provide a description of operators acting within spacetime bulk subregions, and exhibit properties that signal the emergence of geometric features, like causality, smoothness of the horizon, and local spacetime symmetries \cite{DeBoer:2019kdj,
Furuya:2023fei, Ouseph:2023juq, deBoer:2025rxx}. Moreover, the theoretical framework of von Neumann algebras leads to a precise formulation of the relevant notions of generalized entropy (see e.g. \cite{Chandrasekaran:2022cip,Chandrasekaran:2022eqq}) needed in the study of the QES prescription and the Page curve. 

The aim of this paper is to bridge part of the gap between the description of information release in the underlying microscopic theory and its description in the emergent bulk spacetime. Our goal is to describe this information transfer and its associated complexity in an algebraic framework that also illuminates the role of black hole complementarity. The algebraic model we develop for this purpose is similar to that introduced by a subset of the authors in \cite{vanderHeijden:2024tdk} in the context of type II$_1$ von Neumann algebras. In that work, a type II$_1$ algebraic formulation of the Hayden–Preskill quantum teleportation protocol, as reformulated by Yoshida and Kitaev \cite{yoshida2017efficient}, was presented using algebraic tools such as algebra inclusions, Jones projections \cite{jones1983index, longo1989index} and the canonical shift \cite{Ocneanu1989operator}. In the present (and in upcoming) work, we will extend and generalize the results of \cite{vanderHeijden:2024tdk} and arrive at a description that is applicable to all types of von Neumann algebras, thereby going beyond the simple qubit models often used in the literature.


One of the central algebraic notions that we will exploit is that of a conditional expectation, which in more intuitive terms may be interpreted as coarse-graining over a subsystem. To simplify our discussions and make the physical interpretation more transparent, we restrict attention to conditional expectations with \emph{finite index}. Intuitively, this means that the coarse-graining is taken over a finite-dimensional quantum system. The existence of such a conditional expectation is related to the statement that, even when the black hole Hilbert space is infinite-dimensional, we want to extract a finite amount of information from it.\footnote{We comment on the applicability of our protocol to infinite index inclusions in the Discussion section.} The information puzzle then amounts to describing the mechanism by which this information is transferred, both from a microscopic and from a spacetime perspective, and to explain how information can be recovered from the radiation algebra. 

The notion of an index was originally introduced by Jones for type II$_1$ algebras \cite{jones1983index} and later extended to more general von Neumann algebras by Kosaki \cite{kosaki1998type}. Jones's index theory was reformulated in a particularly useful way by Longo \cite{longo1989index, Longo:1990zp}, who applied it to type III$_1$ algebras and emphasized the role of isometric and non-isometric embeddings. In the present paper, we will follow both Jones's and Longo’s approach and demonstrate that they provide the necessary tools for information recovery. Because all ingredients are defined for general von Neumann algebras, this framework offers a concrete route toward connecting the microscopic description of the release process to its emergent spacetime description.

This being said, in the present paper, we will remain close to the type II$_1$ language that we used in the previous work \cite{vanderHeijden:2024tdk}, since it allows us to focus on the physics questions related to the black hole information puzzle and initially avoid the full mathematical machinery required to describe inclusions of type III$_1$ algebras. In particular, we will not introduce the concepts of weights, modular crossed products and relative entropy. These aspects will be deferred to an upcoming paper \cite{2Tu2025}, in which we also present the definition of the Jones--Kosaki index in terms of Connes's spatial theory and explain the Page curve in this algebraic language. Throughout the paper, we explicitly indicate which statements hold for von Neumann algebras of arbitrary type and which require the properly infinite setting.

We emphasize that the discussion in this paper is deliberately kept at an abstract level. Our primary goal is to develop an operator algebraic description of information transfer between von Neumann algebras, independent of any particular physical setting. At the same time, motivated by the black hole information problem and in order to make this presentation more accessible to a broader audience, we will adopt suggestive notation and terminology throughout, following \cite{yoshida2017efficient}. It is important to stress, however, that the existence of the algebraic structures in question does not by itself imply a corresponding realization in spacetime. Further discussion of this issue, including possible spacetime interpretations and potential obstacles to implementing the algebraic structures required in quantum gravity, is deferred to the Discussion section. We now briefly summarize the main results using terminology familiar from the black hole information problem.

\subsection*{A physicist's summary of results: an algebraic information puzzle}

\begin{figure}[t]
     \centering       

\begin{tikzpicture}[x=0.75pt,y=0.75pt,yscale=-1,xscale=1]

\draw (141.67,125) node  {\includegraphics[width=203.5pt,height=159.61pt]{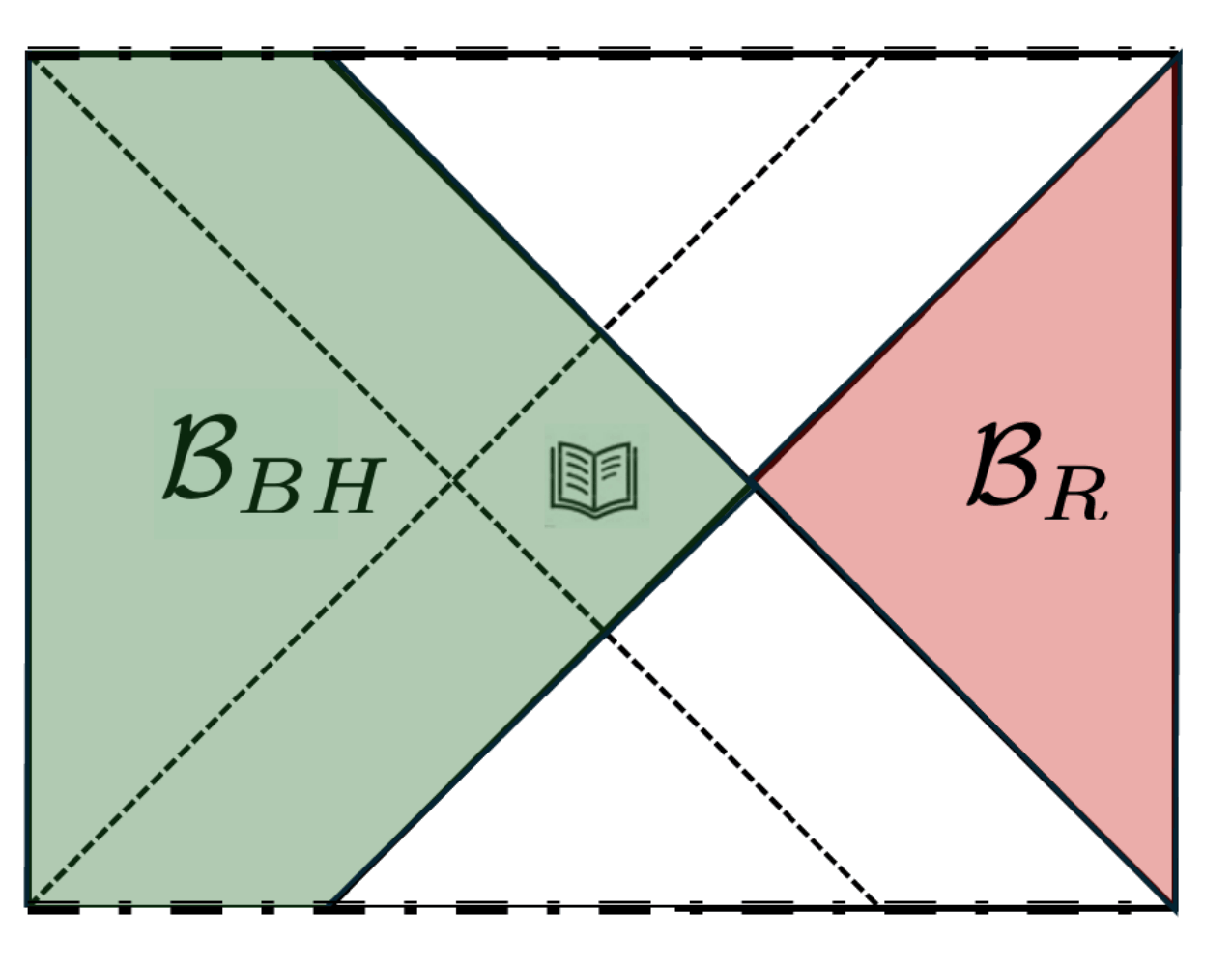}};
\draw (460,125) node  {\includegraphics[width=203.5pt,height=159.61pt]{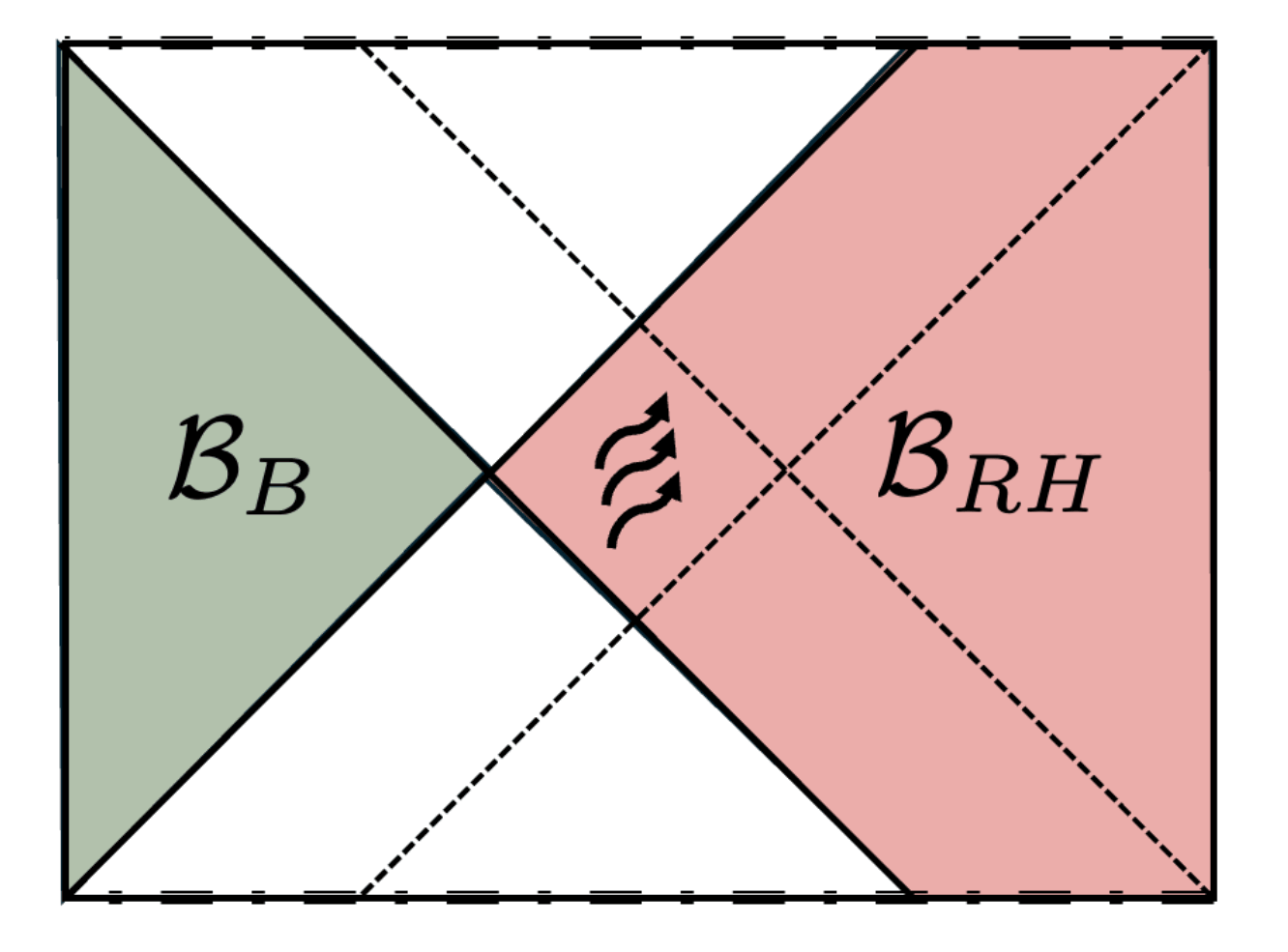}};

\draw (157,95) node [anchor=north west][inner sep=0.75pt]  [font=\small]  {$\ket{0_H}$};
\draw (415,95) node [anchor=north west][inner sep=0.75pt]  [font=\small]  {$\ket{0_I}$};

\end{tikzpicture}
\caption{\small{A pictorial representation of the spacetime algebras before (left) and after (right) the evaporation step. The states $\ket{0_H}$, and $\ket{0_I}$, indicate that the state is in a frozen vacuum near the respective horizons.}}
    \label{Fig:BH-alg-before}
\end{figure}

Our goal is to model a single discrete evaporation step of a black hole that has already passed the Page time. At this stage, the black hole degrees of freedom are sufficiently entangled with the previously emitted Hawking radiation that newly emitted radiation can carry away quantum information from the interior. We will describe this process from two complementary perspectives. The first is a \emph{microscopic} description, where the evaporation process is represented as a unitary evolution of the underlying quantum degrees of freedom. The second is an \emph{effective} description where the same process is reflected by the emission of Hawking particles. Throughout the discussion, we will have in mind the AdS/CFT correspondence, where the microscopic description is provided by the boundary CFT while the effective description is given by bulk effective field theory.\footnote{One of the main advantages of the operator-algebraic framework is that it is formulated independently of this particular setting and may therefore be applicable beyond AdS/CFT.} 

We consider an evaporation step during which a finite amount of quantum information, represented by Alice’s diary, is transferred from the black hole to the radiation. The operators used by Alice to encode the quantum information contained in the diary are denoted by $a$, while the radiation modes eventually received by Bob are denoted by $b$. In the effective description, this process is associated with a small shift in the location of the QES (which will be defined momentarily). We denote the QES before and after the evaporation step by $H$ and $I$, respectively. When the QES is located at $H$, the quantum information contained in Alice’s diary is still localized inside the black hole. The radiation modes $b$ are initially in the vacuum and are maximally entangled with their interior partner modes $\tilde b$, as required by the regularity of the horizon. At this stage, the $b$-modes carry no information about the diary.

As the QES moves from $H$ to $I$, the quantum information is transferred from the black hole to the outgoing radiation. Once the process is complete, the operators $b$ no longer describe vacuum Hawking modes but instead encode the quantum information originally contained in Alice’s diary, although in a complicated form. Simultaneously, the original interior modes $a$ become maximally entangled with their own partner modes $\tilde a$, restoring the local vacuum structure at the new location $I$ of the QES. Thus, from the spacetime perspective, the evaporation step is naturally interpreted as an entanglement swap.

In this setting, we introduce operator algebras as follows (see Figure~\ref{Fig:BH-alg-before} for a pictorial representation\footnote{We comment on potential obstructions to realizing our protocol for algebras associated to spacetime subregions in the Discussion section.}).
The complete algebra on the left of the extremal surface $H$ corresponds to the black hole degrees of freedom and will be denoted by ${\cal B}_{BH}$. The algebra on the right of the QES describes the radiation and is denoted by ${\cal B}_R$. In our model, the fact that $H$ is a quantum extremal surface corresponds algebraically to the statement that the bulk state is \emph{separating} for both algebras ${\cal B}_{BH}$ and ${\cal B}_R$. We will denote the corresponding state by $|\Psi_H\rangle$. It can be thought of as the Hawking state in which the modes $b$ and $\tilde{b}$ are entangled and the information carrying $a$ modes are still inside the black hole. One can think of this state as imposing a vacuum at the QES $H$, i.e. roughly $\ket{\Psi_H}\approx \ket{\Phi}\otimes\ket{0_H}$, where $\ket{\Phi}$ is the microscopic state, and $\ket{0_H}$ is the vacuum near $H$ in the effective description.

A similar description holds after the evaporation step. We then have a separating state $|\Psi_I\rangle$ for two mutually commuting algebras denoted by ${\cal B}_{B}$ and ${\cal B}_{RH}$. The letter $I$ in the state $|\Psi_I\rangle$ indicates that the diary's quantum information has been transferred to the radiation algebra. The state might again be viewed as $\ket{\Psi_I} \approx \ket{\Phi}\otimes\ket{0_I}$.
The $b$-modes were already a part of ${\cal B}_R$ before the evaporation step, and hence afterward the $a$ modes are still contained in ${\cal B}_B$. The difference between the two situations lies in the entanglement structure of the state $|\Psi_I\rangle$ compared to $|\Psi_H\rangle$. While initially in the state $|\Psi_H\rangle$ the quantum information about the diary was contained in ${\cal B}_{BH}$, in the new state $|\Psi_I\rangle$ it has been transferred to the algebra ${\cal B}_{RH}$. 
We will assume that the following inclusions of algebras hold 
\beq
{\cal B}_{B}\subset {\cal B}_{BH}\qquad \mbox{and} \qquad {\cal B}_{R}\subset {\cal B}_{RH}~.
\eeq 
These encode that the black hole wedge shrinks and the radiation wedge grows during an evaporation step.

In addition to the $b$-modes, the algebra ${\cal B}_{RH}$ also contains the partner modes $\tilde{a}$ that are entangled with the operators $a$ in the state $|\Psi_I\rangle$. After the evaporation step the $a$ and $\tilde{a}$ modes are entangled in a unique ``frozen vacuum" and hence no longer carry any quantum information. In other words, during the evaporation process the information has been transferred from the black hole to the radiation via an entanglement swap that takes the state $|\Psi_H\rangle$ to the state $|\Psi_I \rangle$: 
\beq
S:\, |\Psi_H\rangle\  \to \ |\Psi_I\rangle~. 
\eeq

We will provide an explicit description of this entanglement swap and explain how the quantum information initially encoded in the $a$-modes is transferred from the black hole interior to the radiation modes $b$. The main algebraic ingredients in this process are the projectors onto the vacuum states at the two locations, $H$ and $I$. These projectors correspond algebraically to the so-called Jones projections, which we denote by $e_H$ and $e_I$. We will show that these projections can also be characterized in terms of the non-isometric map from the effective algebra and Hilbert space to the underlying microscopic theory, in which the information transfer is described by unitary evolution.

Another issue on which our algebraic model of the evaporation process sheds light is the notion of black hole complementarity. The basic idea of black hole complementarity \cite{susskind1993strechted} is that the horizon-crossing process can be coherently described in a basis adapted to exterior physics or in a basis appropriate to the infalling observer, without contradiction. In our algebraic model, there will likewise be two different bases that give complementary descriptions of where the diary information is encoded. The novelty of our approach is to relate black hole complementarity to computational complexity.

In the effective description, the quantum information contained in the diary remains inside the black hole interior and simply enters the island region. There is no complicated decoding required and the diary just moves through spacetime without encountering any apparent barriers. There is also
another complementary description in which the information is accessible from the outgoing radiation, where it is stored in a complex way. Our goal is to give a physical explanation for why both descriptions are mutually consistent and to describe their relation. We will find that both complementary descriptions are equally valid and are related by a transformation with very high complexity. To be precise, this works in the most simple and elegant way for algebra inclusions of ``depth two''. In this case, the relation between the $a$ and $b$ modes can be expressed as a so-called algebraic Fourier transform first discovered by Ocneanu \cite{ocneanu1991quantum}. As we shall see, the required operation is essentially the same process responsible for the entanglement swap and the shift of the QES.

The organization of the paper is as follows. In section~\ref{sec:Algebraic Model for Black Hole Evaporation}, we give an algebraic description of the black hole evaporation step from the microscopic perspective and introduce the operator algebras for the black hole and radiation degrees of freedom. We then relate the microscopic unitary evaporation step to an effective description, showing how (non-)isometric maps and an entanglement swap implement the dynamics of the black hole during evaporation. In section~\ref{sec:Information Transfer and Algebraic Complexity}, we show that implementing the entanglement swap has a computational complexity determined by the index of the algebra inclusion. In section~\ref{sec: The Island Algebra}, we develop an algebraic notion of an island region in our model, built from the Jones tower and the canonical shift, that contains the information released by the black hole. In section~\ref{sec: Information recovery and black hole complementarity}, we formulate an algebraic version of the information puzzle and show how information recovery is implemented by an algebraic Fourier transform on the island algebra, providing an operator algebraic realization of black hole complementarity in our model.

\section{An Algebraic Model for the Black Hole Evaporation Step}
\label{sec:Algebraic Model for Black Hole Evaporation}

In this section, we present our algebraic model of the black hole evaporation step from the \emph{microscopic} (or boundary) perspective. In this description, there is a priori no need to introduce entangled Hawking pairs, since these spacetime modes are in a unique “frozen” vacuum state and hence do not carry microscopic quantum information. The \emph{effective} (or bulk) description arises from the microscopic one through an isometry that creates these entangled pairs, while conversely the bulk theory is embedded back into the microscopic description in a non-isometric way. In the microscopic picture, the evaporation is represented by a unitary evolution of the black hole and radiation algebras, and we quantify how much information is transferred in an evaporation step and how it is encoded in these algebras.

\subsection{The microscopic black hole and radiation algebras}

The microscopic system is divided into the degrees of freedom associated with the black hole and the radiation. We assume that the radiation has permanently left the black hole and has been placed in a non-gravitational reservoir in a region that is no longer in causal contact with the black hole. We will consider the operator algebra of the radiation system that is required to purify the black hole state. Let us denote the microscopic operator algebra associated with this radiation by ${\cal A}_R$ and the microscopic algebra corresponding to the black hole by ${\cal A}_{BH}$ . We can identify ${\cal A}_R$ with the commutant algebra of ${\cal A}_{BH}$:
\beq 
\mathcal{A}_R=\{ y\in \mathcal{B}(\mathcal{H})\, |\, [y,x]=0~, \, \forall x\in \mathcal{A}_{BH} \}~.
\eeq
We also assume that the von Neumann algebras are \emph{factors}, which means that the center is trivial, i.e.
\begin{equation}
{\cal A}_{BH}\cap {\cal A}_R =\mathbb{C}1~.
\end{equation} 
We expect that our results can be extended to general inclusions with non-trivial center.

Microscopically, the evaporation process is described by unitary evolution and hence preserves the quantum information. 
After a finite evaporation step, some of the information is transferred from the black hole to the radiation. At an algebraic level, this means that the algebra of the radiation system has increased while the black hole algebra has decreased. One can thus expect the following inclusions of algebras
\beq
{\cal A}_B\subset {\cal A}_{BH}~,\qquad  \qquad {\cal A}_R\subset {\cal A}_{RH}~,
\eeq
where we have the commutant identifications
\beq
{\cal A}_R = {\cal A}_{BH}'~, \qquad \qquad {\cal A}_{RH} = {\cal A}_{B}'~.
\eeq
In writing these inclusions we have excluded the possibility that some of the information contained in the radiation system before the evaporation is being returned to the black hole system. In a physical setting in which the black hole is in contact with the radiation, there will generally be an exchange of quantum information in both directions. We will avoid this issue by defining the radiation algebra as being associated with the radiation degrees of freedom that have permanently left the black hole system. By definition, this algebra can only increase, and hence the above inclusion relations will be obeyed. 

We do not put restrictions on the type of von Neumann algebras for the black hole and radiation systems. Our discussion therefore includes the case where the operators acting on the black hole degrees of freedom form a properly infinite von Neumann algebra, thus extending the setting of \cite{vanderHeijden:2024tdk}.

The microscopic quantum state of the system is denoted by 
\begin{equation}
|\Phi\rangle \in \mathcal{H}~,
\end{equation}
where $\mathcal{H}$ is identified with the GNS Hilbert space of the black hole algebra ${\cal A}_{BH}$. Our discussion applies to general microscopic states, provided they are sufficiently entangled to allow the development of the relevant algebraic tools. We assume that before the evaporation step the state $|\Phi\rangle $ is \emph{cyclic} for the algebra ${\cal A}_{BH}$ (or equivalently \emph{separating} for ${\cal A}_{R}$). Cyclicity means that all states in the Hilbert space $\cal H$ can be obtained by acting with operators in ${\cal A}_{BH}$ on $|\Phi\rangle$, while the separating condition ensures that no non-zero operator in ${\cal A}_{R}$ annihilates $|\Phi\rangle$. Similarly, after the evaporation step the state $|\Phi\rangle$ will be assumed to be at least cyclic for ${\cal A}_{RH}$ (or separating for ${\cal A}_B$).

For type III algebras, we can even go a step further and assume that $|\Phi\rangle$ is cyclic {\it and} separating for both ${\cal A}_{BH}$ as well as ${\cal A}_{B}$ and their commutant algebras. It will be useful to discuss this type III case in more detail. We will make use of known results and methods developed in the mathematical literature, in particular by Longo \cite{longo1989index,longo1990index}. We will slightly rewrite and reformulate his methods and mathematical language to make them more adapted to the current setting\footnote{A detailed translation from our conventions to those of Longo will be presented in Appendix \ref{app:math_background}.}.
In the type III situation, and in other properly infinite cases, it is possible for an algebra inclusion to be unitarily equivalent to the original algebras. This means that we could have chosen to represent the evaporation process in the Schr\"{o}dinger picture by a unitary $U \in B(\mathcal{H})$ which acts on the state as $|\Phi\rangle \to U|\Phi\rangle$. In the Heisenberg picture, the evaporation step would then leave the state invariant, but change the algebras associated with the radiation and the black hole according to:
\beq \label{eq:algebraevolution}
{\cal A}_{B}\equiv U^\dagger {\cal A}_{BH}\, U~,  \qquad \qquad  \caRH \equiv  U^\dagger  {\cal A}_{R} \, U~.
\eeq
We emphasize that this description of the inclusions as a unitary equivalence is only valid for certain von Neumann algebras. The unitary $U$ will be a useful ingredient, but its existence will not be essential for the main part of our discussion.

\begin{figure}
    \centering    \includegraphics[width=0.9\linewidth]{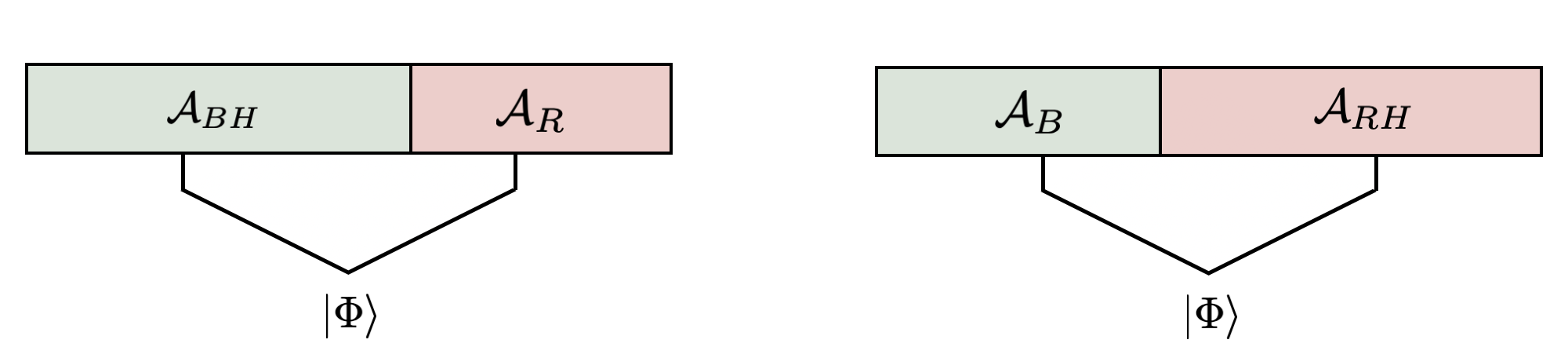}
    \caption{\small{A pictorial representation of the division of the full system into von Neumann subalgebras $\mathcal{A}_{BH},\mathcal{A}_R$ (left) and $\mathcal{A}_{B},\mathcal{A}_{RH}$ (right), corresponding to the situation before and after the discrete evaporation step, respectively. The global state $|\Phi\rangle$ determines the pattern of entanglement in the system. We assume that $|\Phi\rangle$ is at least cyclic for the algebra $\mathcal{A}_{BH}$ (or separating for $\mathcal{A}_R$) before the evaporation step, and cyclic for the algebra $\mathcal{A}_{RH}$ (or separating for $\mathcal{A}_B$) after the evaporation step.}
    }
    \label{fig:algebradivision}
\end{figure}

\subsection{The index as a measure of quantum information}

We are interested in the question of how a certain amount of information that has been thrown into the black hole is transferred to, and eventually retrieved from, the radiation after an evaporation step. To identify the information transferred from the black hole to the radiation, we make use of the concept of {\it conditional expectations}. In the present situation, we can define two separate conditional expectations, one for each algebra inclusion:
\beq 
E_I: {\cal A}_{BH}\to {\cal A}_B~, \qquad  \qquad  E'_H: {\cal A}_{RH} \to {\cal A}_R~.\label{condexp1}
\eeq
Here, the prime of $E'_H$ indicates that this conditional expectation is associated with the commutant of the black hole algebra. A conditional expectation is a linear map from an algebra to a subalgebra obeying various natural properties (see Appendix \ref{app:math_background} for more details). In a simplified setting in which the information is being carried by a finite-dimensional subsystem $I$, the conditional expectation $E_I$ is given by the (normalized) partial trace over this subsystem $I$ before the evaporation step, i.e.
\beq \label{eq:partialtrace}
E_I(x) = \frac{1}{d_I}\mathrm{Tr}_{I}(x)~, \qquad \mbox{with}\qquad x\in {\cal A}_{BH}~,
\eeq
where $d_I$ is the dimension of the subsystem $I$.
Similarly, $E'_H$ would in this case be given by the trace over a finite subsystem $H$ of the emitted Hawking radiation after the evaporation step. The algebraic generalization involves a unital map $E_I$ that leaves the subalgebra invariant, i.e. $E_I(x)=x$ for $x\in {\cal A}_B$. Physically, this coarse-graining operation preserves the information in the black hole that is not transferred to the radiation during this step.

A second important algebraic concept is that of a \emph{relative commutant}. The relative commutant of the black hole algebra before and after evaporation is an algebra we denote as ${\cal A}_I$, consisting of all operators in ${\cal A}_{BH}$ that commute with all the operators in ${\cal A}_B$. Another way to characterize the algebra ${\cal A}_I$ is as the set of operators that are common to ${\cal A}_{BH}$ and ${\cal A}_{RH}$. Hence, this same algebra ${\cal A}_I$ also represents the relative commutant of the radiation algebras ${\cal A}_{RH}$ and ${\cal A}_R$. We can thus define the algebra ${\cal A}_I$ as the \textit{first} relative commutant in two possible ways
\beq
{\cal A}_I\ \equiv\  {\cal A}_{BH}
\cap
 {\cal A}'_B~, 
 \qquad\quad \mbox{or equivalently}\qquad\quad   {\cal A}_I\ \equiv  
 \, {\cal A}_{RH} 
\cap {\cal A}_R'~.
 \label{AI}
\eeq

Intuitively, one might think that all the quantum information that is transferred from the black hole is contained in the algebra ${\cal A}_{I}$. This intuition is not correct in general. Since the algebra ${\cal A}_I$ commutes with ${\cal A}_B$ as well as ${\cal A}_R$, it does not change their embeddings in ${\cal A}_{BH}$ and ${\cal A}_{RH}$. This means that ${\cal A}_I$ generates a ``symmetry" of the common black hole and radiation system. In mathematical terms, this implies that the inclusions ${\cal A}_B\subset {\cal A}_{BH}$ and ${\cal A}_R\subset {\cal A}_{RH}$ are \textit{reducible}: one can exploit the spectrum of operators in ${\cal A}_I$ to project the algebras onto smaller irreducible components.

For \textit{irreducible} inclusions, the first relative commutant $\mathcal{A}_I$ is trivial and consists of (multiples of) the identity. Nevertheless, even in this irreducible case, a non-zero amount of quantum information is transferred from the black hole to the radiation. This already shows that the ``size" of the inclusion is not captured by the first relative commutant alone: the additional degrees of freedom live in higher relative commutants in what is known as the \emph{Jones tower}. In our black hole setting, this is precisely what happens in the irreducible depth two case: $\mathcal{A}_I$ is trivial, yet information is still transferred, and it is encoded in \textit{second} relative commutants that we will introduce later. In this paper, we will discuss both the reducible situation and an irreducible case.

The previous discussion makes clear that the amount of information released by the black hole is not determined by the size of the algebra ${\cal A}_I$. A more precise way to measure the transferred information is to consider the relative size of the algebras ${\cal A}_{BH}$ and ${\cal A}_B$. In mathematical terms, this quantity is known as the \emph{Jones index}. In the finite-dimensional setting, it is simply given by the ratio of the dimensions of the GNS Hilbert spaces associated with the algebra ${\cal A}_{BH}$ and its subalgebra ${\cal A}_{B}$. Jones introduced its infinite-dimensional analogue for type II$_1$ von Neumann factors \cite{jones1983index}, which was subsequently generalized to arbitrary von Neumann algebras by Kosaki \cite{kosaki1991index} and Longo \cite{longo1989index}. The Jones index will be denoted by\footnote{For reducible inclusions there can be more than one conditional expectation. In this work, we assume that we are dealing with the ``minimal" conditional expectation (i.e. with minimal value of the index). For irreducible inclusions, the conditional expectation is always the minimal one.}
\beq
{\rm Ind}(E_I) =\left\lbrack {\cal A}_{BH}:{\cal A}_B\right]~, \qquad   \qquad {\rm Ind}(E'_H) =\left\lbrack {\cal A}_{RH}:{\cal A}_R\right]~. 
\eeq
By thinking of the conditional expectation 
as a partial trace, the index can intuitively be interpreted as the square of the dimension of the Hilbert space associated to the emitted Hawking radiation. More generally, it is given in terms of the quantum (or statistical) dimension $d_I$ via 
\beq
{\rm Ind}(E'_H) = {\rm Ind}(E_I) \ \equiv \  d_I^2~.
\eeq
As we are interested in a physical situation in which a finite amount of information is extracted from the black hole, we assume the index is finite.\footnote{Note that many physically relevant inclusions do not admit a finite-index conditional expectation. We comment on the applicability of our protocol to infinite-index inclusions in the Discussion section.} 

We restrict our discussion to the physical situations in which the released quantum information is carried by a relative commutant. In these special cases, the index is an integer and can be thought of as the number of operators $a_i$ that are added to ${\cal A}_B$ so that all elements of ${\cal A}_{BH}$ can be written as linear combinations of $a_i$ with coefficients in ${\cal A}_B$. The operators $a_i$ can be chosen so that they form an orthonormal basis with respect to the conditional expectation $E_I$: 
\beq 
E_I(a_i^\dagger a_j) =\delta_{ij}~.
\eeq 
We can thus decompose any operator $x\in{\cal A}_{BH}$ as
\beq 
x = \sum_i a_i x_i~,  
\qquad \quad \mbox{where }\qquad x_i= E_I(a^\dagger_i x)~,
\eeq
are elements of ${\cal A}_{B}$. In the mathematics literature, the set of operators $a_i$ is called a \textit{Pimsner--Popa basis} \cite{pimsner1986entropy} (we refer to Appendix \ref{app:math_background} for more details). A similar basis can be defined for the inclusion of radiation algebras. The corresponding set of operators $b_i \in {\cal A}_{RH}$ can be chosen to be orthonormal with respect to the conditional expectation $ E'_H$: 
\beq 
E'_H(b_i^\dagger b_j) =\delta_{ij}~,
\eeq 
and can thus be used to decompose any operator $y\in {\cal A}_{RH}$ as 
\beq  y = \sum_i b_i y_i~, \  
\qquad 
\mbox{where }\qquad y_i = E'_H(b^\dagger_i y)~, \\[-3mm]
\eeq
are elements of ${\cal A}_{R}$.

Our goal is to explain how information initially carried by the operators $a_i \in{\cal A}_{BH} $ is transferred to the operators $b_i\in{\cal A}_{RH} $ by the evaporation step. The operators in a Pimsner--Popa basis do not, in general, form an algebra: multiplying two basis elements does not necessarily produce a linear combination of these basis elements. To circumvent this issue and make our discussion physically more intuitive, we will assume that the operators $a_i$ and $b_i$ are contained in a relative commutant. The simplest situation occurs when the algebra inclusion is fully reducible and the operators $a_i$ and $b_i$ are all contained in the \emph{first} relative commutant ${\cal A}_I$. We will also study the information transfer for an irreducible inclusion for which the operators $a_i$ and $b_i$ are contained in the \emph{second} relative commutant. This terminology and its physical interpretation will be explained further below. 

\subsection{Spacetime vacuum states and projectors}
\label{Spacetime vacuum states and projectors}

An important question is how the information transfer from the black hole to the radiation can be explained in terms of the Hawking process. To address this question, we introduce an additional set of algebras associated with the effective description of the evaporation process.

In the microscopic theory, the evaporation process is represented by a map that acts on the black hole and radiation algebras while leaving the underlying quantum state $|\Phi\rangle$ unchanged,
\beq 
 \bigl(\, \mathcal{A}_{BH} \times \mathcal{A}_R,\  |\Phi\rangle\,  \bigr)\qquad  \longrightarrow \qquad   \bigl(\, \mathcal{A}_B \times \mathcal{A}_{RH},\  |\Phi\rangle\, \bigr)~. 
\eeq 
In the spacetime description, the situation is analogous but with an important difference: in addition to the algebras, the quantum state also changes during the evaporation step,
\beq 
 \bigl(\, \mathcal{B}_{BH} \times \mathcal{B}_R,\  |\Psi_H\rangle\,  \bigr)\qquad  \longrightarrow  \qquad   \bigl(\, \mathcal{B}_B \times \mathcal{B}_{RH},\  |\Psi_I\rangle\, \bigr)~.
\eeq 
As we will explain, this change of state reflects a shift of the minimal quantum extremal surface from position $H$ to position $I$.

Our description of the effective algebra involves particular (non-)isometric maps associated to the conditional expectations $E_I$ and $E'_H$. We denote the isometric maps by $V_I$ and $V_H$. The Hermitian conjugates $V^\dagger_I$ and $V_H^\dagger$ describe their non-isometric counterparts. 
The physical origin of these (non-)isometric maps can be understood from the nature of the vacuum state near the horizon.

According to Hawking's argument, the evaporation process enlarges not only the radiation algebra, but also the black hole algebra, since after the emission it must include the interior partners of the emitted radiation. To describe this process, we make use of the isometry $V_H$ that maps the microscopic global state $|\Phi\rangle$ to a new state
\begin{equation}    \left|\Psi_H\right\rangle=V_H|\Phi\rangle~,\label{PsiH}
\end{equation}
which includes extra ancillary degrees of freedom that we interpret as entangled Hawking pairs across the black hole horizon. Intuitively, one can think of the isometry $V_H$ as adding the radiation modes $b_i$ as well as their Hawking partners $\tilde{b}_i$ in the ``frozen vacuum" state at the black hole horizon (see Figure \ref{fig:V_H}). The definition of the isometry $V_H$ and its conjugate $V_H^\dagger$ is directly related to that of the conditional expectation $E'_H$ \cite{Hollands_2021,longo1990index}. That is, the latter can be expressed through the relation 
\beq
\label{EH'}
E'_H(y) = V_H^\dagger y V^{}_H~, \qquad  y \in {\cal A}_{RH}~.
\eeq
It can be straightforwardly shown that this map satisfies the properties of a conditional expectation (see e.g. Appendix \ref{app:math_background}). In particular, it takes operators from the radiation algebra ${\cal A}_{RH}$ and maps them onto an operator in the subalgebra ${\cal A}_{R}$, while leaving operators that were already in ${\cal A}_R$ invariant.

\begin{figure}[t]

\centering
\includegraphics[width=0.55\linewidth]{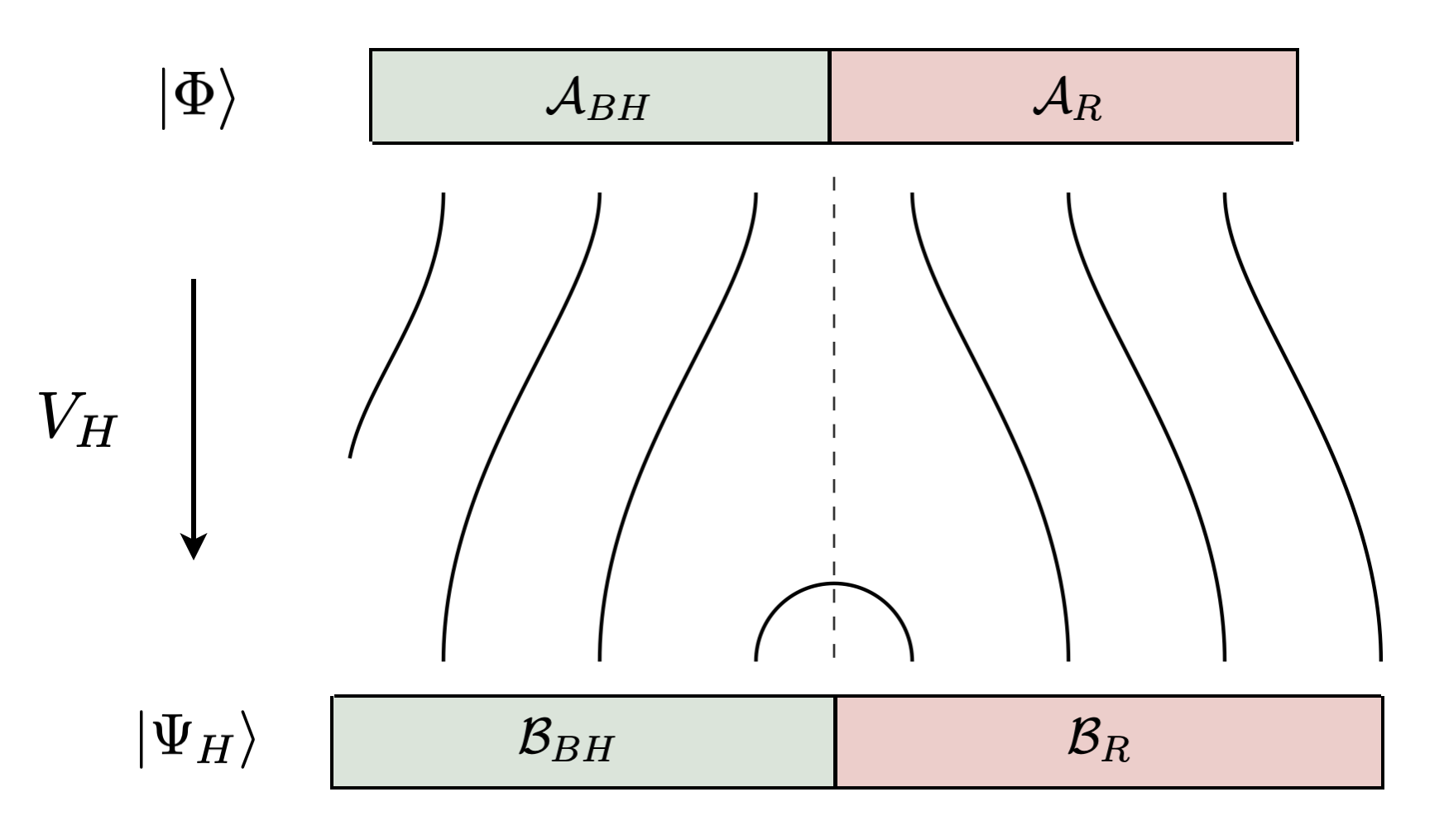}
        \caption{\small{The isometry $V_H$. It maps the state $|\Phi\rangle$ to $|\Psi_H\rangle$, thereby adding an entangled Hawking pair to the system. The lines represent the different degrees of freedom, grouped according to the subalgebras acting on them. The striped line denotes the decomposition into subalgebras and defines the algebraic analogue of a QES at position $H$.} 
   }
    \label{fig:V_H}
\end{figure}

The state $|\Psi_H\rangle$ describes the quantum state before the information is released from the black hole. To describe the state after the information is released to the radiation, we need another isometry that is defined in a similar way but now for the inclusion of the black hole algebras. 
Similarly, for the radiation we can introduce an isometry $V_I$ for the black hole algebra by writing
\beq \label{EI}
E_I(x) = V_I^\dagger x V^{}_I~, 
\qquad  x\in {\cal A}_{BH}~.
\eeq
By acting with this isometry $V_I$ on the underlying microscopic state $|\Phi\rangle$ we create a different state $|\Psi_I\rangle$ given by 
\beq
|\Psi_I\rangle = V_I |\Phi\rangle ~.
\label{PsiI}
\eeq 
 In this new state $|\Psi_I\rangle$, the modes $a_i$ that initially carried the information about Alice's diary are put in a maximally entangled state with partner modes $\tilde{a}_i$ contained in the radiation algebra. As we will explain in detail below, 
this also means that in the state $|\Psi_I\rangle$ the information content of Alice's diary has been transferred to the radiation system. The black hole information puzzle can thus be stated as the problem of transferring the original Hawking state $|\Psi_H\rangle$ into the new state $|\Psi_I\rangle$.

The isometries satisfy
\beq
\label{isomdef}
V_I^{\dagger}V_I\ =\ {\mathbf 1}~,\quad \qquad\mbox{and}\quad  \qquad  V_H^{\dagger}V_H\ =\ {\mathbf 1}~,
\eeq
which ensures that the conditional expectations are unital maps. Crucially, the conjugate maps $V_H^\dagger$ and $V_I^\dagger$ are non-isometric. Together with their corresponding isometries, they give rise to the so-called \emph{Jones projections} $e_H$ and $e_I$ associated with the above inclusions:
\beq \label{eq:Jonesprojections}
V^{}_H V_H^\dagger = e^{}_H~, \quad \qquad\mbox{and}\quad  \qquad  V^{}_I V_I^\dagger = e^{}_I~.
\eeq
In the effective description, these projectors naturally acquire a physical interpretation as projecting on the entangled vacuum states $|\Psi_H\rangle$ and $|\Psi_I\rangle$ of the newly created Hawking modes, respectively, at different stages of the evaporation process. Indeed, one easily verifies that these states obey 
\begin{equation}
    \ e_H |\Psi_H\rangle = |\Psi_H\rangle~, \qquad \text{and} \qquad \ e_I |\Psi_I\rangle = |\Psi_I\rangle~. \label{constraints}
\end{equation}

We will explain in the next section that these projectors are conceptually reminiscent of the final state projection proposed by Horowitz and Maldacena \cite{horowitz2004blackhole}.
The isometries $V_I$ and $V_H$ and their conjugates $V^\dagger_I$ and $V^\dagger_H$ represent intertwining operators between the microscopic and effective description. The effective algebras are viewed as extensions of the microscopic algebras that include the projections $e_H$ and $e_I$, both of which are not part of the microscopic algebras.

So far, all ingredients in our algebraic model are defined in a meaningful way for all types of von Neumann algebras, provided we have a finite index conditional expectation.\footnote{We should note that for type III algebras it is in principle possible to represent the effective and microscopic theory in the same Hilbert space. In this alternative approach, one treats the projectors and (non-)isometries to be part of the original microscopic algebra.}

\subsection{Non-isometric encoding}

In this subsection, we describe the map from the effective to the microscopic description. This map is implemented by non-isometric encoding from the effective algebra into the microscopic algebra, together with the corresponding action on states. It must satisfy the condition that expectation values of effective operators in the states $|\Psi_H\rangle$ and $|\Psi_I\rangle$ coincide with those of the corresponding microscopic operators in the underlying state $|\Phi\rangle$. To construct the effective-to-microscopic map (i.e. the analogue of the holographic bulk-to-boundary map), we use the isometries $V_H$ and $V_I$ and their non-isometric conjugates $V_H^\dagger$ and $V_I^\dagger$. Our construction is closely related to the non-isometric quantum error-correcting codes proposed in recent discussions of the black hole information paradox \cite{akers2022black}, with the discussion here adapted to infinite-dimensional systems.

We define the (non)-isometric maps by making use of the (minimal) conditional expectations associated to the inclusions of the effective algebras before and after evaporation. These will be denoted by
\begin{equation}
    \begin{aligned}
& E_H: \cbBH \to \cbB~, \qquad \mbox{and}\qquad 
& E_I^\prime: \cbRH\to \cbR~. \label{condexp2}
\end{aligned}
\end{equation}
These are the analogues of $E_I, E_{H}^{\prime}$ in \eqref{condexp1}, but now defined on the spacetime algebras.
We assume that these conditional expectations have the same index $d_I$ and can be expressed in a similar way as in (\ref{EH'}) and (\ref{EI}) in terms of the (non-)isometries $V_H$, $V_I$ and their conjugates. The corresponding Jones projections are denoted by $e_H$ and $e_I$ and are expressed in terms of the (non-)isometries via the relations (\ref{eq:Jonesprojections}). At this point, these projections, conditional expectations and isometries are defined from the effective perspective and are not yet related to the microscopic description. Our goal is to derive these relations. 

Before the evaporation step the bulk-to-boundary map is given by
\beq
    V^\dagger_H: \   \bigl(\, \cbBH \times \cbR,\  |\Psi_H\rangle\,  \bigr)  \quad  \longrightarrow \quad   \bigl(\, \caBH \times \caR,\  |\Phi\rangle\, \bigr)~.
\label{bulk-to-boundary}
\eeq
The action of $V^\dagger_H$ on the black hole and radiation algebras can be written out explicitly and leads to the following identifications of the microscopic algebras 
\beq
   \caBH = V^\dagger_H   \,\cbBH  V_H~,       \qquad\mbox{and} \qquad            \caR = V^\dagger_H \,\cbR V_H~. \label{b-to-b-before}
\eeq
We also find that the microscopic quantum state is recovered from the bulk state by 
\beq
|\Phi\rangle = V^\dagger_H |\Psi_H\rangle~. 
\eeq
Note that this last equation is equivalent to equation (\ref{PsiH}) together with the constraint in (\ref{isomdef}).

To verify that (\ref{bulk-to-boundary}) satisfies the requirements of a bulk-to-boundary map, let us compare the expectation values of an arbitrary operator ${ A}\in {\cal A}_{BH}\times{\cal A}_R$ in the microscopic algebra and of its corresponding operator $B \in {\cal B}_{BH}\times {\cal B}_R$ in the effective algebra. It is straightforward to show that these expectation values agree 
\begin{equation}
\label{proof}
\langle \Phi | A | \Phi \rangle
= \langle \Psi_H | V_H \bigl(V_H^\dagger B V_H\bigr) V_H^\dagger | \Psi_H \rangle
= \langle \Psi_H | e_H B e_H | \Psi_H \rangle
= \langle \Psi_H | B | \Psi_H \rangle ~,
\end{equation}
where in the second and third step we made use of the identities (\ref{eq:Jonesprojections}) and (\ref{constraints}). 

The bulk-to-boundary map after the evaporation step proceeds completely analogously, except now we make use of the non-isometry $V_I^\dagger$ to map the spacetime algebras and quantum state to their corresponding quantities in the microscopic description 
\beq
    V^\dagger_I: \   \bigl(\, \cbB \times \cbRH,\  |\Psi_I\rangle\,  \bigr)  \quad  \longrightarrow \quad   \bigl(\, \caB \times \caRH,\  |\Phi\rangle\, \bigr)~.
\label{bulk-to-boundary2}
\eeq
This leads to analogous relations to the ones given in (\ref{b-to-b-before}) but now applied to the microscopic and effective algebras after evaporation 
\beq
   \caB = V^\dagger_I   \,\cbB  V_I~,       \qquad\mbox{and} \qquad            \caRH = V^\dagger_I \,\cbRH V_I~.
\eeq
The non-isometry $V^\dagger_I$ takes the spacetime quantum state to the same underlying microscopic quantum state
\beq
|\Phi\rangle =  V^\dagger_I |\Psi_I\rangle~. 
\eeq
By following the same steps as in (\ref{proof}) one shows that the combined action of the non-isometric map $V^\dagger_I$ on the state and the algebra after evaporation has again the required property to represent a proper bulk-to-boundary map in the sense that it ensures that the expectation values in the microscopic theory agree with those in the effective description. 

\begin{figure}[t!]
    \centering
    \includegraphics[width=0.55\linewidth]{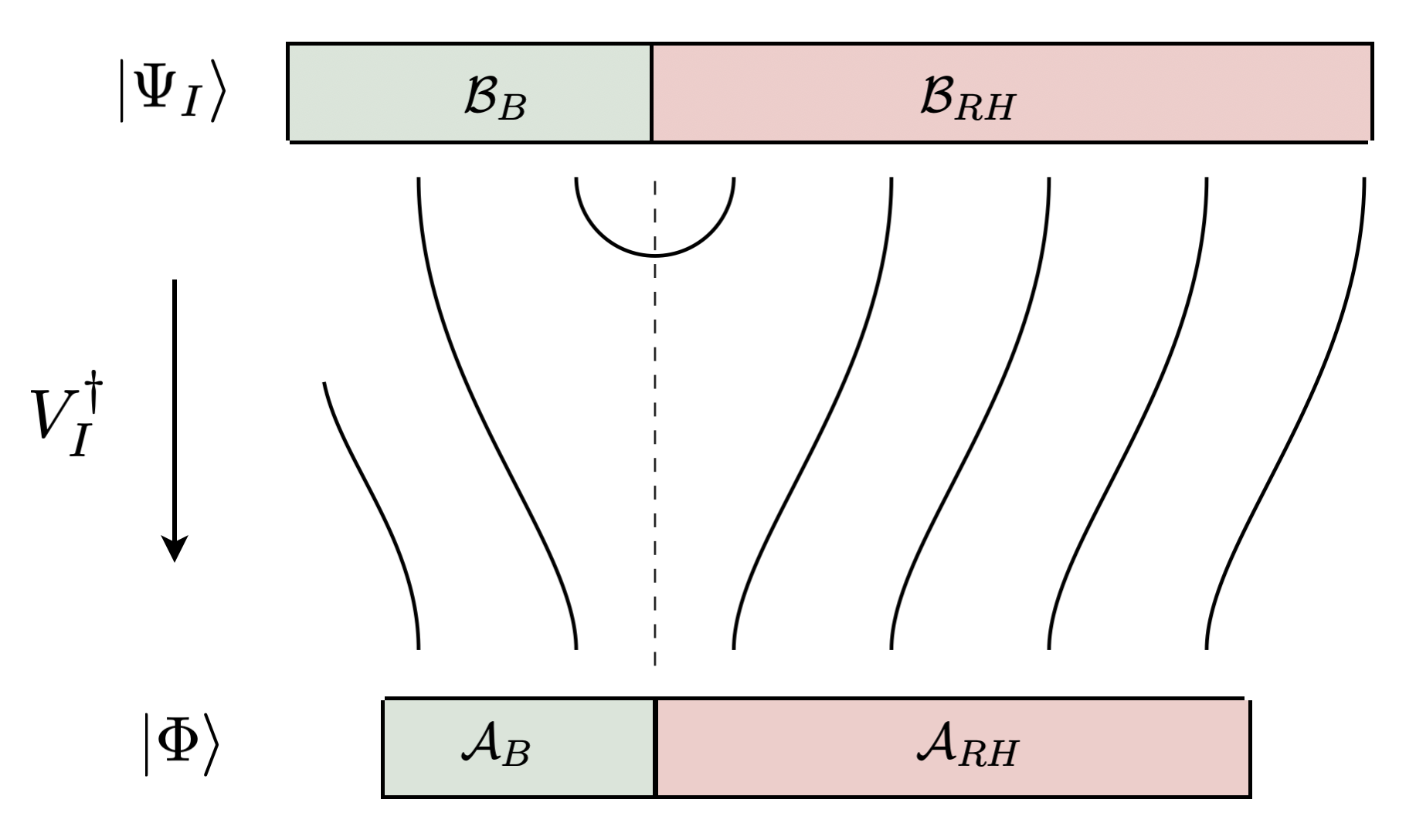}
    \caption{\small{The adjoint map $V_I^{\dagger}$. It maps the state $|\Psi_I\rangle$ to $|\Phi\rangle$, thereby removing an entangled Hawking pair from the system. The lines represent the different degrees of freedom, grouped according to the subalgebras acting on them. The striped line denotes the decomposition into subalgebras and defines the algebraic analogue of a QES at position $I$.}
    }
    \label{fig:V_I}
\end{figure}

An important feature of the evaporation step is that the ``location'' of the QES is different before and after \cite{engelhardt2015quantum,almheiri2019entropy}. In our model, the location of the QES is determined algebraically by requiring that the effective quantum state is cyclic and separating for the black hole and radiation algebras. Recall that  for the type III case the microscopic state $|\Phi\rangle$ can be chosen to be cyclic and separating both before as well as after the evaporation step. It follows that the state $|\Psi_H\rangle$ is cyclic and separating for ${\cal B}_{BH}\times{\cal B}_R$, while $|\Psi_I\rangle$ is cyclic and separating for ${\cal B}_{B}\times{\cal B}_{RH}$. The converse statements are not true.
This follows from (\ref{constraints}), since the states $|\Psi_H\rangle$ and $|\Psi_I\rangle$ 
are annihilated by $(e_H-1)$ and $(e_I-1)$ respectively. The Jones projector $e_H$ is an element of the radiation algebra $\cbRH$, since this operator commutes with the black hole algebra $\cbB$. Similarly, $e_I$ is contained in $\cbBH$ since it commutes with $\cal{B}_\mathrm{R}$. We conclude that the state $|\Psi_H\rangle$ is not separating for ${\cal B}_{RH}$, since its GNS Hilbert space contains null states, and hence it is not cyclic for ${\cal B}_B$. Similarly, the state $|\Psi_I\rangle$ is not separating for ${\cal B}_{BH}$ and not cyclic for ${\cal B}_R$.

\subsection{Effective decomposition of the microscopic evolution}
\label{Bulk decomposition of the microscopic evolution}

For the type III (or properly infinite) case, the microscopic algebras before and after evaporation are related by the unitary $U$ via (\ref{eq:algebraevolution}). In the effective description, the evaporation process is represented in an analogous way: one finds that the effective (spacetime) algebras are mapped onto each other by a unitary $\overline{U}$ that is related to, but distinct from, the microscopic operator $U$:
\beq
  \cbB = \overline{U}^{\dagger} \cbBH \overline{U}~, \qquad\quad \mbox{and} \qquad \quad 
\cbRH = \overline{U}^{\dagger}\cbR \overline{U}~.
\eeq

\begin{figure}[ht]
    \centering
    \includegraphics[width=0.5\linewidth]{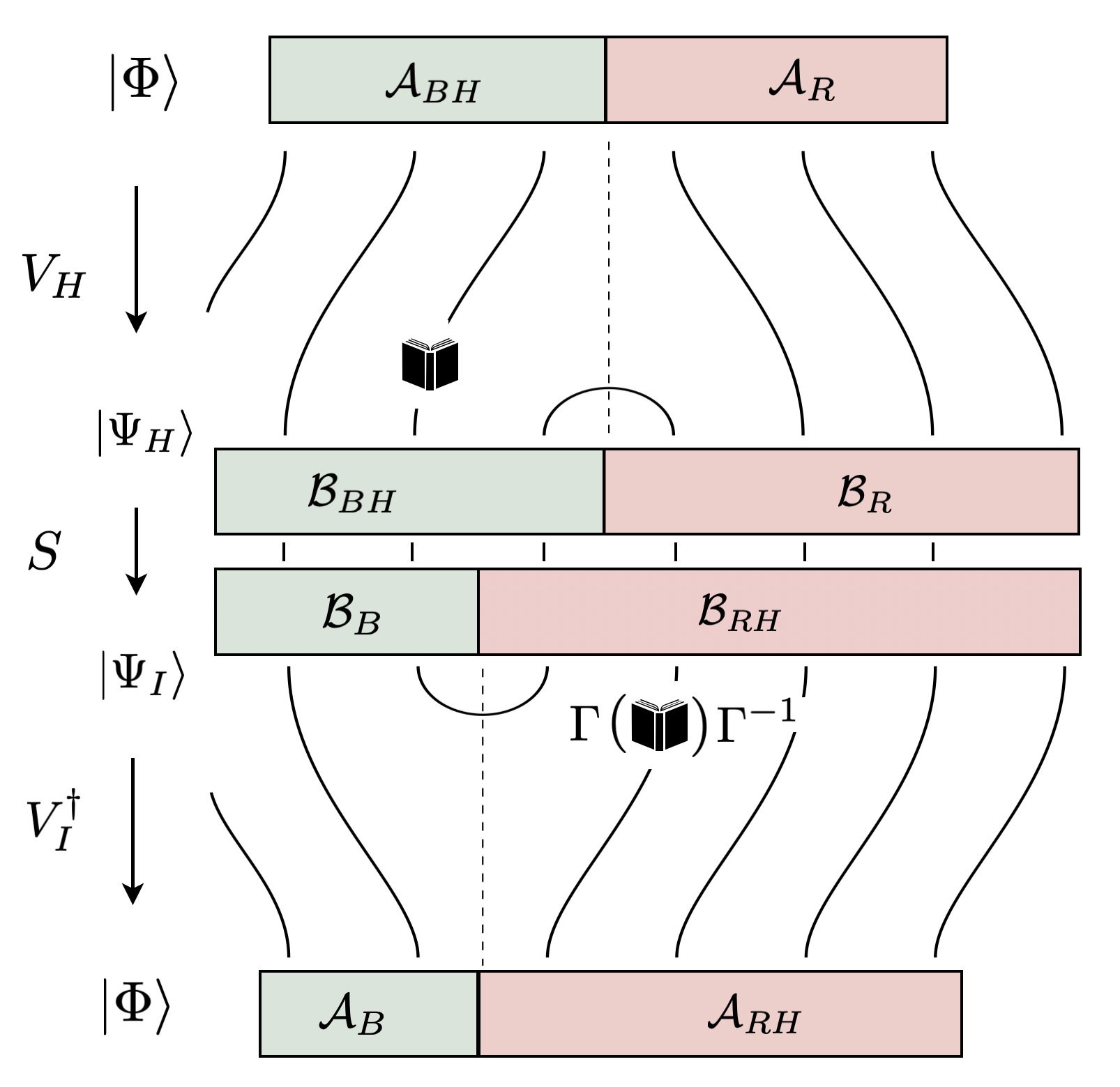}
    \caption{\small{Information transfer. The evaporation step consists of three steps: an isometry $V_H$, the entanglement swap $S$ and an application of the map $V_I^{\dagger}$. The  entanglement structure in the states $|\Psi_H\rangle$ and $|\Psi_I\rangle$ can be exploited to transfer operators from the interior algebra $\mathcal{A}_{BH}$ to the radiation algebra $\mathcal{A}_{RH}$. We have depicted the \emph{reducible} situation, where the diary information is contained in the first relative commutant $\mathcal{A}_H$. In that case, the corresponding radiation operators are simply related to the diary operators by an application of the canonical shift, i.e. $\hat{a}_i=\Gamma a_i\Gamma^{-1}$. At the same time, the entanglement swap amounts to a shift of the QES from the position $H$ to $I$.}
    }
\label{fig:informationtransfer}
\end{figure}

The unitary $\overline{U}$ plays a central role in Longo's approach to the theory of inclusions of type III algebras. Rather than relying directly on its definition in terms of modular conjugation $J^\Phi_{BH}$ of the microscopic theory via the identity $\overline{U}=J_{BH}^{\Phi}U^\dagger J_{BH}^{\Phi}$ (see Appendix \ref{app:math_background}), it is more convenient for our purposes to characterize the relationship between $\overline{U}$ and $U$ using the isometries $V_I$ and $V_H$. These operators implement the intertwiners between the microscopic and spacetime descriptions before and after an evaporation step. They imply that the spacetime evolution step $\overline{U}$ and the microscopic evolution operator $U$ obey the intertwining relation
\beq
    \overline{U} V_I^\dagger = V_H^\dagger U~. \label{Ubar}
\eeq
This relation ensures that the effective description of the evaporation process is compatible with the microscopic evolution step. In particular, it is equivalent to the statement\footnote{One has to remember in the order of applying these operators, that we are thinking in the Heisenberg picture, and thus applying these operators on the algebra, $U^\dagger a U = V_I \bar{U}^\dagger V^\dagger_H a V_H \bar{U}V^\dagger_I$, and in the same way, the relation \eqref{Ubar} states that applying the spacetime evaporation step and then mapping to the microscopic system is the same as first mapping down to the microscopic system and then applying the microscopic evaporation step.}
$U = V_H \overline{U} V_I^\dagger$, i.e.~that the microscopic time step $U$ can be realized by mapping to the effective description via $V_H$, applying the evaporation step $\overline{U}$, and then projecting back to the microscopic description via $V_I^\dagger$.

Therefore, in the effective description, the microscopic evaporation step generated by the unitary $U$ is decomposed into the following three steps:
\begin{enumerate}
    \item[1)] \textit{Micro-to-effective mapping}: An isometry $V_H$ representing the Hawking process, which maps the microscopic Hilbert space $\mathcal{H}$ and algebras to the effective description, and enlarges both the black hole and radiation algebras by adding additional entangled degrees of freedom.

    \item[2)] \textit{Information release}: An entanglement swap $S$ that maps the state $|\Psi_H\rangle$, representing the vacuum state at $H$, onto the new vacuum state $|\Psi_I\rangle$ at the location $I$, together with the spacetime evolution step $\overline{U}$ that shifts the location of the QES.

    \item[3)] \textit{Effective-to-micro mapping}: A non-isometric map $V_I^\dagger$ that removes the auxiliary degrees of freedom introduced in the effective description and returns the system to its underlying microscopic description, i.e.~the microscopic Hilbert space $\mathcal{H}$.
\end{enumerate}
The equivalence between the microscopic unitary $U$ and the composition of these three steps is summarized in the diagram below:

\begin{equation}
    \begin{array}{cccc}

    \text{Effective:}\quad  &\cbBH\times \cbR  &  \xlongrightarrow[\phantom{text to stretch the arrow}]{\text{\normalsize$\overline{U}$}} & \ \cbB\times \mathcal{B}_{RH} \\
   \qquad  &|\Psi_H\rangle  &  \xlongrightarrow[\phantom{text to stretch the arrow}]{{\text{\normalsize$S$}}} & \ |\Psi_I\rangle \\
    \qquad & 

\rotatebox{90}{$\xrightarrow{\phantom{long}{\rotatebox{270}{$V_H$}}\phantom{long}}$} \qquad 
 &  & \rotatebox{90}{$\xleftarrow{\phantom{long}{\rotatebox{270}{$V_I^\dagger$}}\phantom{long}}$}  \qquad   \\
   \text{Microscopic:} \quad & |\Phi \rangle &  & |\Phi \rangle \\
   \qquad &  \caBH\times \caR  &  \xlongrightarrow[\phantom{text to stretch the arrow}]{\text{\normalsize${U}$}} & \ \caB\times \caRH \\
   
    \end{array}
    \label{SummaryDiagram}
\end{equation}

In the next section, we describe the swap operator $S$ in detail and show that it can be implemented either by a suitably rescaled Jones projection or by an algorithm akin to Grover-search. The Jones projection implementation amounts to applying the projector $e_I$ and is represented diagrammatically in Figure~\ref{fig:informationtransfer}. This representation makes clear that the entanglement swap can also be viewed as a quantum teleportation protocol, a perspective we will elaborate on below.

\section{Algebraic Complexity of the Entanglement Swap}
\label{sec:Information Transfer and Algebraic Complexity}

In the effective description, the information transfer during the evaporation step is accomplished by an entanglement swap that changes the structure of the vacuum state, leading to a shift of the QES. In this section, we show that this entanglement swap admits two algebraic realizations of existing qubit protocols: a probabilistic projector-based implementation and a deterministic Grover-type implementation. Our main result is that the Jones index controls the success probability in the former and the computational complexity in the latter.

\subsection{A probabilistic entanglement swap }

In our model, information is released to the radiation through an entanglement swap, represented by an operator $S$ 
\beq 
|\Psi_{H}\rangle \overset{S}{\longrightarrow}  |\Psi_{I}\rangle~.
\eeq
which maps the state $|\Psi_H\rangle$ onto the state $|\Psi_I\rangle$. Two prescriptions can be given for $S$, analogous to the two implementations of the Hayden--Preskill protocol according to \cite{yoshida2017efficient}. Both protocols make use of the Jones projections $e_I$
and $e_H$. The first protocol is probabilistic, while the second version is deterministic but highly complex.

It is a famous result due to Jones that the projectors $e_I$ and $e_H$ satisfy the \emph{Temperley--Lieb algebra} relations \cite{jones1983index}:
\beq
e^{}_H e^{}_I e^{}_H = \frac{1}{ d_I^2}  e^{}_H~, \qquad e^{}_I e^{}_H e^{}_I = \frac{1}{d_I^2} e^{}_I~,
\label{TL2}
\eeq 
where $d_I=\sqrt{\Ind\, E_I}$ is the square root of the Jones index. For our purpose, it will be more convenient to reformulate the Temperley--Lieb algebra as a set of identities for the isometries $V_H$ and $V_I$, which are related to the Jones projections via (\ref{eq:Jonesprojections}). It is easily verified that the identities \cite{Longo:1990zp}
\beq 
\label{isom-TL}
V_I^\dagger V_H = \frac{1}{d_I} {\mathbf 1}~, \qquad 
\qquad\qquad  V_H^\dagger V_I = \frac{1}{d_I} {\mathbf 1}~,
\eeq 
imply and thus are equivalent to the Temperley--Lieb algebra relations. 

The equations (\ref{isom-TL}) will be used to describe both the probabilistic implementation and the Grover-search implementation of the entanglement swap. We find that the probabilistic nature of the former and the computational complexity of the latter are closely related. Both are determined by the value of the inner product between the states $|\Psi_H\rangle$ and $|\Psi_I\rangle$. Using the identities (\ref{isom-TL}) one finds that 
\beq
\langle\Psi_I|\Psi_H\rangle = \langle\Phi|V_I^\dagger V_H |\Phi\rangle = \frac{1}{d_I}~.
\label{overlap}
\eeq 
We are interested in a situation in which the Jones index is large, so that this overlap of the states is small. We find it convenient to re-express the right-hand side in terms of a small angle $\theta$ as
\begin{equation}
\sin{\frac{1}{2}\theta} 
\equiv \frac{1}{d_I}~.
\label{thetadef}
\end{equation}

The simplest way to implement the entanglement swap is by making use of a single Jones projection. In this procedure, one maps $|\Psi_{H}\rangle$ to $|\Psi_I\rangle$ by acting with the projector $e_I$. This algorithm is probabilistic since it does not preserve the norm of the state. 
The probability of successfully projecting onto the desired state is determined by the overlap between the initial and final states. From the result (\ref{overlap}) we thus deduce 
that the probability of success is $1/d_I^2$, precisely the inverse of the index. To obtain a normalized state, the entanglement swap operator $S_I$ is given by the Jones projection $e_I$ rescaled by the factor $d_I$. One easily verifies that
\beq
\label{SI}
S_I: |\Psi_H\rangle \to |\Psi_I\rangle~,  \qquad \mbox{with} \qquad S_I \equiv d_I e_I~. 
\eeq 
Since the projector is part of the extended algebra $\cbBH$ in the effective description, we interpret this process as intrinsic to the black hole dynamics rather than as an action performed by an external, powerful observer. It involves post-selection onto the correct maximally entangled state and effectively implements a ``perfect'' quantum teleportation protocol that requires no additional correcting unitary to recover the state. This post-selection is reminiscent of the ``final state projection'' of \cite{horowitz2004blackhole} at the black hole singularity, but in our model it occurs after a single evaporation step rather than only at the end of the evaporation process.

The probabilistic algorithm is not unitary, but nevertheless it does preserve all the quantum information that is contained in the original state. Indeed, there exists an ``inverse" algorithm that maps the image of $e_I$ back to that of the projector $e_H$. This inverse algorithm consists of the map
\beq
\label{SH}
S_H: |\Psi_I\rangle \to |\Psi_H\rangle~,  \qquad \mbox{with} \qquad S_H \equiv d_I e_H~. 
\eeq 
The first of the Temperley--Lieb relations (\ref{TL2}) can be interpreted as the condition that the composition $S_H S_I$ is equal to the identity when restricted to the image of $e_H$. Similarly, the second relation in (\ref{TL2}) implies that $S_I S_H$ gives the identity on the image of the Jones projector $e_I$.

\subsection{A deterministic entanglement swap}

In the second procedure, a Grover-search algorithm is applied to amplify the probability of obtaining $|\Psi_I\rangle$ without performing any projection. However, this comes at a cost: the low probability of success in the probabilistic procedure translates into a high computational complexity for implementing the deterministic algorithm.

The search algorithm admits a well-known simple geometric interpretation, illustrated in Figure~\ref{Fig:GroverS}. Consider the plane corresponding to the subspace spanned by the states $|\Psi_I\rangle$ and $|\Psi_H\rangle$, where the vertical axis corresponds to $|\Psi_I\rangle$ and the horizontal axis corresponds to its orthogonal complement $|\Psi_I\rangle_{\perp}$ within the subspace. The initial state $|\Psi_H\rangle$ lies at a small angle equal to $\theta/2$ from $|\Psi_I\rangle_{\perp}$ defined in (\ref{thetadef}). 

Grover's algorithm consists of a unitary sequence that progressively rotates $|\Psi_H\rangle$ within this plane toward $|\Psi_I\rangle$. Each iteration is implemented as a composition of two reflections across the axes defined by $|\Psi_{I}\rangle_{\perp}$ and $|\Psi_{H}\rangle$, respectively, 
\beq 
\mathcal{R(\theta)}\equiv \left(2e_H - {1}\right)\left({1} - 2e_I\right)~. 
\eeq 
Since $(1- 2e)^2 = \mathbf{1}$ for any projector $e$, it is easy to verify that $\mathcal{R}(\theta)$ is unitary:
\beq 
\mathcal{R}(\theta)\,\mathcal{R}^{\dagger}(\theta) = (2e_H-1)(1-2e_I)(1-2e_I)(2e_H-1) = \mathbf{1}~.
\eeq
As our notation suggests, the unitary $\mathcal{R}(\theta)$ corresponds to a rotation of the state by an angle $\theta$. Let us now show this explicitly, and along the way determine the number of times that we need to act with $\mathcal{R}(\theta)$ to map the state $|\Psi_H\rangle$ onto the state $|\Psi_I\rangle$. First, we express the operator $\mathcal{R}(\theta)$ in terms of the isometries $V_H$ and $V_I$ as
 \beq 
 {\cal R}(\theta) = \Bigl( 2V_H V_H^\dagger -1\Bigr)\Bigl(1- 2V_I V_I^\dagger\Bigr)~. 
 \eeq 
 Its actions on the states $|\Psi_I\rangle$ and $|\Psi_H\rangle$ can then be derived using the identities (\ref{isomdef}), (\ref{eq:Jonesprojections}) and (\ref{isom-TL}). One finds
 \beq
{\cal R}(\theta)\, |\Psi_H\rangle  = \Bigl (1-{4\over d_I^2}\Bigr)|\Psi_H\rangle+ {2\over d_I} \,|\Psi_I\rangle~,   \qquad {\cal R}(\theta)\, |\Psi_I\rangle = |\Psi_I\rangle- {2\over d_I} \,|\Psi_H\rangle~. 
\label{R-action}
\eeq
The first relation shows that each application of $\mathcal{R}(\theta)$ partially rotates $|\Psi_H\rangle$ towards $|\Psi_I\rangle$. By repeatedly applying ${\cal R}(\theta)$ a certain number of times, the state $|\Psi_H\rangle$ can be rotated (almost) entirely into $|\Psi_I\rangle$. Equation (\ref{R-action}) makes clear that we need to apply ${\cal R}(\theta)$ a number of times\footnote{Here we assume that $d_I$ is large: for small $d_I$ one needs to be lucky. For instance, when $d_I=2$ (corresponding to a single qubit) one has the fortune to obtain the correct result after one single step.} that is of the order $d_I$. 
\begin{figure}
    \centering
    \includegraphics[width=0.4\linewidth]{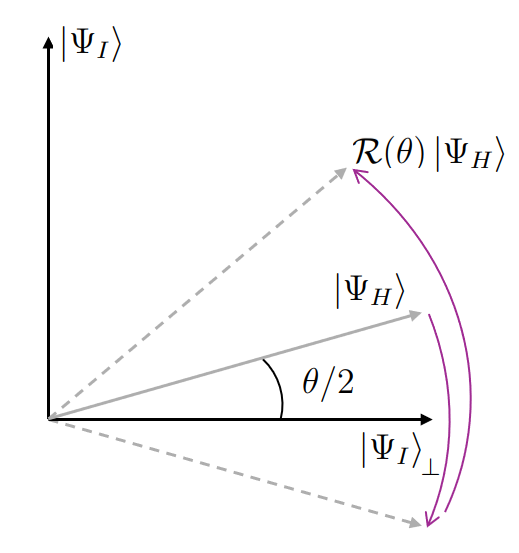}
    \caption{\small{A single iteration of Grover’s search. The rotation $\mathcal{R}(\theta)$ is realized as the composition of two reflections in the plane spanned by $|\Psi_I\rangle$ and $|\Psi_H\rangle$: first about the axis defined by $|\Psi_I\rangle_{\perp}$, and then about the axis defined by $|\Psi_H\rangle$.}}
    \label{Fig:GroverS}
\end{figure}

To evaluate the precise number of times that we need to act, let us explicitly compute the repeated action of ${\cal R}(\theta)$ on the states $|\Psi_H\rangle$ and $|\Psi_I\rangle$. From the relation (\ref{R-action}), one can derive a set of recursive relations for the matrix elements of the $n$-fold action of ${\cal R}(\theta)$. It is clear that these matrix elements are polynomials of order $2n$ in $1/d_I$. After a straightforward calculation, one finds that they are given by Chebyshev polynomials of the first and second kind. An alternative and more convenient way to compute the result is to write the transformation in terms of the angle $\theta$ defined in (\ref{thetadef}). For $n$ iterations, the action on $|\Psi_H\rangle$ is
\beq
\label{RnH}
\bigl({\cal R}(\theta)\bigr)^n \,|\Psi_H\rangle = {\cos{(n\!+\!{1\over 2})\theta}\over \cos{1\over 2}\theta} \left (|\Psi_H\rangle - {1\over d_I}|\Psi_I\rangle\right) + \sin (n\!+\!\textstyle{1\over 2})\theta \, |\Psi_I\rangle ~,
\eeq 
where the term in parentheses together with the factor $1/\cos{1\over 2}\theta$ represents the orthogonal state $|\Psi_I\rangle_\perp$. Using the overlap (\ref{overlap}), we can determine the number of iterations $n_G$ required to get close to the final state $|\Psi_I\rangle$. We thus impose that the first term in (\ref{RnH}) is as small as possible, and hence we find 
\beq 
\left(n_G\!+\!{\textstyle{1\over 2}}\right)\! \theta \approx \frac{\pi}{2}~.
\label{nGapprox}
\eeq 
When the angle between the state $\left (R(\theta)\right)^{n} |\Psi_H\rangle$ and $|\Psi_I\rangle$ is approaching zero, the probability of measuring the state $|\Psi_I\rangle$ as the final outcome of the search is very close to one. The approximation thus becomes more reliable when $d_I$ is sufficiently large. 

We are now ready to define the deterministic swap operator $S_G$. 
 It turns out to be useful to represent $S_G$ as a reflection so that it squares to the identity. This leads us to define the deterministic swap operator $S_G$ as
\begin{equation}
S_G \equiv ~ (2e_I-1)\bigl(\mathcal{R}(\theta)\bigr)^{n_G}~. 
\end{equation}
Here, the number $n_G$ is defined as the integer that gives the best approximation to the equation (\ref{nGapprox}).
One easily verifies that the swap operator $S_G$ obeys the identity $S_G^2 \!=\! {\mathbf 1}$ and, to a very good approximation, exchanges the states $|\Psi_H\rangle$ and $|\Psi_I\rangle$:
\beq
S_G|\Psi_H\rangle\approx |\Psi_I\rangle~, \qquad \qquad 
S_G|\Psi_I\rangle \approx |\Psi_H\rangle~. 
\eeq 
 This property will prove to be useful in the following.

\subsection{An algebraic notion of complexity}

We define the computational complexity associated with the quantum operation $S_G$ as determined by the number of steps in Grover's search algorithm. As shown in (\ref{nGapprox}) the number of iterations $n_G$ depends on the angle $\theta$. For large values of $d_I$ we may approximate $\theta \sim 2/d_I$ and drop the $1\over2$ in this equation. In this way we find that 
\beq 
\quad n_G \sim \frac{\pi d_I}{4} \qquad \text { for large } \quad d_I~.
\eeq 
 Let us denote the complexity of the deterministic swap operator $S_G$ by $\mathcal{C}[S_G]$. 
 We have thus shown that this complexity is directly related to the Jones index and may be expressed as
\begin{equation}
 \mathcal{C}[S_G]\sim \sqrt{\operatorname{Ind} E_I}~,
\end{equation}
up to some arbitrary choice of normalization. 

This simple computation identifies the index as an algebraic measure of computational complexity of the entanglement swap.\footnote{A relation between computational complexity and the Jones index was previously noted by Hollands in \cite{Hollands_2021} using different reasoning. Interestingly, the index of an algebra inclusion has also recently been discussed as a means to understand the volume of certain subregions in the bulk of AdS spacetime \cite{Leutheusser:2025zvp}, suggesting a potential ``volume-complexity" relation.} Since the Jones--Kosaki index is defined for general von Neumann algebras admitting a finite-index conditional expectation, the above relation naturally extends the notion of decoding complexity beyond finite-dimensional Hilbert spaces to type II and type III von Neumann algebras with finite conditional expectation. It would be interesting to investigate the physical implications of this algebraic notion of complexity in greater detail.

This result should be compared with the more familiar finite-dimensional tensor network realization of the Python’s Lunch proposal \cite{brown2020python}. There, the bulk-to-boundary map is modeled by a tensor network in which part of the circuit implements a postselection on a subset of qubits. Recovering the information requires undoing this postselection by a Grover-type unitary, leading to a computational cost that scales as the square root of the Hilbert space dimension of the postselected subsystem. In our algebraic framework, the Jones projection plays the role of the postselection operator, while the Jones index quantifies the size of the corresponding projection. For finite-dimensional qubit models, the Jones index reduces to the Hilbert space dimension of the postselected subsystem, reproducing the Python’s Lunch scaling. 

This postselected subsystem is further related to geometric Python’s Lunch regions by using the fact that locally minimal tensor network cuts correspond to quantum extremal surfaces. Our setup, while remaining at the abstract, non-geometric level of operator algebras, mirrors this structure: we have two QESs at distinct algebraic locations, $H$ and $I$, and the entanglement swap from $|\Psi_H\rangle$ to $|\Psi_I\rangle$ can be viewed as shifting the minimal QES across an intermediate “bulge” region. In the next section, we will try to connect to this picture and define a corresponding island algebra and show how, in our operator-algebraic setting, the index quantifies the effective ``size'' of this region.

\section{The Island Algebra}
\label{sec: The Island Algebra}

We have algebraically described how information is processed by the black hole in both the microscopic and the effective descriptions, identifying the evaporation step as a complex entanglement swap. In recent years, the spacetime understanding of information release (in terms of the Page curve) has been phrased in terms of the appearance of an island region that is part of the entanglement wedge of the radiation. 

In this section, we define the analogue of the island in our algebraic model as an effective algebra that represents part of the radiation system encoded in the  black hole interior, using the notion of Jones tower and the canonical shift.

\subsection{The Jones tower and the canonical shift}

The fact that the map $V_H^\dagger$ is non-isometric implies that the microscopic algebra ${\cal A}_{BH}\times {\cal A}_R$ is a subalgebra of the spacetime algebra ${\cal B}_{BH}\times {\cal B}_R$. A similar statement holds after evaporation and identifies ${\cal A}_{B}\times {\cal A}_{RH}$ with a subalgebra of ${\cal B}_{B}\times {\cal B}_{RH}$. The spacetime algebras can thus be regarded as extensions of the microscopic algebras, and, as we now explain, are obtained from the microscopic algebras via the \textit{basic construction} \cite{jones1983index} (see \cite{vanderHeijden:2024tdk} for a more detailed discussion). 

The maps before and after evaporation can be rewritten as follows in terms of conditional expectations 
by exploiting their definition in terms of the isometries $V_H$ and $V_I$ and their non-isometric conjugates $V_H^\dagger$ and $V^\dagger_I$. 
For the black hole algebras this leads to the relations
\beq 
{\cal A}_{BH} = E_H\bigl({\cal B}_{BH}\bigr)~,\qquad\mbox{and} \qquad {\cal A}_B = E_I\bigl({\cal B}_B\bigr)~,
\eeq 
and for the radiation algebras one finds 
\beq 
{\cal A}_{RH} = E'_I\bigl({\cal B}_{RH}\bigr)~,\qquad\mbox{and} \qquad {\cal A}_R = E'_H\bigl({\cal B}_R\bigr)~.
\eeq 
Comparing these equations with earlier identities (\ref{condexp1}) and (\ref{condexp2}) for the conditional expectations, one draws the important conclusion that the larger microscopic black hole and radiation algebras ${\cal A}_{BH}$ and ${\cal A}_{RH}$ are equal to the smaller spacetime black hole and radiation algebras ${\cal B}_{B}$ and ${\cal B}_{R}$:
\beq 
{\cal A}_{BH} = {\cal B}_{B}~, \qquad\mbox{and} \qquad {\cal A}_{RH} = {\cal B}_R~. \label{AequalsB}
\eeq 
Identifying the spacetime algebras with their corresponding microscopic algebras as above, we find that the sequence of algebraic inclusions for the black hole algebras can be written as
\begin{equation} 
     \caB \subset\caBH \subset \cbBH~.
\end{equation} 
Similarly, for the radiation algebras, we find the following sequence of inclusions
\begin{equation} 
    \caR \subset\caRH \subset \cbRH~.
\end{equation}

We claim that these algebraic inclusions are instances of the Jones basic construction, defined as follows (see also Appendix \ref{app:math_background}). Starting with a given algebra inclusion with a corresponding conditional expectation, one can define the Jones projection that maps the GNS Hilbert space of the larger algebra (defined on a given cyclic and separating state) to the GNS Hilbert space of the smaller algebra. This projection is not part of either of these algebras, and hence can be added to the largest of the two algebras. By taking a double commutant of this extended algebra, one again obtains a von Neumann algebra. In our setup, when this construction is applied to the inclusion of the microscopic algebras, it yields the corresponding effective algebras. In other words, we can write the algebras ${\cal B}_{BH}$ and ${\cal B}_{RH}$ as 
\beq 
{\cal B}_{BH} = \langle {\cal A}_{BH}, e_I\rangle~,
\qquad \qquad 
{\cal B}_{RH} = \langle {\cal A}_{RH}, e_H\rangle~.
\eeq
Here, the notation $\langle\cdot, \cdot\rangle$ means that we take the union of the algebras and subsequently the double commutant to ensure that we end up with a von Neumann algebra. 

We already noted that the projector $e_I$ is part of the spacetime black hole algebra ${\cal B}_{BH}$, while the projector $e_H$ is contained in the spacetime radiation algebra ${\cal B}_{RH}$.
One can repeat the Jones basic construction by adding the remaining projection to each of these algebras. In this way, we get an even longer sequence of algebra inclusions, which we refer to as the \textit{Jones tower}. For the black hole algebra, it reads
\begin{equation} \label{eq:Jonestower}
\caB\subset\caBH\subset\cbBH\subset\langle\cbBH,e_H\rangle~.
\end{equation}
The last and largest algebra in \eqref{eq:Jonestower} represents the commutant of the radiation algebra ${\cal A}_R$. 
By taking the commutants of the full sequence we obtain a similar set of inclusions for the radiation algebras 
\begin{equation} \label{eq:Jonestower2}
\caR\subset\caRH\subset\cbRH\subset\langle\cbRH,e_I\rangle~, 
\end{equation}
where the largest algebra corresponds to the commutant of ${\cal A}_B$. 

An alternative way to define the Jones extension of an algebra is to make use of the modular conjugation operators $J_H$ and $J_I$ associated with the state $|\Psi_H\rangle$ and $|\Psi_I\rangle$.
These modular conjugations can be interpreted physically as mapping the operators in the black hole algebras to their entangled partners in the radiation. For instance, using the relations (\ref{AequalsB}), the following algebras are related by modular conjugation 
\beq \label{eq:M1}
\cbBH = J_H\caRH J_H~,
\qquad \mbox{and}\qquad 
\cbRH = J_I \caBH J_I~. 
\eeq 
The modular conjugations also exchange the subalgebras on both sides, and hence one also has the relations 
\beq \label{M1}
\caBH = J_H\caR J_H~,
\qquad \mbox{and}\qquad 
\caRH = J_I \caB J_I~.
\eeq 
Via the repeated application of the modular conjugation operators, it is not difficult to see that one can define a ``two-step'' map that relates alternate algebras in the Jones tower \cite{longo1989index}:
\beq
 \Gamma\equiv J_H J_I~,\qquad\mbox{and} \qquad  \Gamma^\dagger \equiv J_IJ_H~,
 \eeq
 which we call the \textit{canonical shift} operator.\footnote{The notion of the canonical shift was first introduced by Ocneanu \cite{Ocneanu1989operator} in the context of type II$_1$ algebras, and was presented in the general form by Longo in \cite{longo1989index}. 
Its relevance for the algebraic description of the Hayden--Preskill protocol was discussed in \cite{vanderHeijden:2024tdk}, where the canonical shift was interpreted as a teleportation step that takes the quantum information from the black hole and transports it to the radiation system.} By conjugating the algebras with $\Gamma$ one moves two steps up or down the Jones tower. One easily verifies, using the combined modular conjugations (\ref{eq:M1}) and (\ref{M1}), that the bulk algebras $\cbBH$ and $\cbRH$ can be identified with the images of the original algebras ${\cal A}_B$ and ${\cal A}_R$ under the canonical shift. One has
\beq 
\label{can2}
{\cal B}_{BH} = \Gamma {\cal A}_{B} \Gamma^\dagger~,  \qquad \mbox{and  } \qquad 
{\cal B}_{RH} = \Gamma^\dagger {\cal A}_{R} \Gamma~. 
\eeq
Similarly, one shows that the largest extended algebras in the Jones tower are obtained by applying the canonical shift to the algebras ${\cal A}_{BH}$ and ${\cal A}_{RH}$. This leads to
\beq 
\label{can1}
\langle\cbBH,e_H\rangle   = \Gamma \caBH \Gamma^\dagger~, 
 \qquad \mbox{and  } \qquad   \langle\cbRH,e_I\rangle   = \Gamma^\dagger \caRH \Gamma~. 
\eeq
We will make use of these relations below to describe how the information gets transferred from the black hole to the radiation algebras.

The generalization to properly infinite algebras makes use of Longo's treatment of the Jones basic construction \cite{longo1989index, Longo:1990zp}. There, the operator $\Gamma$ is used to define a \textit{canonical endomorphism} on the operator algebra: 
$\gamma: x\to \Gamma^{\dagger}x\Gamma$. Its inverse $\gamma^{-1}$ corresponds to the canonical shift. In this context, the operator $\Gamma$ represents a unitary map: it can be expressed as the product of the unitary $U$ and the conjugated unitary operator $\overline{U}$. 
The canonical endomorphism can be used to reverse the canonical shift and allows one to identify operator algebras that are separated by two steps in the Jones tower.\footnote{This represents an algebraic version of ``Hilbert's hotel paradox'' that only applies to properly infinite algebras. In this paper, we have avoided this change of terminology, since it breaks the correspondence with other types of algebras and would complicate our description of the effective-to-microscopic map.}

\subsection{A definition of the island algebra}

We are now ready to introduce the algebra associated with the interior region near and between the two QESs. We will refer to this as the \emph{extended island algebra}, denoted ${\cal B}_{HI}$. It is the algebra that includes all the degrees of freedom involved in the effective evaporation step, and is therefore larger than the algebra that we will refer to as the \emph{island algebra}, denoted $\mathcal{B}_I$, whose definition will be introduced shortly.

In the effective description, we define ${\cal B}_{HI}$ as the commutant of the union of the microscopic algebras ${\cal A}_B$ and ${\cal A}_R$: 
\beq
{\cal B}_{HI} \equiv \bigl({\cal A}_B \cup {\cal A}_R\bigr)'= {\cal A}_B' \cap {\cal A}_R'~.
\eeq 
The extended island algebra $\mathcal{B}_{HI}$ thus collects all operators that commute with both the black hole boundary algebra $\mathcal{A}_B$ and the radiation boundary algebra $\mathcal{A}_R$.
Physically, it is the algebra of operators supported in the interior region between the two quantum extremal surfaces $H$ and $I$. By construction, $\mathcal{B}_{H I}$ contains, among other operators, the Jones projections $e_H$ and $e_I$ onto the spacetime vacuum states $\left|\Psi_H\right\rangle$ and $\left|\Psi_I\right\rangle$ near the two surfaces, as well as the Hawking partner modes and the modes that carry the information about Alice's diary. The location of the diary and Hawking partner modes inside the island algebra differ significantly between the reducible and irreducible cases, so we will discuss these two situations separately.

Before turning to this distinction, it is useful to understand more concretely which operators are added at each step of the Jones towers \eqref{eq:Jonestower} and \eqref{eq:Jonestower2}, and how they participate in the information release from the black hole interior. We will keep this discussion general at first, and then restrict to situations in which the operators are contained in specific relative commutants.
\begin{enumerate}
    \item For the first inclusions $\caB\subset\caBH$ and $\caR\subset\caRH$ we already identified the operators with the modes $a_i$ that carry the information about Alice's diary and the emitted Hawking modes $b_i$ collected by Bob.  
    \item For the next inclusions, the operators that are added via the inclusion of black hole algebras ${\cal A}_{BH}\subset {\cal B}_{BH}$ are the Hawking partners $\tilde{b}_i$ of radiation modes $b_i$. Similarly, the internal modes $a_i$ that carry the information have partner modes $\tilde{a}_i$ that are contained in the spacetime radiation algebra ${\cal B}_{RH}\supset {\cal A}_{RH}$. It follows from (\ref{eq:M1}) that the partner modes are related to the original $b$- and $a$-modes via the modular conjugation operators $J_H$ and $J_I$ 
\beq
 {\tilde{b}_i} \equiv J_H b_i J_H~,\qquad\mbox{and} \qquad {\tilde{a}_i} \equiv J_I a_i J_I~.
\eeq 
\item In the last step in each Jones tower, the algebras ${\cal B}_{BH}$ and ${\cal B}_R$ are extended via Jones basic construction to the algebras given in (\ref{can1}). These algebras have their own Pimsner--Popa basis, which we denote by $\hat{a}_i$ and $\hat{b}_i$. The relations (\ref{can1}) imply that these hatted operators are related to the original black hole and radiation modes $a_i$ and $b_i$ via the canonical shift
\beq 
\label{hatGamma}
\hat{a}_i= \Gamma a_i \Gamma^\dagger~, \qquad\mbox{and} \qquad \hat{b}_i = \Gamma^\dagger b_i \Gamma~. 
\eeq 
\end{enumerate}

In the following, we will think of the island region as containing the Hawking partners $\tilde{b}_i$ of the radiation, as well as the operators $\hat{a}_i$ that encode the information in Alice's diary. In the reducible case this interpretation is automatic, since these operators already live in first relative commutants inside $\mathcal{B}_{H I}$; in the irreducible case it relies on the special property that the ``depth'' of the inclusion is equal to two, which ensures that the relevant modes appear in second relative commutants inside $\mathcal{B}_{H I}$. We will explain the physical significance of this property below and discuss its mathematical definition in the appendix.

\paragraph{Reducible inclusions.} The easiest and most straightforward case occurs when the inclusions are completely reducible, which means that the relevant diary and radiation modes are already contained in non-trivial relative commutant algebras inside $\mathcal{B}_{H I}$.

In the microscopic description before and after evaporation, the relative commutant of the black hole algebra with respect to the boundary algebra $\mathcal{A}_B$ is an algebra we denoted by $\mathcal{A}_I$ in (\ref{AI}). It can be defined in two equivalent ways in the microscopic theory:
\beq
{\cal A}_I\ \equiv\  {\cal A}_{BH}\cap {\cal A}'_B~, 
 \qquad\quad \mbox{or equivalently}\qquad\quad   {\cal A}_I\ \equiv {\cal A}_{RH}\cap {\cal A}_R'~.
\eeq
In this microscopic setting, the two relative commutants coincide, and both describe the degrees of freedom that are simultaneously compatible with the black hole and radiation algebras. In the effective description, the corresponding relative commutant algebras appear as distinct subalgebras of the extended island algebra. To make this explicit, we define 
\beq
{\cal A}_H \equiv {\cal A}_{BH}\cap {\cal A}'_B~, \qquad\quad \mbox{and} \qquad \quad {\cal A}_I \equiv {\cal A}_{RH}\cap {\cal A}'_R~,
\eeq
where $\mathcal{A}_H$ contains the internal diary modes $a_i$ and $\mathcal{A}_I$ contains the radiation modes $b_i$ that carry the same information seen from the outside. These algebras represent the first relative commutants associated with the quantum extremal surfaces $H$ and $I$.

Although these two algebras differ in the effective description, there is a direct relationship between ${\cal A}_I$ and ${\cal A}_H$ via the canonical shift $\Gamma$. This can be seen as follows. As we explained, the radiation algebra ${\cal A}_{RH}$ is identified with the bulk radiation algebra ${\cal B}_R$. Its commutant is therefore equal to the bulk black hole algebra ${\cal B}_{BH}$, which in turn is related to ${\cal A}_B$ by the canonical shift $\Gamma$. The bulk commutant of the algebra ${\cal A}_R$ is given by the extended black hole algebra $\langle {\cal B}_{BH}, e_H\rangle $, which is related to ${\cal A}_{BH}$ by the canonical shift. These identifications make clear that in the bulk description ${\cal A}_I$ and ${\cal A}_H$ are related by 
\beq
{\cal A}_I = \Gamma {\cal A}_H \Gamma^\dagger~, \qquad \qquad \mbox{in the effective description~.}
\eeq 
It follows from this relation that the operators $b_i$ and $\hat{a}_i$ form a basis of the same algebra ${\cal A}_I$ and hence must be related by a unitary transformation. Similarly, the operators $a_i$ and $\hat{b}_i$ must also be unitarily related. We thus learn that 
\beq
\hat{a}_i ={\cal U}_{ij}b_j~, \qquad\quad\mbox{and}\qquad\quad  \hat{b}_i = {\cal U}^\dagger_{ij} a_j~. 
\label{unitary}
\eeq 

\tikzset{every picture/.style={line width=0.5pt}} 
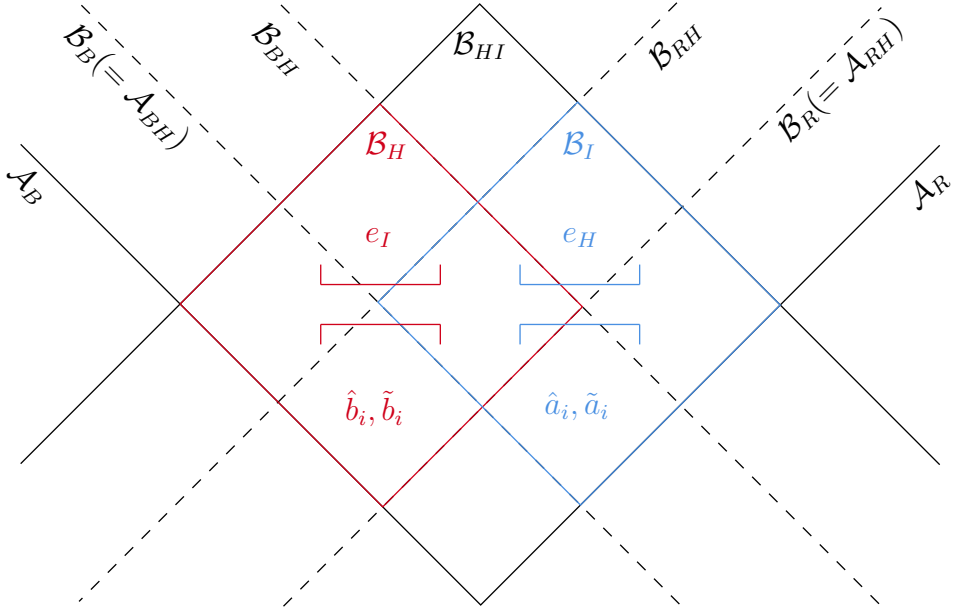
\begin{figure}[t]
    \centering
    \begin{tikzpicture}[x=0.75pt,y=0.75pt,yscale=-1,xscale=1]

\draw   (179.47,169.75) -- (329.75,19.47) -- (480.53,170.25) -- (330.25,320.53) -- cycle ;
\draw   (99.47,89.75) -- (179.47,169.75) -- (99.47,249.75) ;
\draw   (560.53,250.25) -- (480.53,170.25) -- (560.53,90.25) ;
\draw  [dash pattern={on 4.5pt off 4.5pt}] (129.93,20.07) -- (278.69,168.83) -- (128.76,318.76) ;
\draw  [dash pattern={on 4.5pt off 4.5pt}] (230.07,19.93) -- (381.31,171.17) -- (231.24,321.24) ;
\draw  [color={rgb, 255:red, 208; green, 2; blue, 27 }  ,draw opacity=1 ] (179.47,169.75) -- (279.68,69.53) -- (381.31,171.17) -- (281.1,271.38) -- cycle ;
\draw  [dash pattern={on 4.5pt off 4.5pt}] (530.07,319.93) -- (381.31,171.17) -- (531.24,21.24) ;
\draw  [dash pattern={on 4.5pt off 4.5pt}] (429.93,320.07) -- (278.69,168.83) -- (428.76,18.76) ;
\draw  [color={rgb, 255:red, 74; green, 144; blue, 226 }  ,draw opacity=1 ] (278.69,168.83) -- (378.9,68.62) -- (480.53,170.25) -- (380.32,270.47) -- cycle ;
\draw [color={rgb, 255:red, 74; green, 144; blue, 226 }  ,draw opacity=1 ]   (350,160) -- (410,160) ;
\draw [color={rgb, 255:red, 74; green, 144; blue, 226 }  ,draw opacity=1 ]   (350,180) -- (410,180) ;
\draw [color={rgb, 255:red, 74; green, 144; blue, 226 }  ,draw opacity=1 ]   (350,150) -- (350,160) ;
\draw [color={rgb, 255:red, 74; green, 144; blue, 226 }  ,draw opacity=1 ]   (410,150) -- (410,160) ;
\draw [color={rgb, 255:red, 74; green, 144; blue, 226 }  ,draw opacity=1 ]   (350,180) -- (350,190) ;
\draw [color={rgb, 255:red, 74; green, 144; blue, 226 }  ,draw opacity=1 ]   (410,180) -- (410,190) ;
\draw [color={rgb, 255:red, 208; green, 2; blue, 27 }  ,draw opacity=1 ]   (250,160) -- (310,160) ;
\draw [color={rgb, 255:red, 208; green, 2; blue, 27 }  ,draw opacity=1 ]   (250,180) -- (310,180) ;
\draw [color={rgb, 255:red, 208; green, 2; blue, 27 }  ,draw opacity=1 ]   (250,150) -- (250,160) ;
\draw [color={rgb, 255:red, 208; green, 2; blue, 27 }  ,draw opacity=1 ]   (310,150) -- (310,160) ;
\draw [color={rgb, 255:red, 208; green, 2; blue, 27 }  ,draw opacity=1 ]   (250,180) -- (250,190) ;
\draw [color={rgb, 255:red, 208; green, 2; blue, 27 }  ,draw opacity=1 ]   (310,180) -- (310,190) ;

\draw (315,32.4) node [anchor=north west][inner sep=0.75pt]    {$\mathcal{B}_{HI}$};
\draw (271,82.4) node [anchor=north west][inner sep=0.75pt]  [color={rgb, 255:red, 208; green, 2; blue, 27 }  ,opacity=1 ]  {$\mathcal{B}_{H}$};
\draw (370,82.4) node [anchor=north west][inner sep=0.75pt]  [color={rgb, 255:red, 74; green, 144; blue, 226 }  ,opacity=1 ]  {$\mathcal{B}_{I}$};
\draw (100,95) node [anchor=north west][inner sep=0.75pt]  [rotate=-45]  {$\mathcal{A}_{B}$};
\draw (129.21,25) node [anchor=north west][inner sep=0.75pt]  [rotate=-45]  {$\mathcal{B}_B(=\mathcal{A}_{BH})$};
\draw (222.85,25) node [anchor=north west][inner sep=0.75pt]  [rotate=-45]  {$\mathcal{B}_{BH}$};
\draw (410.81,45.43) node [anchor=north west][inner sep=0.75pt]  [rotate=-315]  {$\mathcal{B}_{RH}$};
\draw (475,80.43) node [anchor=north west][inner sep=0.75pt]  [rotate=-315]  {$\mathcal{B}_R(=\mathcal{A}_{RH})$};
\draw (540,115.43) node [anchor=north west][inner sep=0.75pt]  [rotate=-315]  {$\mathcal{A}_{R}$};
\draw (271,130.4) node [anchor=north west][inner sep=0.75pt]    {$\textcolor[rgb]{0.82,0.01,0.11}{e_{I}}$};
\draw (370,130.4) node [anchor=north west][inner sep=0.75pt]    {$\textcolor[rgb]{0.29,0.56,0.89}{e_{H}}$};
\draw (261,211.4) node [anchor=north west][inner sep=0.75pt]  [color={rgb, 255:red, 208; green, 2; blue, 27 }  ,opacity=1 ]  {$\hat{b}_{i} ,\tilde{b}_{i}$};
\draw (361,212.4) node [anchor=north west][inner sep=0.75pt]  [color={rgb, 255:red, 74; green, 144; blue, 226 }  ,opacity=1 ]  {$\hat{a}_{i} ,\tilde{a}_{i}$};

\end{tikzpicture}
    \caption{An illustration of the inclusion structure of the different algebras as well as the location of the Jones projectors $e_I$ and $e_H$ and modes $\hat{a}_i,\tilde{a}_i,\hat{b}_i,\tilde{b}_i$. While the diagram suggests an interpretation in terms of spacetime subregions with a given causal structure, we stress that this is not something we have addressed here. One should simply view the above diagram as a convenient representation of the degrees of freedom without any direct reference to spacetime subregions.}
\label{fig:inclusionstructure}
\end{figure}

\paragraph{Irreducible inclusions.}
We now turn to the case in which the inclusion ${\cal A}_B\subset{\cal A}_{BH}$ is irreducible. 
In this case, the relative commutants ${\cal A}_I$ and ${\cal A}_H$ are trivial and consist of multiples of the identity. This means that we have to look further to find the algebra that contains the operators $a_i$ and $b_i$. The place to look is in higher relative commutants, which leads to the notion of the ``depth of an inclusion". In this paper we only consider inclusions with depth two.\footnote{We refer to the appendix for a more general discussion on the notion of depth for an algebra inclusion. }

The extended island algebra ${\cal B}_{HI}$ corresponds to a third relative commutant, since it involves algebras that are separated by three steps in the Jones tower.
The depth two property means that the $a$-modes or the $b$-modes are contained in a second relative commutant, i.e.~in the algebra of operators that commute with the subfactor when one step up in the Jones tower is taken.
This property then also holds for the partner modes $\tilde{a}$, and $\tilde{b}$ as well as for $\hat{a}$ and $\hat{b}$. We therefore introduce the following second relative commutants: 
\beq
{\cal B}_{I}= {\cal B}_{RH} \cap {\cal A}_R'~,\quad \qquad\mbox{and} \quad \qquad {\cal B}_{H} =  {\cal B}_{BH}\cap {\cal A}_B'~.
\eeq
The algebra ${\cal B}_{I}$, which we call the island algebra, contains the operators $\tilde{a}_i$ and $\hat{a}_i$, while the operators $\tilde{b}_i$ and $\hat{b}_i$ are contained in ${\cal B}_H$. To see why these statements hold, we need to look more closely at the sequence of algebra inclusions and explain what is meant by a depth two inclusion. 

It is easily seen from the Jones towers (\ref{eq:Jonestower}) and (\ref{eq:Jonestower2}) that the following inclusion relations hold:
\beq
{\cal A}_H \subset {\cal B}_H\subset {\cal B}_{HI}~,\qquad\quad \mbox{and}\qquad \quad {\cal A}_I \subset {\cal B}_I\subset {\cal B}_{HI}~. 
\eeq 
The three algebras in each sequence represent the first, second and third relative commutants. The mathematical definition of the depth two property is that these two sequences of inclusions represent an instance of the Jones basic construction. This is a powerful condition, which allows one to draw a number of important conclusions. 

First, this property implies that the conditional expectations for the first algebra inclusion in each sequence are still represented by $E_H$ and $E'_I$. From the equations (\ref{condexp2}) it is easily seen that 
\beq 
E_H\!: {\cal B}_H\to  {\cal A}_H~,\qquad\quad \mbox{and}\qquad\quad  E'_I\!: {\cal B}_I\to {\cal A}_I~.
\eeq
Furthermore, the depth two assumption tells us, by definition, that the extended algebra ${\cal B}_{HI}$ is obtained via Jones's basic construction for both sequences of inclusions. This means in particular that it contains both Jones projections $e_H$ and $e_I$. One has 
\beq
{\cal B}_{HI} = \langle {\cal B}_H, e_H\rangle = \langle {\cal B}_I, e_I\rangle~. 
\eeq
Finally, perhaps the most important implication is that 
for the irreducible case the algebras $\mathcal{B}_H$ and ${\cal B}_I$ each have precisely $d_I^2$ independent algebra elements. This follows from the fact that the algebras ${\cal A}_H$ and ${\cal A}_I$ are one-dimensional and that the index of the conditional expectations is still the same. 

We will identify the basis for the algebra ${\cal B}_I$ with the hatted operators $\hat{a}_i$, and for ${\cal B}_H$ we will choose the operators $\hat{b}_i$ as its basis. We could have used the partner modes $\tilde{a_i}$ or $\tilde{b}_i$ as an alternative basis for these algebras. The hatted and tilde-ed operators are related by the modular conjugations $J_H$ and $J_I$, each of which defines an anti-linear endomorphism of the algebra ${\cal B}_I$ or ${\cal B}_H$ respectively. This fact allows us to choose our bases of operators so that $\hat{a}_i = \tilde{a}_i^\dagger $ and $\hat{b}_i = \tilde{b}_i^\dagger$.

In summary, the extended island algebra $\mathcal{B}_{HI}$ is intrinsically defined as the commutant of the boundary algebras $\mathcal{A}_B$ and $\mathcal{A}_R$. In the reducible case, the diary and Hawking partner modes already live in the first relative commutants $\mathcal{A}_H$ and $\mathcal{A}_I$ inside $\mathcal{B}_{HI}$, and the canonical shift $\Gamma$ identifies the corresponding degrees of freedom near the two quantum extremal surfaces $H$ and $I$. In the irreducible, depth two case, the same modes instead populate the second relative commutants $\mathcal{B}_H$ and $\mathcal{B}_I$, but again these algebras sit inside $\mathcal{B}_{HI}$ and are related by the canonical shift.

\section{Information Recovery and Black Hole Complementarity}
\label{sec: Information recovery and black hole complementarity}

We are now ready to address a version of the black hole information problem in our model. A central role in this discussion is played by the island algebra introduced above. We will also explain how a notion of black hole complementarity shows up, which manifests itself most clearly in the irreducible case and is intimately connected to the algebraic complexity of the information transfer.

\subsection{Algebraic formulation of the information puzzle}
 Suppose Alice encodes the quantum information in her diary by making use of the set of operators $a_i$ that are part of the microscopic black hole algebra ${\cal A}_{BH}$ before evaporation.
After the evaporation Bob has gained control of the modes $b_i\in {\cal A}_{RH}$ in the radiation algebra. Our goal is to explain from the effective description how the operators $a_i$ get mapped onto the operators $b_i$ and determine in which way, and under which circumstances, this knowledge can be used to decode the information in Alice's diary. 

To describe the information transfer in a tractable and physically intuitive manner we assume that the operators $a_i$ and $b_i$ each generate an algebra. We will consider two special situations in which the $a_i$ and $b_i$ operators are contained in a simple relative commutant. The two algebras generated by $a_i$ and $b_i$ can then be used to construct two GNS Hilbert spaces on their respective tracial states. We will denote these GNS Hilbert spaces by ${\cal H}_I$ for the operators $a_i$ and by ${\cal H}_H$ for the modes $b_i$. The dimensions of these Hilbert spaces are, for the two cases that we consider, equal to the index, which therefore necessarily is an integer. 
In fact, the basis of states in each Hilbert space can be identified with the corresponding operators: 
\beq
{\cal H}_I = {\rm Span}\Bigl\{|a_i\rangle: \, i=1,\dots, d_I^2
\Bigr \}~, \quad\mbox{and} \quad 
{\cal H}_H = {\rm Span}\Bigl\{|b_i\rangle: \, i=1,\dots,d_I^2
\Bigr \}~. 
\eeq 
We have deliberately not yet given names to the algebras associated with the operators $a_i$ and $b_i$ because the identification of these algebras will depend on whether the algebraic inclusion ${\cal A}_B\subset {\cal A}_{BH}$ is reducible or irreducible. In both cases the black hole information puzzle amounts to finding the map ${\cal H}_H \to {\cal H}_I$ and identifying how the microscopic state $|\Phi\rangle$ is represented in each of these GNS Hilbert spaces. 

Let us start with the second question. The state $|\Phi\rangle$ can be mapped to an operator $\Phi_I$ in the algebra generated by $a_i$ by imposing that the expectation values in the microscopic state and in the GNS representation agree for all operators $a$ in that algebra
\beq
{\rm Tr}_I \bigl(\Phi_{I\!}^\dagger \,a\,\Phi_I\bigr) =\langle\Phi|a|\Phi\rangle ~.
\eeq 
Here ${\rm Tr}_I$ denotes the natural trace on the algebra generated by the operators $a_i$. 
Similarly, there exists an operator $\Phi_H$ in the algebra generated by $b_i$ that for all operators $b$ in that algebra obeys
\beq
{\rm Tr^{}_{}}_H\bigl(\Phi_H^\dagger{}_{{}_{\,}} b \, \Phi_H\bigr) =\langle\Phi|\,b\,|\Phi\rangle~. 
\eeq 
An alternative and equivalent way to obtain the operators $\Phi_I$ and $\Phi_H$ is to start with the state $|\Phi\rangle$ and construct the density matrices $\rho_I$ and $\rho_H$ for the algebras generated by $a_i$ and $b_i$ by tracing out the commutant algebras. 
We subsequently construct their GNS purifications, which we then denote by $|\Phi_I\rangle$ and $|\Phi_H\rangle$. In this way, we have represented the microscopic state $|\Phi\rangle$ by two different GNS-states 
\beq
|\Phi_I\rangle \in {{\cal H}_I}~,\qquad \mbox{and} \qquad |\Phi_H\rangle \in {{\cal H}_H}~.
\eeq 
These two states carry the same quantum information, but are represented in two complementary ways. The version of the black hole information problem in our model amounts to finding the map between these two states: 
\beq
|\Phi_H\rangle  \qquad \longrightarrow \qquad |\Phi_I\rangle~. 
\eeq 
To resolve this puzzle, we follow our description of the effective evaporation process and study its implementation in the GNS representation of the extended island algebra ${\cal B}_{HI}$. 

\subsection{Reducible inclusions: Information transfer via a unitary}

This task is relatively simple for the \textit{reducible} situation. The identity (\ref{unitary}) shows that the transfer of information from the black hole to the radiation involves the application of the canonical shift $\Gamma$ followed by an appropriate unitary operator. The shift $\Gamma$ can be interpreted as a quantum teleportation step, as illustrated by the pictorial representation of the black hole evaporation process in Figure~\ref{fig:informationtransfer}. 
This leads to a canonical identification between the algebras ${\cal A}_I$ and ${\cal A}_H$. 

In the microscopic theory, the operators $a_i$ and $b_i$ are part of the same algebra and are mapped onto each other by the unitary ${\cal U}_{ij}$. When translated to the GNS Hilbert spaces ${\cal H}_I$ and ${\cal H}_H$ this implies that the bases $|a_i\rangle$ and $|b_i\rangle$ are simply related by
\beq 
|a_i\rangle = {\cal U}_{ij}|b_j\rangle~, 
\qquad\quad \mbox{for the reducible case}~.
\eeq 
After this change of basis the Hilbert spaces ${\cal H}_I$ and ${\cal H}_H$ are directly identified, and also the states $|\Phi_I\rangle$ and $|\Phi_H\rangle$ have become identical. 
We conclude that in the completely reducible situation the information is transferred in a rather direct fashion from the black hole to the radiation and essentially happens via the canonical shift. This is the situation studied in \cite{vanderHeijden:2024tdk}.

\subsection{Irreducible inclusions: Information transfer via a Fourier transform}

We now turn to the case of an \textit{irreducible} inclusion. In this case the relation between the states $|\Phi_I\rangle$ and $|\Phi_H\rangle$ takes a rather non-trivial form and turns out to be given by the algebraic analogue of a Fourier transform.

As a first step, we note that the Hilbert spaces ${\cal H}_I$ and ${\cal H}_H$ can be identified, via the canonical shift $\Gamma$, with the GNS Hilbert spaces for the subalgebras ${\cal B}_I$ and ${\cal B}_H$ of the island algebra ${\cal B}_{HI}$. Strictly speaking, the original operators $a_i$ and $b_i$ are not themselves elements of $\mathcal{B}_I$ and $\mathcal{B}_H$, but their GNS basis states $|a_i\rangle$ and $|b_i\rangle$ can be canonically identified with basis elements $|\hat{a}_i\rangle$ and $|\hat{b}_i\rangle$ in these algebras via the canonical shift $\Gamma$. From now on, we will therefore drop the hats and regard $a_i$ and $b_i$ as basis elements of $\mathcal{B}_I$ and $\mathcal{B}_H$, respectively. In particular, $|\Phi_I\rangle$ and $|\Phi_H\rangle$ can be treated as GNS states associated with these algebras.

To resolve our information puzzle, we need to find a map from ${\cal B}_I$ to ${\cal B}_H$ and vice versa. This problem has been solved by Ocneanu \cite{Ocneanu1989operator} precisely for the case of irreducible depth two inclusions, i.e.~for the case where there are only trivial operators in the first relative commutant. Ocneanu showed that, in the depth two case, there exists a natural map ${\cal F}$ from the first subalgebra ${\cal B}_I$ to the second subalgebra ${\cal B}_H$ given by the algebraic equivalent of a Fourier transform. 
The transform $\cal F$ and its inverse ${\cal F}^{-1}$ thus act as 
\beq 
{\cal F}\!: {\cal B}_I \to {\cal B}_H\qquad\quad\mbox{and}\qquad\quad {\cal F}^{-1}\!: {\cal B}_H \to {\cal B}_I~.
\eeq 
The definitions of ${\cal F}$ and ${\cal F}^{-1}$ involve, besides the Jones projections $e_I$ and $e_H$, 
the conditional expectations
\beq
{\cal E}_I\!: {\cal B}_{HI} \to {\cal B}_H \qquad \quad \mbox{and} \qquad \quad {\cal E}'_H\!: {\cal B}_{HI} \to {\cal B}_I~,
\eeq
which are related to the conditional expectations $E_I$ and $E'_H$ by the canonical shift $\Gamma$.
For arbitrary elements $a \in \mathcal{B}_I$ and $b \in \mathcal{B}_H$, one can show that
\beq 
\label{FourierDef}
{\cal F}(a) = d^3_{I\,}{\cal E}_I\bigl( e_H e_{I}a\bigr) \qquad\quad\mbox{and}\qquad\quad {\cal F}^{-1}(b) = d^3_{I\,} {\cal E}'_H\bigl( e_I e_{H}b\bigr) ~.
\eeq 
These formulas may look formidable at first sight, but they have a simple graphical interpretation in terms of planar diagrams: the Fourier transform is a ``one-click rotation'', i.e.~a $90^\circ$ rotation of the diagram representing the operator \cite{Ocneanu1989operator}. This representation is illustrated in Figure~\ref{Fig:Fourier}.

\begin{figure}
    \centering
    \includegraphics[width=0.9\linewidth]{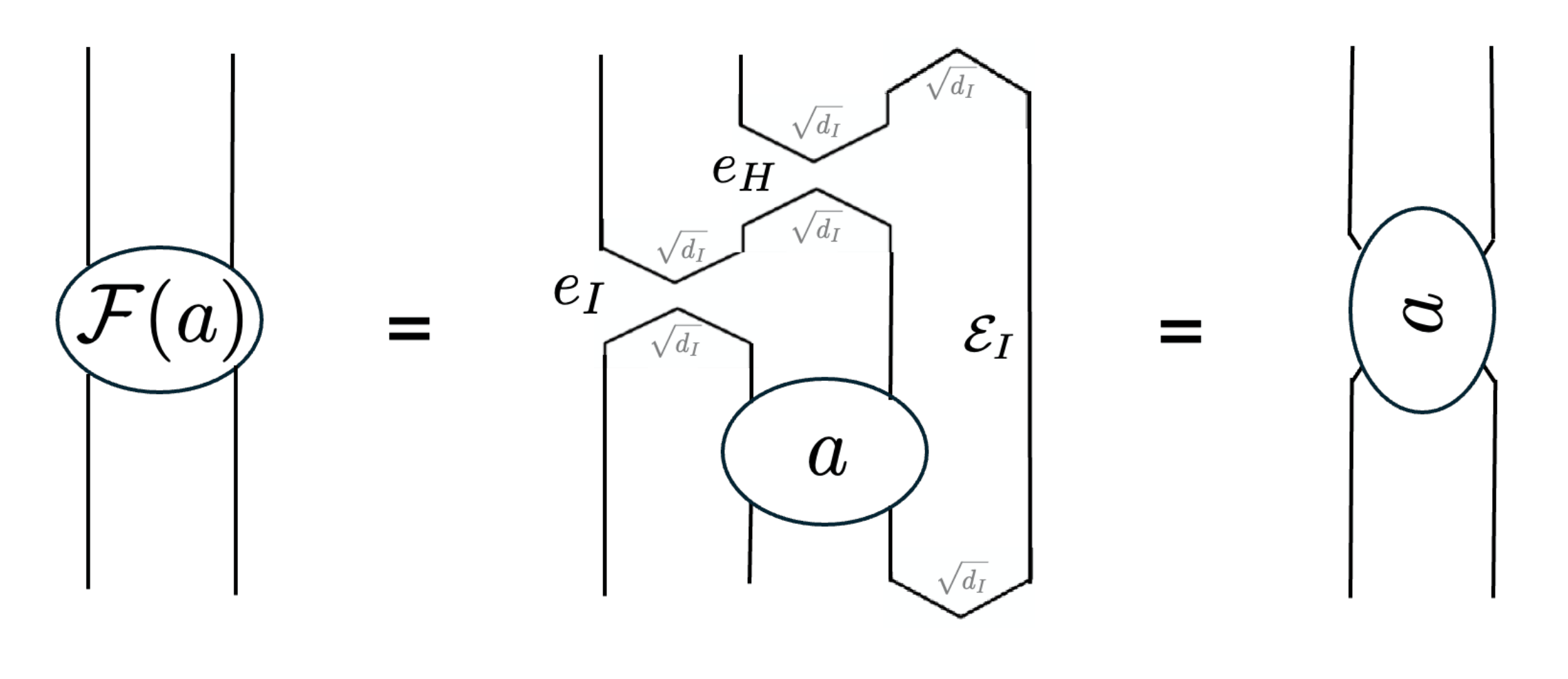}
    \caption{A graphical representation of the Fourier transform. To indicate that the operators $a\in {\cal B}_I$ and ${\cal F}(a)\in \cal B_{H}$ are in the second relative commutant, they are represented graphically by attaching two lines at the bottom and at the top. The first equality gives the definition in terms of the Jones projections and conditional expectation, and includes the normalization factors $\sqrt{d_I}$ (in grey) for each turn in the diagram. The last equality demonstrates the interpretation of the Fourier transform as a $90^\circ$ rotation of the operator $a$.}
    \label{Fig:Fourier}
\end{figure}

The central result of this subsection is that the Fourier transform provides the answer to our algebraic black hole information puzzle. It realizes the map between the interior and radiation representations of the microscopic state: 
\beq
\label{Fourier-transfer}
\bigl|\Phi_H\bigr\rangle = \bigl|\mathcal{F}(\Phi_I)\bigr\rangle~,\qquad
\bigl|\Phi_I\bigr\rangle = \bigl|\mathcal{F}^{-1}(\Phi_H)\bigr\rangle~.
\eeq
We will now show that the inverse transform $\mathcal{F}^{-1}$ decomposes into precisely the three steps of the spacetime evaporation protocol of section~\ref{Bulk decomposition of the microscopic evolution}: an isometric embedding into the island Hilbert space, an entanglement swap, and a non-isometric projection back to the microscopic Hilbert space:
\beq
\label{Inverse-Fourier-transfer}
 \bigl|\Phi_I\bigr\rangle =  \bigl |{\cal F}^{-1}(\Phi_H) \bigr\rangle = \left(V_I^\dagger \circ S_I \circ V_H\right)|\Phi_H\rangle ~.
\eeq
The decomposition of the inverse Fourier transform in these three steps is graphically represented in Figure~\ref{Fig:Fourier-decompose}. We will now discuss each of the steps in more detail.

\begin{figure}[ht]
    \centering
    \includegraphics[width=0.9\linewidth]{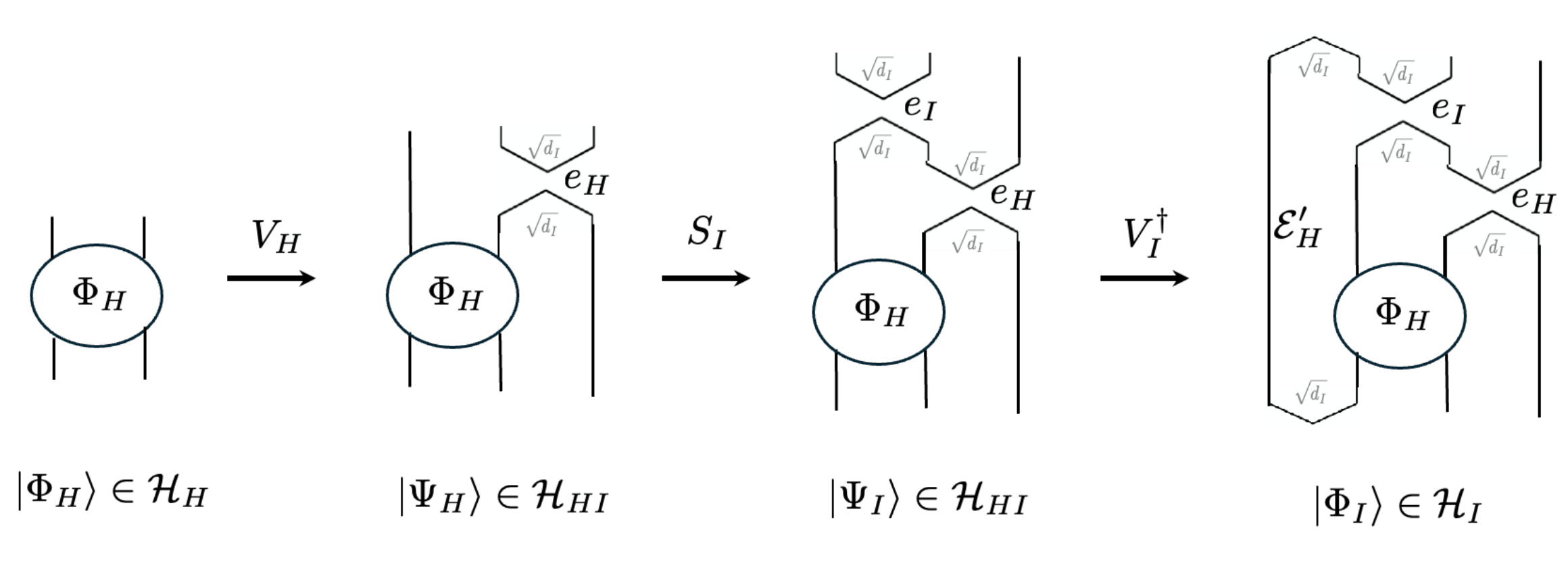}
    \caption{Decomposition of the inverse Fourier transform in terms of (non)-isometries and the entanglement swap illustrating the equivalence with the information transfer via the black hole evaporation process.}
    \label{Fig:Fourier-decompose}
\end{figure}

\paragraph{Step 1: Isometric maps.} The extended Hilbert space is given by the GNS Hilbert space $\mathcal{H}_{HI}$ of the island algebra $\mathcal{B}_{HI}$, defined on its tracial state. Its dimension is $d_I^4$, the square of the Jones index. The Jones projections $e_I$ and $e_H$ project $\mathcal{H}_{HI}$ onto two subspaces of dimension $d_I^2$, which we identify with the images of $\mathcal{H}_I$ and $\mathcal{H}_H$ under the isometries
\beq
V_I: {\cal H}_I\to {\cal H}_{HI}~, \qquad\quad \mbox{and}\qquad \quad V_H:{\cal H}_H\to {\cal H}_{HI}~.
\eeq
In a similar way as in equations \eqref{PsiH} and \eqref{PsiI}, we identify the spacetime states $|\Psi_H\rangle$ and $|\Psi_I\rangle$ with the images of the microscopic states $|\Phi_I\rangle$ and $|\Phi_H\rangle$ under these isometries, 
\beq|\Psi_H\rangle= V_H|\Phi_H\rangle~, \qquad \mbox{and}\qquad |\Psi_I\rangle=V_I|\Phi_I\rangle~.
\eeq
The explicit expression for these states is uniquely fixed by the condition that $|\Psi_H\rangle$ and $|\Psi_I\rangle$ must lie in the image of the Jones projections $e_H$ and $e_I$, and by requiring that these states be properly normalized. This leads to 
\beq
\label{PsiPhi}
|\Psi_I\rangle 
= d_I |e_I \Phi_I\rangle~,  
\qquad \mbox{and}\qquad |\Psi_H\rangle  
= d_I |e_H \Phi_H\rangle~. 
\eeq 
Note that, in contrast to their microscopic counterparts $|\Phi_H\rangle$ and $|\Phi_I\rangle$, the two spacetime states $|\Psi_H\rangle$ and $|\Psi_I\rangle$ are part of the same GNS Hilbert spaces ${\cal H}_{HI}$. 
In the second equation, one recognizes the first step depicted in Figure~\ref{Fig:Fourier-decompose}. 

\paragraph{Step 2: Entanglement swap.}
The action of the entanglement swap operators $S_I$ and $S_H$ in the GNS Hilbert space ${\cal H}_{HI}$ are the same as we found in section~\ref{sec:Information Transfer and Algebraic Complexity} with equations (\ref{SI}) and (\ref{SH}). This step is easily recognized in the second step depicted in Figure~\ref{Fig:Fourier-decompose} as realizing the ``teleportation'' of interior information.

\paragraph{Step 3: Non-isometric maps.} Finally, to explain the third step, we need to define the action of the non-isometries 
\beq
V^\dagger _I: {\cal H}_{HI}\to {\cal H}_{I} \qquad\quad \mbox{and}\qquad \quad V_H^\dagger :{\cal H}_{HI}\to {\cal H}_{H}~.
\eeq
These non-isometric maps reverse the isometries $V_I$ and $V_H$ and thus return the spacetime states $|\Psi_I\rangle$ and $|\Psi_H\rangle$ back to their corresponding microscopic states $|\Phi_I\rangle$ and $|\Phi_H\rangle$. We claim that this is achieved by defining $V_I^\dagger$ and $V_H^\dagger$ in terms of the conditional expectations ${\cal E}_I$ and ${\cal E}'_H$ via the relations 
\beq
\label{nonisom}
V^\dagger_I|\Psi_I\rangle 
= d_I \bigl|{\cal E}'_H (\Psi_I)\bigr\rangle 
\qquad \mbox{and}\qquad V^\dagger_H|\Psi_H\rangle 
= d_I \bigl|{\cal E}_I (\Psi_H)\bigr\rangle~.
\eeq 
Before verifying that these non-isometries indeed have the right properties, we observe that, together with the first two steps of the evaporation protocol, they indeed complete the decomposition of the Fourier transform and its inverse.

Our remaining task is to show that the right-hand side of (\ref{nonisom}) gives back the microscopic states $|\Phi_I\rangle$ and $|\Phi_H\rangle$. For this, we need to evaluate the right-hand side by inserting the operator expressions $\Psi_I = d_I e_I \Phi_I$ and $\Psi_H= d_Ie_H\Phi_H$ that follow from (\ref{PsiPhi}). A straightforward calculation then gives
\beq
d^2_I {\cal E}'_H (e_I\Phi_I) =  \Phi_I   \qquad \quad \mbox{and}\qquad\qquad   d^2_I {\cal E}_I (e_H\Phi_H) = \Phi_H~. 
\eeq 
These equations are most easily derived using the graphical representations of the conditional expectations and Jones projections. 
The results indeed show that the non-isometries (\ref{nonisom}) act in the way they should
\beq
\label{nonisom-check}
V^\dagger_I|\Psi_I\rangle 
= \bigl|\Phi_I\bigr\rangle
\qquad \mbox{and}\qquad V^\dagger_H|\Psi_H\rangle 
= \bigl|\Phi_H\bigr\rangle~.
\eeq 
This completes our proof of our main results (\ref{Inverse-Fourier-transfer}), (\ref{Fourier-transfer}) and (\ref{FourierDef}) and their equivalence to the algebraic description of the black hole evaporation process in spacetime.

In our description of the inverse Fourier transform ${\cal F}^{-1}$, we used the probabilistic version $S_I$ of the entanglement swap. 
The Fourier transform ${\cal F}$ has a similar description in terms of the subsequent application of $V_I$, $S_H$ and $V^\dagger_H$, and simply inverts each of the three steps in ${\cal F}^{-1}$. The probabilistic entanglement swaps $S_I$ and $S_H$ can also be replaced by the deterministic version $S_G$ given in terms of the Grover-search. This makes clear that the Fourier transform and its inverse also carry a high computational complexity that is quantified by the Jones index, following the results of section \ref{sec:Information Transfer and Algebraic Complexity}.

\subsection{Information recovery and black hole complementarity}

The microscopic state $|\Phi\rangle$ has two distinct representatives $|\Phi_I\rangle$ and $|\Phi_H\rangle$ in the GNS Hilbert spaces associated with the $a$- and $b$-modes. The quantum information of Alice's diary is contained in an easily accessible form in the state $|\Phi_I\rangle$ and may be read from the expectation values of the operators $a \in {\cal B}_I$. 
This information is not easily accessible to Bob, who only gets access to the state $|\Phi_H\rangle$. This means he can easily find the expectation values of operators $b\in {\cal B}_H$.

We have also just shown that the state $|\Phi_H\rangle$ is obtained from $|\Phi_I\rangle$ via the Fourier transform ${\cal F}$. The Fourier transform also maps operators $a \in {\cal B}_I$ to operators ${\cal F}(a)\in {\cal B}_H$. Does this mean that Bob can read the information about Alice's diary? In other words, is Bob able to determine the expectation values of the operators $a \in {\cal B}_I$? To do so, he needs to undo the Fourier transform ${\cal F}$. This requires a complex quantum operation that involves an entanglement swap similar to the black hole evaporation process and in essence, amounts to an implementation of the Hayden--Preskill protocol. 

For this recovery process to succeed, the black hole and radiation system must be after Page time, so that Bob can apply the required quantum operations. Even after Page time, the quantum information of Alice's diary has two complementary representations that are related by a complex operation: it either is easily accessible in terms of the $a$-operators or it is represented in a highly complex way in terms of $b$-operators. 
This leads to an algebraic explanation of the notion of black hole complementarity. 

We will now describe the complex recovery process that Bob needs 
to perform to recover the information contained in Alice's diary. First, he needs to identify the operator 
\begin{equation}
{\cal F}(a) \in {\cal B}_H\qquad \qquad \mbox{for}\qquad \qquad  a \in {\cal B}_I~. \\[2mm]
\end{equation}
Suppose Bob knows how to perform this Fourier transform. The next step is to ``act" with the operator ${\cal F}(a)$ on the state $|\Phi_H\rangle \in {\cal H}_H$. The required action is, however, not given by ordinary operator multiplication, but by a new kind of product defined in terms of the Fourier transform $\cal F$. Like the ordinary Fourier transform for functions, the transformation ${\cal F}$ takes the product of two operators $a_1$ and $a_2$ to a convolution product of their Fourier transforms:
\beq
{\cal F}(a_1 a_2) = {\cal F}(a_1)*{\cal F}(a_2)~.
\eeq
This equation defines the convolution product.
We will not attempt to work out a more explicit formula: it requires a combination of the Fourier transform and its inverse, both of which are already quite mathematically involved and computationally complex. 
It is more straightforward to give a graphical representation, as shown in 
Figure~\ref{Fig:convolution-product}. This figure shows that the convolution operator product can be interpreted as a $90^\circ$ rotation of the ordinary operator product. 

\begin{figure}
    \centering
    \includegraphics[width=1\linewidth]{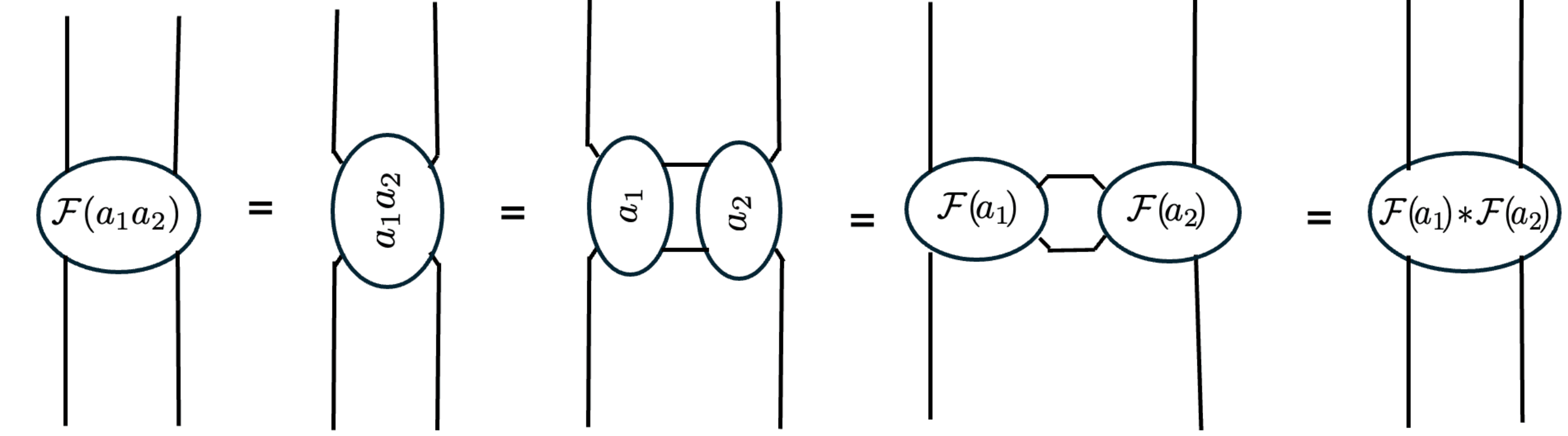}
    \caption{The Fourier transform maps operator multiplication onto a convolution product of the Fourier-transformed operators}
    \label{Fig:convolution-product}
\end{figure}

Suppose now that Bob has also figured out how to implement this convolution product. With this knowledge and his access to the operator ${\cal F}(a)$ he is able to determine the expectation values of the operator $a$ and thus recover the information in Alice's diary. The following formula holds
\beq
{\rm Tr^{}_{}}_H\bigl(\Phi_H^\dagger{}_{{}_{\,}}\! *  {\cal F}(a)\! * \Phi_H\bigr) =\langle\Phi|\,a\,|\Phi\rangle ~.
\eeq 
This identity shows that Bob is in principle able to decode the information if he can identify the operator ${\cal F}(a)$ and take its convolution product with the state $\Phi_H$. These operations are highly complex and involve adding additional auxiliary operators and applying (non-)isometric maps similar to the ones we introduced in our description of the evaporation process. 
They require that Bob has sufficient control over the emitted radiation and can access the entangled partner modes. This requirement is equivalent to imposing that the entanglement entropy of the black hole with the radiation has decreased by a sufficient amount.

\section{Discussion}

The black hole information problem concerns the gap between the description of information release in the underlying microscopic theory and its description in the emergent spacetime. While by now we have developed a good understanding of both perspectives separately, a unified language in which both mechanisms can be described simultaneously is still lacking. In this paper, we aimed to close part of this gap by formulating a model for the black hole information transfer, following the Yoshida--Kitaev protocol, in a type-independent algebraic language, thereby extending the type II$_1$ description developed in \cite{vanderHeijden:2024tdk}. This is another step towards explicitly relating the microscopic and spacetime mechanisms for information transfer.

We showed that the relevant dynamics decomposes into an isometric embedding, an entanglement swap, and a non-isometric map, together with a shift in the subdivision of algebras, i.e. our algebraic analogue of the quantum extremal surface. In the case of an irreducible depth two inclusion, this process can be viewed as an algebraic Fourier transform that relates interior and radiation operators. The evaporation step takes place in the part of the system that we call the extended island algebra. It contains the diary information in the so-called island algebra, a second relative commutant associated to a choice of Pimsner--Popa basis for the inclusion. This definition is rather abstract, and as we have not yet given an explicit spacetime realization it would be interesting to connect this to other approaches for describing entanglement islands in an algebraic language, e.g. \cite{Geng:2025bcb}.

After the evaporation step, the information of the diary is encoded in the radiation degrees of freedom in a non-trivial way and can only be decoded once there is sufficient entanglement between the interior and radiation degrees of freedom. In a follow-up paper \cite{2Tu2025}, we will show how to keep track of changes in the entanglement entropy during an evaporation step, and demonstrate a version of the Page curve in our model. After an evaporation step, Bob can in principle recover Alice’s diary from the island algebra by applying the inverse of the algebraic Fourier transform to the relevant radiation degrees of freedom. We show that the associated complexity of the decoding is quantified by the Jones--Kosaki index of the inclusion. This gives an operational meaning to complexity in terms of Jones projectors that applies to a general finite-index inclusion (going beyond its definition in finite-dimensional models). 

The concept of complementarity has long interacted with that of complexity in the context of black holes. The firewall paradox showed that, after the Page time, semiclassical interior and exterior descriptions imply incompatible entanglement patterns, threatening entanglement monogamy. In qubit toy models, complexity then entered as a “safety barrier”: the operations needed to reveal this incompatibility are exponentially hard in the black hole entropy, so within the regime of low-complexity experiments the two descriptions never visibly conflict \cite{akers2022black}. When that barrier breaks down in late evaporation, complexity alone no longer protects complementarity, motivating a new operational version of complementarity \cite{engelhardt2025observer}. In our model, this connection becomes explicit: the same quantum information admits two complementary representations: a simple one in terms of interior operators and a complicated one in terms of radiation operators, which are related by a transformation of very high algebraic complexity, the algebraic Fourier transform, rather than by a simple local unitary. We interpret this as an algebraic realization of black hole complementarity in our model.

Our approach does have important limitations. Not all von Neumann algebra inclusions admit conditional expectations with finite index, and this is a rather strong assumption\footnote{Recent work has explored algebra extensions for certain generalized notions of conditional expectations, in which some of the defining properties are relaxed \cite{AliAhmad:2024saq, AliAhmad:2025oli, Klinger:2026kqj}. It would be interesting to investigate whether there is a version of our protocol in terms of generalized conditional expectations.}; we motivated it physically by requiring that a finite amount of information can be extracted from the black hole in a single evaporation step. However, a naive application of the protocol to an inclusion of type III$_1$ von Neumann algebras associated to spacetime subregions in QFT fails, because these algebras generally do not admit conditional expectations of finite index. It is possible, however, to extend elements of our protocol even to the case of an infinite index inclusion, see e.g. \cite{HermanOcneanu1989, Nill:1994be, EnockNest1996}. In this approach the notion of a conditional expectation gets replaced by that of an operator-valued weight, and for the case of an irreducible depth two inclusion there is still a well-defined Jones tower, canonical endomorphism and algebraic Fourier transform. The main issue seems to be the entanglement swap, which is crucial for the information transfer. For an infinite-index inclusion, the algebraic complexity barrier associated to this operation goes to infinity preventing us from a naive implementation. Without the entanglement swap in place, it seems that both the black hole and radiation subsystem keep growing with each evaporation step. One could interpret this failure of our protocol as a manifestation of the black hole information problem. 

Our discussion has remained at the level of abstract operator algebras and has not addressed in detail how this framework is mapped onto a specific spacetime geometry. Recent investigations into the algebraic structure of quantum gravity reveal that such a mapping could be rather non-trivial (see e.g. \cite{Sahu:2025upe}). The abstract discussion has the benefit that some of the algebraic tools developed here have the potential of being more widely applicable. For example, we expect that the discussion of complementarity and complexity is relevant for understanding the existence of the baby universe in the Antonini-Sasieta-Swingle (AS$^2$) model \cite{antonini2023cosmology}. This will be the subject of future work.

\section*{Acknowledgements}
We would like to thank Elliott Gesteau, Misha Isachenkov, Marc Klinger, Bahman Najian, Antony Speranza and Herman Verlinde for useful discussions. AC is supported by the Heising-Simons Foundation ‘Observational Signatures
of Quantum Gravity’ QuRIOS collaboration grant. JT, PvdH and EV are supported by the Spinoza grant and the Delta ITP consortium, both of which are programs of the Netherlands Organisation for Scientific Research (NWO) that is funded by the Dutch Ministry of Education, Culture and Science (OCW). JvdH is supported in part by the National Science and Engineering Research Council of Canada (NSERC) and the Simons Foundation via a Simons Investigator Award. 

\appendix

\section{Mathematical background}
\label{app:math_background}

In this appendix, we collect several operator algebraic constructions used in the main text. Our aim is to fix the notation and explain the relations among conditional expectations, Jones projections, canonical endomorphisms, the Jones tower, finite depth, and Pimsner--Popa bases.

Throughout this appendix, we consider an inclusion
\begin{equation}
    \mathcal A_B\subset\mathcal A_{BH}
    \label{eq:appendix-AB-ABH-inclusion}
\end{equation}
of properly infinite von Neumann factors.\footnote{This includes von Neumann algebras of type I$_{\infty}$, II$_{\infty}$ and type III.} The algebras are assumed to act on a common Hilbert space
\(\mathcal H_\Phi\). The corresponding commutant
inclusion is
\begin{equation}
    \mathcal A_{BH}'\subset\mathcal A_B'~.
\end{equation}
Within the Hilbert space $\mathcal{H}_{\Phi}$, we identify
\begin{equation}
    \mathcal A_R:=\mathcal A_{BH}'~,
    \qquad
    \mathcal A_{RH}:=\mathcal A_B'~.
    \label{eq:appendix-commutant-dictionary}
\end{equation}

We assume that there exists a vector
\(\lvert\Phi\rangle\in\mathcal H_\Phi\) that is cyclic and separating for
both \(\mathcal A_B\) and \(\mathcal A_{BH}\). We use the same symbol
\(\Phi\) for the associated faithful normal vector state,
\begin{equation}
    \Phi(x)
    :=
    \langle\Phi|x|\Phi\rangle~.
\end{equation}
Whenever the Jones tower, finite depth, or Pimsner--Popa bases are discussed, we additionally assume that the inclusion has finite index.

\subsection{Standard inclusions and modular conjugations}
\label{app:standard-inclusions}

\begin{definition}[Common standard vector]
\label{def:common-standard-vector}
A vector \(\lvert\Phi\rangle\in\mathcal H_\Phi\) is a common standard vector
for the inclusion
\begin{equation}
    \mathcal A_B\subset\mathcal A_{BH}
\end{equation}
if it is cyclic and separating for both algebras. We denote the associated
modular conjugations by
\begin{equation}
    J_B^\Phi~,
    \qquad
    J_{BH}^\Phi~.
\end{equation}
\end{definition}

The Tomita--Takesaki theorem gives the relations
\begin{equation}
    J_B^\Phi\mathcal A_BJ_B^\Phi
    =
    \mathcal A_B'
    =
    \mathcal A_{RH}~,
    \label{eq:JB-commutant}
\end{equation}
and
\begin{equation}
    J_{BH}^\Phi\mathcal A_{BH}J_{BH}^\Phi
    =
    \mathcal A_{BH}'
    =
    \mathcal A_R~.
    \label{eq:JBH-commutant}
\end{equation}
Moreover, an algebra and its commutant have the same modular conjugation:
\begin{equation}
    J_{RH}^\Phi
    =
    J_B^\Phi~,
    \qquad
    J_{R}^\Phi
    =
    J_{BH}^\Phi~.
    \label{eq:algebra-commutant-modular-conjugations}
\end{equation}

Although the same vector \(\lvert\Phi\rangle\) is used for both algebras,
the modular conjugations \(J_B^\Phi\) and \(J_{BH}^\Phi\) are generally
different. In particular, \(J_B^\Phi\) is not simply obtained by restricting
\(J_{BH}^\Phi\) to the subalgebra \(\mathcal A_B\). Their composition gives rise to a unitary operator:

\begin{definition}[Canonical shift and canonical endomorphism]
\label{def:canonical-endomorphism}
The \emph{canonical shift operator} is defined as the unitary
\begin{equation}
    \Gamma:=J_{BH}^{\Phi}J_B^{\Phi}~,
    \label{eq:canonical-unitary-direct}
\end{equation}
and the corresponding \emph{canonical endomorphism} is
\begin{equation}
    \gamma:
    \mathcal A_{BH}\longrightarrow\mathcal A_B~,
    \qquad
    \gamma(x):=\Gamma^{\dagger} x\Gamma~.
    \label{eq:canonical-endomorphism-direct}
\end{equation}
\end{definition}

\begin{proposition}
\label{prop:canonical-endomorphism-range-direct}
The map \(\gamma\) is a normal unital injective \(^*\)-homomorphism
satisfying
\begin{equation}
    \gamma(\mathcal A_{BH})
    \subset
    \mathcal A_B.
    \label{eq:gamma-range-direct}
\end{equation}
\end{proposition}

\begin{proof}
Normality, unitality, injectivity, and multiplicativity follow because
\(\gamma\) is implemented by the unitary \(\Gamma\). Moreover, for
\(x\in\mathcal A_{BH}\) we have
\begin{align}
    \gamma(x)
    &=
    J_B^\Phi J_{BH}^\Phi xJ_{BH}^\Phi J_B^\Phi 
    \in
    J_{B}^\Phi\mathcal A_R J_{B}^\Phi~.
\end{align}
Since $\mathcal A_R
    =
    \mathcal A_{BH}'
    \subset
    \mathcal A_B'
    =
    \mathcal A_{RH}$ we have
\begin{equation}
    J_{B}^\Phi\mathcal A_R J_{B}^\Phi
    \subset
    J_{B}^\Phi\mathcal A_{RH}J_{B}^\Phi
    =
    \mathcal A_B~.
\end{equation}
Therefore,
\begin{equation}
    \gamma(\mathcal A_{BH})
    \subset
    \mathcal A_B~.
\end{equation}
\end{proof}
\begin{remark}[Unitary realization]
Suppose that, in a particular realization, there exists a unitary \(U\) such
that
\begin{equation}
    \mathcal A_B=U^\dagger\mathcal A_{BH}U~.
\end{equation}
If the standard forms are chosen compatibly so that
\begin{equation}
    J_{B}^{\Phi}
    =
    U^\dagger J_{BH}^{\Phi}U~,
\end{equation}
then
\begin{equation}
    \Gamma
    =
    J_{BH}^{\Phi}U^{\dagger}J_{BH}^{\Phi}U ~.
\end{equation}
Introducing the conjugate unitary
\begin{equation}
    \overline U
    :=
    J_{BH}^{\Phi}U^\dagger J_{BH}^{\Phi}~,
\end{equation}
one obtains
\begin{equation}
    \Gamma=\overline U U~.
\end{equation}
\end{remark}

\subsection{Conditional expectations and isometric embeddings}
\label{app:conditional-expectation-PsiI-direct}

We assume there exists a faithful normal conditional expectation:\footnote{
Faithfulness ensures that no positive operator is annihilated by the conditional expectation, i.e. $E(x^\dagger x)=0$ implies $x=0$. Normality guarantees compatibility with the von Neumann algebra topology, i.e. the conditional expectation is continuous with respect to increasing limits of bounded positive sequences $x_\alpha\uparrow x$ in the sense that $E(x) = \sup_\alpha E(x_\alpha)$.}
\begin{equation}
    E_I:
    \mathcal A_{BH}\longrightarrow\mathcal A_B
    \label{eq:EI-definition-direct}~.
\end{equation}

\begin{definition}[Conditional expectation]
\label{def:conditional-expectation-direct}
A \emph{conditional expectation}
\begin{equation}
    E_I:
    \mathcal A_{BH}\to\mathcal A_B
\end{equation}
is a unital, completely positive map satisfying
\begin{equation}
    E_I(a_1xa_2)
    =
    a_1E_I(x)a_2~,
    \qquad
    a_1,a_2\in\mathcal A_B~,
    \quad
    x\in\mathcal A_{BH}~.
    \label{eq:EI-bimodule-property-direct}
\end{equation}
\end{definition}

Composing the state \(\Phi\) with \(E_I\), we obtain the state
\begin{equation}
    \Psi_I
    :=
    \Phi\circ E_I
    \label{eq:PsiI-definition-direct}
\end{equation}
on \(\mathcal A_{BH}\). Equivalently,
\begin{equation}
    \Psi_I(x)
    =
    \langle\Psi_I|x|\Psi_I\rangle
    =
    \langle\Phi|E_I(x)|\Phi\rangle~,
    \qquad
    x\in\mathcal A_{BH}~.
    \label{eq:PsiI-vector-state-direct}
\end{equation}
Here, \(\lvert\Psi_I\rangle\) denotes the cyclic GNS vector associated with
the state \(\Psi_I\). Because \(E_I\) acts trivially on \(\mathcal A_B\), the states agree on the
subalgebra:
\begin{equation}
    \Psi_I(a)
    =
    \Phi(a)~,
    \qquad
    a\in\mathcal A_B~.
    \label{eq:PsiI-Phi-agree-AB}
\end{equation}
We denote by
\begin{equation}
    \mathcal H_{\Psi_I}
    :=
    \overline{
        \mathcal A_{BH}|\Psi_I\rangle
    }
\end{equation}
the GNS Hilbert space associated to the state \(\Psi_I\).

\begin{proposition}[Isometric embedding]
\label{prop:VI-isometry-direct}
The map
\begin{equation}
    V_I
    \bigl(
        a|\Phi\rangle
    \bigr)
    :=
    a|\Psi_I\rangle~,
    \qquad
    a\in\mathcal A_B~,
    \label{eq:VI-definition-direct}
\end{equation}
extends uniquely to an isometry
\begin{equation}
    V_I:
    \mathcal H_\Phi
    \longrightarrow
    \mathcal H_{\Psi_I}~.
\end{equation}
Its adjoint satisfies
\begin{equation}
    V_I^{\dagger}x|\Psi_I\rangle
    =
    E_I(x)|\Phi\rangle~,
    \qquad
    x\in\mathcal A_{BH}~,
    \label{eq:VI-adjoint-direct}
\end{equation}
from which it follows that
\begin{equation}
    E_I(x) =V_I^{\dagger}xV_I~,
    \qquad
    x\in\mathcal A_{BH}~.
    \label{eq:EI-compression-direct}
\end{equation}
\end{proposition}

\begin{proof}
For \(a\in\mathcal A_B\), we have
\begin{align}
    \left\|
        a|\Psi_I\rangle
    \right\|^2
    &=
    \Psi_I(a^*a)
    =
    \Phi\bigl(E_I(a^*a)\bigr)
    \nonumber\\
    &=
    \Phi(a^*a)
    =
    \left\|
        a|\Phi\rangle
    \right\|^2.
\end{align}
Hence, \(V_I\) is isometric. To prove the second statement, we note that for \(a\in\mathcal A_B\) and \(x\in\mathcal A_{BH}\),
\begin{align}
    \langle\Phi|
    a^{*}V_I^{\dagger}x
    |\Psi_I\rangle
    &=
    \langle\Psi_I|a^{*}x|\Psi_I\rangle
    \nonumber\\
    &=
    \Phi\bigl(E_I(a^{*}x)\bigr)
    \nonumber\\
    &=
    \Phi\bigl(a^{*}E_I(x)\bigr)
    \nonumber\\
    &=
    \langle\Phi|a^{*}E_I(x)|\Phi\rangle~.
\end{align}
Since \(\mathcal A_B|\Phi\rangle\) is dense in
\(\mathcal H_\Phi\), the result follows.
\end{proof}

\begin{definition}[Jones projection]
\label{def:Jones-projection-direct}
The \emph{Jones projection} associated with the conditional expectation \(E_I\) is given by
\begin{equation}
    e_I:=V_IV_I^\dagger~.
    \label{eq:eI-definition-direct}
\end{equation}
It is the orthogonal projection onto
\begin{equation}
   e_{I}: \mathcal H_{\Psi_I} \to  V_I\mathcal H_\Phi =\overline{
        \mathcal A_B|\Psi_I\rangle
    }
    ~.
\end{equation}
\end{definition}

For every \(x\in\mathcal A_{BH}\), the Jones projection satisfies
\begin{equation}
    e_Ix|\Psi_I\rangle
    =
    E_I(x)|\Psi_I\rangle~,
\end{equation}
which we can equivalently write as the following identity:
\begin{equation}
    e_Ixe_I
    =
    E_I(x)e_I~.
    \label{eq:Jones-relation-direct}
\end{equation}

On the radiation side, we can do a similar construction. We have the following commutant inclusion
\begin{equation}
    \mathcal A_R
    \subset
    \mathcal A_{RH}~,
\end{equation}
where we have identified $\mathcal A_R
    = \mathcal A_{BH}',
    \mathcal A_{RH}
=\mathcal A_B'$ within $\mathcal{H}_{\Phi}$. We assume that there exists a faithful normal conditional expectation
\begin{equation}
    E_H':
    \mathcal A_{RH}
    \longrightarrow
    \mathcal A_R~.
    \label{eq:EH-definition-direct}
\end{equation}
Let \(\Phi'\) denote the restriction of the vector state \(|\Phi\rangle\) to
\(\mathcal A_R\). We define
\begin{equation}
\Psi_H
    :=
    \Phi'\circ E_H'~.
    \label{eq:PsiH-definition-direct}
\end{equation}
Equivalently,
\begin{equation}
    \Psi_H(x)
    =
    \langle\Phi|E_H'(x)|\Phi\rangle~,
    \qquad
    x\in\mathcal A_{RH}~.
    \label{eq:PsiH-vector-state-direct}
\end{equation}
In the same way as for $E_I$, one can construct an isometry
\begin{equation}
V_H: \mathcal{H}_{\Phi} \to \mathcal{H}_{\Psi_H}~,
\end{equation}
and corresponding Jones projector $e_H  =  V_{H} V_H^{\dagger}$.

\subsection{The Jones basic construction}
\label{app:basic-construction-direct}

We now give a definition of the Jones basic construction. The idea is that one can use the Jones projector $e_I$ to extend the algebra inclusion $\mathcal{A}_B\subset \mathcal{A}_{BH}$:

\begin{definition}[Basic construction]
\label{def:basic-construction-direct}
The \emph{Jones basic construction} associated with the inclusion 
\(\mathcal A_B\subset\mathcal A_{BH}\) is the von Neumann algebra generated by $\mathcal{A}_{BH}$ and $e_I$:
\begin{equation}
    \mathcal B_{BH}
    :=
    \langle
        \mathcal A_{BH},e_I
    \rangle~.
    \label{eq:basic-construction-direct}
\end{equation}
\end{definition}

It is a standard result that one can equivalently express the basic construction as
\begin{equation}
\mathcal B_{BH} = J_{BH}^{\Phi} \mathcal A_{RH} J_{BH}^{\Phi}~.
\end{equation}
We have:
\begin{equation}
    \mathcal B_{BH} = J_{BH}^{\Phi} \mathcal A_{RH} J_{BH}^{\Phi}= J_{BH}^{\Phi} J_{B}^{\Phi}\mathcal A_{B} J_{B}^{\Phi} J_{BH}^{\Phi} =
    \gamma^{-1}(\mathcal A_B)~,
    \label{eq:basic-construction-gamma-direct}
\end{equation}
where the inverse image is taken under the canonical endomorphism.\footnote{In writing the above expressions, we have implicitly assumed that the basic extension algebra $\mathcal{B}_{BH}$ is represented on the same Hilbert space $\mathcal{H}_{\Phi}$.}

More generally, repeated applications of \(\gamma\) and
\(\gamma^{-1}\) define a bi-infinite family of von Neumann algebras,
\begin{equation}
    \mathcal A_{2k-1}
    =
    \gamma^{-k}(\mathcal A_B)~,
    \qquad
    \mathcal A_{2k}
    =
    \gamma^{-k}(\mathcal A_{BH})~,
    \qquad
    k\in\mathbb Z~.
    \label{eq:Ak-definition}
\end{equation}
In particular,
\begin{equation}
    \mathcal A_{-1}
    =
    \mathcal A_B,
    \qquad
    \mathcal A_0
    =
    \mathcal A_{BH}~,
    \qquad
    \mathcal A_1
    =
    \mathcal B_{BH}~.
\end{equation}

\begin{definition}[Jones tower/tunnel]
These algebras form an infinite sequence of inclusions:
\begin{equation}
    \cdots
    \subset
    \mathcal A_{-2}
    \subset
    \mathcal A_{-1}
    \equiv
    \mathcal A_B
    \subset
    \mathcal A_0
    \equiv
    \mathcal A_{BH}
    \subset
    \mathcal A_1
    \equiv
    \mathcal B_{BH}
    \subset
    \mathcal A_2
    \subset
    \cdots~.
    \label{eq:Jones-tower-tunnel}
\end{equation}
The increasing part of this sequence is called the \emph{Jones tower}, while
the decreasing part is called the \emph{Jones tunnel}.
\end{definition}

For every \(k\in\mathbb Z\), the basic construction gives rise to a faithful normal conditional
expectation
\begin{equation}
    E_k:
    \mathcal A_k
    \longrightarrow
    \mathcal A_{k-1}~,
    \label{eq:Ek-definition}
\end{equation}
with $E_0:=E_I$ and Jones projector $e_0: =e_I$. The corresponding Jones projection \(e_k\) belongs to
\(\mathcal A_{k+1}\) and satisfies
\begin{equation}
    e_kxe_k
    =
    E_k(x)e_k~,
    \qquad
    x\in\mathcal A_k~.
    \label{eq:ek-Jones-relation}
\end{equation}
Moreover,
\begin{equation}
    \mathcal A_{k+1}
    =
    \langle
        \mathcal A_k,e_k
    \rangle~.
    \label{eq:Ak-basic-construction}
\end{equation}
For future reference, we denote the conditional expectation $E_1$ in the Jones tower by
\begin{equation}
E_H: \mathcal{B}_{BH}\to \mathcal{A}_{BH}~.
\end{equation}

By construction, every consecutive inclusion in the tower has the same
Jones index as the original inclusion,
\begin{equation}
    [\mathcal A_{k+1}:\mathcal A_k]
    =
    [\mathcal A_{BH}:\mathcal A_B]~.
    \label{eq:index-constant-tower}
\end{equation}
For now we state the useful Jones relation:
\begin{equation}
E_{k+1}(e_k)=[\mathcal A_{BH}:\mathcal A_B]^{-1}\,\mathbf{1}~.
\end{equation}
By definition, the canonical endomorphism shifts the tower by two steps,
\begin{equation}
    \gamma(\mathcal A_k)
    =
    \mathcal A_{k-2}~.
    \label{eq:gamma-shifts-tower}
\end{equation}

On the radiation side, we can proceed in a similar fashion and apply the basic construction to the inclusion $\mathcal{A}_R\subset \mathcal{A}_{RH}$ to obtain the basic extension:
\begin{equation}
\mathcal{B}_{RH}:=\langle \mathcal{A}_{RH}, e_H \rangle~.
\end{equation}
Repeating the construction one obtains the sequence
\begin{equation}
    \cdots
    \subset
    \overline{\mathcal A}_{-2}
    \subset
    \overline{\mathcal A}_{-1}
    \equiv
    \mathcal A_R
    \subset
    \overline{\mathcal A}_0
    \equiv
    \mathcal A_{RH}
    \subset
    \overline{\mathcal A}_1
    \equiv
    \mathcal B_{RH}
    \subset
    \overline{\mathcal A}_2
    \subset
    \cdots~,
\end{equation}
with conditional expectations
\begin{equation}
E_k':\overline{\mathcal{A}}_{k}\to \overline{\mathcal{A}}_{k-1~,}
\end{equation}
and corresponding Jones projectors $e'_k\in \overline{\mathcal{A}}_{k+1}$ with $E'_0 = E_H'$ and $e'_0=e_H$. We denote the conditional expectation $E'_1$ by
\begin{equation}
E_I':\mathcal{B}_{RH}\to \mathcal{A}_{RH}~.
\end{equation}
Within the Hilbert space $\mathcal{H}_{\Phi}$ we now have the following relations: 
\begin{equation}
E_I' = \mathrm{Ad}J_{B}^{\Phi} \circ E_I \circ \mathrm{Ad}  J_{B}^{\Phi}~,\qquad E_H' = \mathrm{Ad}J_{BH}^{\Phi} \circ E_H \circ \mathrm{Ad}  J_{BH}^{\Phi}~.
\end{equation}

\subsection*{The canonical endomorphism in terms of $\Psi_I$ and $\Psi_H$}

To compare with the notation of the main text, we now represent the canonical endomorphism on the effective ``spacetime'' Hilbert space. 

Note that the state \(\lvert\Psi_I\rangle\) is cyclic and separating for
\(\mathcal A_{BH}\). We have 
\begin{equation}
\mathcal{H}_{\Psi_I}:= \overline{\mathcal{A}_{BH}|\Psi_I\rangle}~.
\end{equation}
We denote the modular conjugation for the algebra $\mathcal{A}_{BH}$ with respect to $|\Psi_I\rangle$ by
\begin{equation}
   J_I \equiv  J_{BH}^{\Psi_I}~.
\end{equation}
Similarly, the state
\(\lvert\Psi_H\rangle\) is cyclic and separating for
\(\mathcal A_{RH}\) in its GNS representation 
\begin{equation}
\mathcal{H}_{\Psi_H}:=\overline{\mathcal{A}_{RH}|\Psi_H\rangle}~.
\end{equation}
Denote by
\begin{equation}
    J_H \equiv J_{RH}^{\Psi_H}~
\end{equation}
the corresponding modular conjugation associated to the algebra $\mathcal{A}_{RH}$ within $\mathcal{H}_{\Psi_H}$.

The modular conjugations \(J_I\) and
\(J_H\) initially act on different Hilbert spaces. It will be useful to introduce a common Hilbert space on which they are represented. Let us use the canonical endomorphism to introduce the state:
\begin{equation} \label{eq:states_gamma}
\widetilde{\Psi}_{H} := \Psi_H\circ \gamma~.
\end{equation}
By writing out the definition of $\widetilde{\Psi}_H$, we obtain 
\begin{equation}
\widetilde{\Psi}_H = \Phi\circ E_H\circ \mathrm{Ad}( J_{BH}^{\Phi}J_{B}^{\Phi}J_{BH}^{\Phi})~,
\end{equation}
which shows that $\widetilde{\Psi}_H$ is obtained by extending the state $\Phi$ via the conditional expectation 
\begin{equation}
E_H:\mathcal{B}_{BH}\to \mathcal{A}_{BH}~,
\end{equation}
and is furthermore composed with the modular conjugation $J_{BH}^{\Phi}J_{B}^{\Phi}J_{BH}^{\Phi}$ to ensure that the state is defined on the commutant algebra $\mathcal{B}_{BH}'$. We can thus write the Hilbert space associated to the state $\widetilde{\Psi}_H$ in terms of the algebra $\mathcal{B}_{BH}$ as
\begin{equation}
\mathcal{H}_{\widetilde{\Psi}_H}:= \overline{\mathcal{B}_{BH}'|\widetilde{\Psi}_H\rangle} = \overline{\mathcal{B}_{BH}|\Phi \circ E_H\rangle}~.
\end{equation}
The relation 
\eqref{eq:states_gamma} gives us the following isomorphism:
\begin{equation}
\Gamma: \mathcal{H}_{\widetilde{\Psi}_H} \xrightarrow{\sim} \mathcal{H}_{\Psi_H}~.
\end{equation}
In the main text, whenever we mention the effective spacetime Hilbert space we shall mean the Hilbert space $\mathcal{H}_{\widetilde{\Psi}_H}$, but we will often denote it by $\mathcal{H}_{\Psi_H}$ with the above isomorphism understood. In particular, $J_H$ will denote the modular conjugation for the algebra $\mathcal{B}_{BH}$ with respect to the state $\widetilde{\Psi}_H$. In addition, since $\mathcal{A}_{BH}\subset \mathcal{B}_{BH}$ we can naturally identify the Hilbert space $\mathcal{H}_{\Psi_I}$ with the subspace $\mathcal{H}_{\Psi_I}\subset \mathcal{H}_{\widetilde{\Psi}_H}$ under the projector $e_I$.

Let us denote by $\mathcal{B}_B$ and $\mathcal{B}_{R}$ the algebras $\mathcal{A}_{BH}$ and $\mathcal{A}_{RH}$ when they are acting on the common Hilbert space $\mathcal{H}_{\widetilde{\Psi}_H}$. In this representation, we thus have:
\begin{equation}
J_I\mathcal{B}_BJ_I=\mathcal{B}'_B =\mathcal{B}_{RH}~, \qquad J_H\mathcal{B}_RJ_H=\mathcal{B}_R'=\mathcal{B}_{BH}~.
\end{equation}
Note that the modular conjugations can be identified with
\begin{equation}
J_I \simeq  J_{BH}^{\Phi}~, \qquad J_{H}\simeq J_{BH}^{\Phi}J_{B}^{\Phi}J_{BH}^{\Phi}~.
\end{equation}
Consequently, on the common representation the canonical shift operator may be expressed in terms of the modular data associated with \(\Psi_I\) and \(\Psi_H\) as
\begin{equation}
    \Gamma
    \simeq
    J_H J_I~.
    \label{eq:Gamma-PsiI-PsiH-direct}
\end{equation}

\subsection{Higher relative commutants and finite depth inclusions}

We now study the higher relative commutants:

\begin{definition}[First derived tower]
Given the Jones tower, we define the \emph{first derived tower} by
\begin{equation}
    \mathcal Y_k
    :=
    \mathcal A_B'
    \cap
    \mathcal A_{k-1}~,
    \qquad
    k\geq0~.
    \label{eq:Yk-definition}
\end{equation}
The algebra \(\mathcal Y_k\) consists of all operators in
\(\mathcal A_{k-1}\) that commute with \(\mathcal A_B\).
\end{definition}
The first term in the derived tower is
\begin{equation}
    \mathcal Y_0
    =
    \mathcal A_B'
    \cap
    \mathcal A_B
    =
    \mathbb C1~,
\end{equation}
since we have assumed that \(\mathcal A_B\) is a factor. The second term $\mathcal Y_1$ is given by the first relative commutant, which we denote by
\begin{equation}
  \mathcal A_H
    :=
    \mathcal A_B'
    \cap
    \mathcal A_{BH}
    ~.
    \label{eq:AI-relative-commutant}
\end{equation} 
\begin{definition}[Irreducible inclusion]
\label{def:irreducible-inclusion}
The inclusion \(\mathcal A_{B}\subset\mathcal A_{BH}\) is \emph{irreducible} if
\begin{equation}
    \mathcal A_H = \mathcal A_B'\cap\mathcal A_{BH}
    =
    \mathbb C1~.
\end{equation}
Otherwise, the inclusion is \emph{reducible}.
\end{definition}

For a finite-index
inclusion of factors, each \(\mathcal Y_k\) is finite-dimensional and
\begin{equation}
\dim_{\mathbb{C}}(\mathcal{Y}_{k}) \leq [\mathcal{A}_{BH}:\mathcal{A}_B]^{k}~.
\end{equation}
The conditional expectations in the Jones tower restrict to conditional
expectations between the relative commutants,
\begin{equation}
    E_{k-1}\big|_{\mathcal Y_k}:
    \mathcal Y_k
    \longrightarrow
    \mathcal Y_{k-1}~.
    \label{eq:restricted-expectation-Y}
\end{equation}
The inclusions
\begin{equation}
    \mathcal Y_{k-1}
    \subset
    \mathcal Y_k
    \subset
    \mathcal Y_{k+1}
\end{equation}
therefore inherit a compatible sequence of finite-dimensional
conditional expectations.

\begin{definition}[Finite depth]
\label{def:finite-depth}
The inclusion
\begin{equation}
    \mathcal A_B
    \subset
    \mathcal A_{BH}
\end{equation}
is said to have \emph{finite depth} if there exists an integer \(k\geq1\) such
that
\begin{equation}
    \mathcal Y_{k-1}
    \subset
    \mathcal Y_k
    \subset
    \mathcal Y_{k+1}
    \label{eq:finite_depth1}
\end{equation}
is an instance of the Jones basic construction. Equivalently,
\begin{equation}
    \mathcal Y_{k+1}
    =
    \langle
        \mathcal Y_k,e_k
    \rangle~,
    \label{eq:Y-basic-construction}
\end{equation}
where the relevant conditional expectations and Jones projections are
obtained by restricting those of the original Jones tower. The smallest such \(k\) is called the \emph{depth} of the inclusion.
\end{definition}

At and beyond the depth, the inclusions in the derived tower have the
same index as the original inclusion,
\begin{equation}
    [\mathcal Y_k:\mathcal Y_{k-1}]
    =
    [\mathcal A_{BH}:\mathcal A_B]~.
    \label{eq:finitedepth2}
\end{equation}
If this condition holds for some \(k\), then the corresponding basic
construction and index relations continue to hold for all \(n\geq k\).

\subsubsection*{Depth one inclusions}

Suppose first that the inclusion has depth one. Then
\begin{equation}
    \mathbb C1
    =
    \mathcal A_B'
    \cap
    \mathcal A_B
    \subset
    \mathcal A_B'
    \cap
    \mathcal A_{BH}
    \subset
    \mathcal A_B'
    \cap
    \mathcal B_{BH}
    \label{eq:depth-one-basic-construction}
\end{equation}
is a basic construction. We thus get that the corresponding index equality is
\begin{equation}
    [
        \mathcal A_H
        :
        \mathbb C1
    ]
    =
    [
        \mathcal A_{BH}
        :
        \mathcal A_B
    ]~,
    \label{eq:depth-one-AI-index}
\end{equation}
where $\mathcal A_H
    =
    \mathcal A_B'
    \cap
    \mathcal A_{BH}$. For the trace-preserving conditional expectation,
the index of the finite-dimensional inclusion
\(\mathbb C1\subset\mathcal A_H\) agrees with the complex vector space
dimension of \(\mathcal A_H\). Hence, we have
\begin{equation}
    \dim_{\mathbb C}(\mathcal A_H)
    =
    [
        \mathcal A_{BH}
        :
        \mathcal A_B
    ]~.
    \label{eq:depth-one-dimension}
\end{equation}

\subsubsection*{Depth two inclusions}

We write
\begin{equation}
\mathcal{B}_{H} := \mathcal{A}_B'\cap \mathcal{B}_{BH}~, \qquad  \mathcal{B}_{HI}:=\mathcal{A}_B'
    \cap
    \mathcal A_2~.
\end{equation}
For a depth two inclusion, the sequence 
\begin{equation}
    \mathcal{A}_H
    \subset \mathcal{B}_{H}
    \subset
    \mathcal{B}_{HI} ~,
    \label{eq:depth2}
\end{equation}
is an instance of the basic construction, i.e. $\mathcal{B}_{HI}=\langle \mathcal{B}_H, e_H\rangle$. We distinguish two cases:
\begin{itemize}
\item We first consider an \emph{irreducible} inclusion, i.e. with 
\begin{equation}
    \mathcal A_H
    =
    \mathcal A_B'
    \cap
    \mathcal A_{BH}
    =
    \mathbb C1~.
    \label{eq:irreducible-inclusion}
\end{equation}
In this case, we have the index equality
\begin{equation}
    \dim_{\mathbb C}
    \left(
        \mathcal B_{H}
    \right)
    =[
      \mathcal B_{H}
        :
        \mathbb C1
    ] =
    [
        \mathcal A_{BH}
        :
        \mathcal A_B
    ]~.
    \label{eq:irreducible-depth-two-dimension}
\end{equation}

\item We then consider a \emph{reducible} inclusion, i.e. with
\begin{equation}
    \mathcal A_H
    \neq
    \mathbb C1~.
    \label{eq:reducible-inclusion}
\end{equation}
In that case, the depth two index relation becomes
\begin{equation}
    [
    \mathcal{B}_H
        :
        \mathcal A_H
    ]
    =
    [
        \mathcal A_{BH}
        :
        \mathcal A_B
    ]~.
    \label{eq:reducible-depth-two-index}
\end{equation}

Using multiplicativity of the index,
\begin{align}
    [
        \mathcal{B}_H
        :
        \mathbb C1
    ]
    =
    [
        \mathcal{B}_H
        :
        \mathcal A_H
    ]
    [
        \mathcal A_H
        :
        \mathbb C1
    ]=
    [
        \mathcal A_{BH}
        :
        \mathcal A_B
    ]
    [
        \mathcal A_H
        :
        \mathbb C1
    ]~,
    \label{eq:depth-two-index-multiplicativity}
\end{align}
this becomes
\begin{equation}
    \dim_{\mathbb C}
    \left(
        \mathcal{B}_H
    \right)
    =
    [
        \mathcal A_{BH}
        :
        \mathcal A_B
    ]
    \dim_{\mathbb C}(\mathcal A_H)~.
    \label{eq:depth-two-dimension-reducible}
\end{equation}
\end{itemize}
More generally, if the inclusion has depth \(k\), then for every
\(n\geq k\), one can show that:
\begin{equation}
    \dim_{\mathbb C}(\mathcal Y_n)
    =
    [
        \mathcal A_{BH}
        :
        \mathcal A_B
    ]^{\,n-k+1}
    \dim_{\mathbb C}(\mathcal Y_{k-1})~.
    \label{eq:general-relative-commutant-dimension}
\end{equation}

It will also be useful to introduce the second derived tower:

\begin{definition}[Second derived tower]
The \emph{second derived tower} is defined by:
\begin{equation}
    \mathcal X_k
    :=
    \mathcal A_{BH}'
    \cap
    \mathcal A_{k-1}~,
    \qquad
    k\geq1~.
    \label{eq:Xk-definition}
\end{equation}
The algebra \(\mathcal X_k\) consists of all operators in
\(\mathcal A_{k-1}\) that commute with \(\mathcal A_{BH}\).
\end{definition}
Because we have the inclusion
\begin{equation}
    \mathcal A_{BH}'
    \subset
    \mathcal A_B'~,
\end{equation}
the second derived tower is contained in the first one. The two towers may therefore be arranged into the diagram
\begin{equation}
\begin{matrix}
\mathbb C1=\mathcal Y_0
&\subset&
\mathcal Y_1
&\subset&
\mathcal Y_2
&\subset&
\mathcal Y_3
&\subset&
\cdots
\\
&&\cup&&\cup&&\cup
\\
&&\mathbb C1=\mathcal X_1
&\subset&
\mathcal X_2
&\subset&
\mathcal X_3
&\subset&
\cdots~.
\end{matrix}
\label{eq:two-derived-towers}
\end{equation}
For each \(n\), this diagram contains a square of inclusions of the form:
\begin{equation}
\begin{matrix}
\mathcal Y_n
&\subset&
\mathcal Y_{n+1}
\\
\cup&&\cup
\\
\mathcal X_n
&\subset&
\mathcal X_{n+1}
\end{matrix}
\label{eq:commuting_square}
\end{equation}
One can show that this is an example of a commuting square:

\begin{definition}[Commuting square]
\label{def:commuting-square}
A diagram of von Neumann algebras
\begin{equation}
\begin{matrix}
\mathcal A&\subset&\mathcal D\\
\cup&&\cup\\
\mathcal C&\subset&\mathcal B
\end{matrix}
\label{eq:abstract-commuting-square}
\end{equation}
is a commuting square if the
corresponding conditional expectations commute:
\begin{equation}
    E_{\mathcal A}E_{\mathcal B}
    = E_{\mathcal B}E_{\mathcal A} =
    E_{\mathcal C}~.
\end{equation}
Equivalently,
\begin{equation}
    E_{\mathcal A}(\mathcal B)
    \subset
    \mathcal C~.
\end{equation}
\end{definition}

The square in \eqref{eq:commuting_square} is a commuting square because
the expectations on both the \(\mathcal X\)- and \(\mathcal Y\)-towers
are obtained by restricting the same conditional expectation
\begin{equation}
    E_{n}:
    \mathcal A_{n}
    \longrightarrow
    \mathcal A_{n-1}~,
\end{equation}
to the respective subalgebras.

\begin{definition}[Non-degenerate commuting square]
\label{def:nondegenerate-commuting-square}
Let
\begin{equation}
\begin{matrix}
\mathcal A&\subset&\mathcal D\\
\cup&&\cup\\
\mathcal C&\subset&\mathcal B
\end{matrix}
\end{equation}
be a commuting square of finite-index inclusions. The commuting square is
called \emph{non-degenerate} if any of the following
equivalent conditions are satisfied:
\begin{enumerate}

    \item The algebra \(\mathcal D\) is generated as a right
    \(\mathcal A\)-module by \(\mathcal B\). In other words, the linear
    span of
    \begin{equation}
        \{ba
        \mid
        b\in\mathcal B,\,
        a\in\mathcal A\}
    \end{equation}
    is weakly dense in \(\mathcal D\).

    \item The indices of the inclusions $\mathcal{A}\subset \mathcal{D}$ and $\mathcal{C}\subset \mathcal{B}$ satisfy
    \begin{equation}
        [\mathcal D:\mathcal A]
        =
        [\mathcal B:\mathcal C]~.
    \end{equation}
\end{enumerate}
\end{definition}
If \(\mathcal A_B\subset\mathcal A_{BH}\) has depth \(k\), the square \eqref{eq:commuting_square} is non-degenerate for each $n\geq k$. Consequently,
\begin{equation}
    \mathcal X_k
    \subset
    \mathcal X_{k+1}
    \subset
    \mathcal X_{k+2}
    \label{eq:X-basic-construction}
\end{equation}
is also an instance of the basic construction and we have the index equality
\begin{equation}
    [
        \mathcal X_{k+1}
        :
        \mathcal X_k
    ]
    =
    [
        \mathcal Y_{k+1}
        :
        \mathcal Y_k
    ]
    =
    [
        \mathcal A_{BH}
        :
        \mathcal A_B
    ]~.
    \label{eq:XY-index-equality}
\end{equation}

\subsection*{Shifted relative commutants}

We next relate the derived towers to relative commutants of the form
\begin{equation}
    \mathcal A_{-k}'
    \cap
    \mathcal A_{BH}~.
\end{equation}
The precise relation depends on the parity of \(k\), since odd and even
levels of the Jones tunnel are obtained from \(\mathcal A_B\) and
\(\mathcal A_{BH}\), respectively.

\begin{proposition}
\label{prop:tunnel-derived-tower-identification}
Let
\begin{equation}
    \mathcal A_B
    \subset
    \mathcal A_{BH}
\end{equation}
be a finite-index inclusion with Jones tower and tunnel
\(\{\mathcal A_j\}_{j\in\mathbb Z}\), and let
\begin{equation}
    \mathcal Y_j
    :=
    \mathcal A_B'
    \cap
    \mathcal A_{j-1}~,
    \qquad
    \mathcal X_j
    :=
    \mathcal A_{BH}'
    \cap
    \mathcal A_{j-1}~,
\end{equation}
be the first and second derived tower. Then the following statements hold:
\begin{itemize}
\item If \(k=2m+1\) is odd, with \(m\in\mathbb Z_{\geq0}\), then
\begin{equation}
    \mathcal A_{-k}'
    \cap
    \mathcal A_B
    =
    \gamma^m(\mathcal Y_{k-1})~, \quad \mathcal A_{-k}'
    \cap
    \mathcal A_{BH}
    =
    \gamma^m(\mathcal Y_k)~,\quad 
    \mathcal A_{-k}'
    \cap
    \mathcal B_{BH}
    =
    \gamma^m(\mathcal Y_{k+1})~.
\end{equation}
\item If \(k=2m\) is even, with \(m\in\mathbb Z_{\geq1}\), then
\begin{equation}
    \mathcal A_{-k}'
    \cap
    \mathcal A_B=
    \gamma^m(\mathcal X_k)~,\quad 
    \mathcal A_{-k}'
    \cap
    \mathcal A_{BH}=
    \gamma^m(\mathcal X_{k+1})~,\quad 
    \mathcal A_{-k}'
    \cap
    \mathcal B_{BH}=
    \gamma^m(\mathcal X_{k+2})~.
\end{equation}
\end{itemize}
\end{proposition}

\begin{proof}
Suppose first that \(k\) is odd. Writing
\begin{equation}
    k=2m+1~,
    \qquad
    m\in\mathbb Z_{\geq0}~,
\end{equation}
the definition
\begin{equation}
    \mathcal A_{2j-1}
    =
    \gamma^{-j}(\mathcal A_B)~,
    \qquad
    \mathcal A_{2j}
    =
    \gamma^{-j}(\mathcal A_{BH})
\end{equation}
gives
\begin{equation}
    \mathcal A_{-k}
    =
    \gamma^m(\mathcal A_B)~.
    \label{eq:odd-k-Aminus}
\end{equation}
The same shift relation yields
\begin{equation}
    \mathcal A_B
    =
    \gamma^m(\mathcal A_{k-2})~,
    \qquad
    \mathcal A_{BH}
    =
    \gamma^m(\mathcal A_{k-1})~,
    \qquad
    \mathcal B_{BH}
    =
    \gamma^m(\mathcal A_k)~.
    \label{eq:odd-k-gamma-relations}
\end{equation}

Since \(\gamma^m\) is spatially implemented on the standard
representation, it preserves commutants and intersections. In
particular, we obtain
\begin{equation}
    \gamma^m(\mathcal Y_{k-1}) =
    \gamma^m
    \left(
        \mathcal A_B'
        \cap
        \mathcal A_{k-2}
    \right)=
    \gamma^m(\mathcal A_B)'
    \cap
    \gamma^m(\mathcal A_{k-2})
    =
    \mathcal A_{-k}'
    \cap
    \mathcal A_B~.
\end{equation}
This proves the first equality. The other equalities are proved in a similar fashion.

Suppose now that \(k\) is even. Writing
\begin{equation}
    k=2m~,
    \qquad
    m\in\mathbb Z_{\geq1}~,
\end{equation}
the definition of the Jones tunnel gives
\begin{equation}
    \mathcal A_{-k}
    =
    \gamma^m(\mathcal A_{BH})~.
    \label{eq:even-k-Aminus}
\end{equation}
Moreover,
\begin{equation}
    \mathcal A_B
    =
    \gamma^m(\mathcal A_{k-1})~,
    \qquad
    \mathcal A_{BH}
    =
    \gamma^m(\mathcal A_k)~,
    \qquad
    \mathcal B_{BH}
    =
    \gamma^m(\mathcal A_{k+1})~.
    \label{eq:even-k-gamma-relations}
\end{equation}
Using again that \(\gamma^m\) preserves relative commutants and intersections, we get
\begin{equation}
    \gamma^m(\mathcal X_k)=
    \gamma^m
    \left(
        \mathcal A_{BH}'
        \cap
        \mathcal A_{k-1}
    \right)=
    \gamma^m(\mathcal A_{BH})'
    \cap
    \gamma^m(\mathcal A_{k-1})
    \nonumber=
    \mathcal A_{-k}'
    \cap
    \mathcal A_B~,
\end{equation}
which proves the required relation. The other equalities are proved in a similar fashion.
\end{proof}

We now derive the following result:

\begin{proposition}
\label{prop:depth-k-tunnel-basic-construction}
Suppose that $\mathcal A_B
    \subset
    \mathcal A_{BH}$ has depth \(k\). Then
\begin{equation}
    \mathcal A_{-k}'
    \cap
    \mathcal A_B
    \subset
    \mathcal A_{-k}'
    \cap
    \mathcal A_{BH}
    \subset
    \mathcal A_{-k}'
    \cap
    \mathcal B_{BH}
    \label{eq:relative_commutants}
\end{equation}
is an instance of the Jones basic construction. Consequently, we have the index equality:
\begin{equation}
    \left[
        \mathcal A_{-k}'
        \cap
        \mathcal A_{BH}
        :
        \mathcal A_{-k}'
        \cap
        \mathcal A_B
    \right]
    =
    \left[
        \mathcal A_{BH}
        :
        \mathcal A_B
    \right]~.
    \label{eq:index_equality}
\end{equation}
\end{proposition}

\begin{proof}
Suppose first that \(k\) is odd. By definition of depth \(k\), the
sequence
\begin{equation}
    \mathcal Y_{k-1}
    \subset
    \mathcal Y_k
    \subset
    \mathcal Y_{k+1}
    \label{eq:Y-depth-k-basic-construction}
\end{equation}
is an instance of the Jones basic construction. By
Proposition~\ref{prop:tunnel-derived-tower-identification}, applying
\(\gamma^m\), where \(k=2m+1\), maps this sequence to
\begin{equation}
    \mathcal A_{-k}'
    \cap
    \mathcal A_B
    \subset
    \mathcal A_{-k}'
    \cap
    \mathcal A_{BH}
    \subset
    \mathcal A_{-k}'
    \cap
    \mathcal B_{BH}~.
\end{equation}
Since \(\gamma^m\) is spatially implemented, it preserves the basic
construction and its index. Hence,
\begin{equation}
    \left[
        \mathcal A_{-k}'
        \cap
        \mathcal A_{BH}
        :
        \mathcal A_{-k}'
        \cap
        \mathcal A_B
    \right]
    =
    [
        \mathcal Y_k
        :
        \mathcal Y_{k-1}
    ]=
    [
        \mathcal A_{BH}
        :
        \mathcal A_B
    ]~.
\end{equation}

Suppose next that \(k\) is even. At the stabilization level, the
corresponding sequence 
\begin{equation}
    \mathcal X_k
    \subset
    \mathcal X_{k+1}
    \subset
    \mathcal X_{k+2}
    \label{eq:X-depth-k-basic-construction}
\end{equation}
is an instance of the Jones basic construction. By
Proposition~\ref{prop:tunnel-derived-tower-identification}, applying
\(\gamma^m\), where \(k=2m\), maps this sequence to
\begin{equation}
    \mathcal A_{-k}'
    \cap
    \mathcal A_B
    \subset
    \mathcal A_{-k}'
    \cap
    \mathcal A_{BH}
    \subset
    \mathcal A_{-k}'
    \cap
    \mathcal B_{BH}~.
\end{equation}
Therefore,
\begin{equation}
    \left[
        \mathcal A_{-k}'
        \cap
        \mathcal A_{BH}
        :
        \mathcal A_{-k}'
        \cap
        \mathcal A_B
    \right]
    =
    [
        \mathcal X_{k+1}
        :
        \mathcal X_k
    ] = [
        \mathcal A_{BH}
        :
        \mathcal A_B
    ]~.
\end{equation}
\end{proof}

\begin{remark}
The parity distinction arises because the Jones tunnel alternates between
images of \(\mathcal A_B\) and \(\mathcal A_{BH}\):
\begin{equation}
    \mathcal A_{-(2m+1)}
    =
    \gamma^m(\mathcal A_B)~,
    \qquad
    \mathcal A_{-2m}
    =
    \gamma^m(\mathcal A_{BH})~.
\end{equation}
Consequently, odd levels are naturally related to the first derived
tower \(\mathcal Y_k\), whereas even levels are naturally related to the
second derived tower \(\mathcal X_k\).
\end{remark}

\subsection{Pimsner--Popa bases}
\label{app:Pimsner-Popa-bases}

We now introduce a notion of basis for \(\mathcal A_{BH}\), regarded as
a right module over \(\mathcal A_B\).

\begin{definition}[Pimsner--Popa basis]
\label{def:Pimsner-Popa-frame}
A finite family
\begin{equation}
    \{a_i\}_{i=1}^r
    \subset
    \mathcal A_{BH}
\end{equation}
is called a \emph{(right) Pimsner--Popa basis} for the expectation
\(E_I:\mathcal A_{BH}\to\mathcal A_B\) if every
\(x\in\mathcal A_{BH}\) admits the reconstruction
\begin{equation}
    x
    =
    \sum_{i=1}^r
    a_iE_I(a_i^\dagger x)~,
    \label{eq:Pimsner-Popa-reconstruction}
\end{equation}
with the property that
\begin{equation}
    E_I(a_i^\dagger a_j)
    =
    \delta_{ij}1_{\mathcal A_B}~.
    \end{equation}
\end{definition}

In our definition, we have assumed that the basis is orthonormal with respect to the natural \(\mathcal A_B\)-valued inner product on
\(\mathcal A_{BH}\), i.e.
\begin{equation}
    \langle a,b\rangle_{\mathcal A_B}
    :=
    E_I(a^\dagger b)~.
    \label{eq:AB-valued-inner-product}
\end{equation} 
The coefficients
\begin{equation}
    x_i
    :=
    E_I(a_i^\dagger x)
    \in
    \mathcal A_B
    \label{eq:PP-coefficients}
\end{equation}
are \(\mathcal A_B\)-valued coefficients of \(x\) with respect to the basis. There is a direct relation between the Pimsner--Popa basis and the index:

\begin{definition}[Index]
\label{def:index-element}
Let
\begin{equation}
    E_I:
    \mathcal A_{BH}
    \longrightarrow
    \mathcal A_B
\end{equation}
be a finite-index conditional expectation, and let
\begin{equation}
    \{a_i\}_{i=1}^r
    \subset
    \mathcal A_{BH}
\end{equation}
be a (two-sided) Pimsner--Popa basis. The \emph{Watatani index} of
\(E_I\) is defined by
\begin{equation}
    \operatorname{Ind}(E_I)
    :=
    \sum_{i=1}^r
    a_ia_i^\dagger~.
    \label{eq:index-element}
\end{equation}
\end{definition}

This index element is independent of the choice of two-sided
Pimsner--Popa basis and belongs to the center of
\(\mathcal A_{BH}\). If \(\mathcal A_{BH}\) is a factor and \(E_I\) is
the minimal conditional expectation, then
\begin{equation}
    \operatorname{Ind}(E_I)
    =
    [\mathcal A_{BH}:\mathcal A_B]\,
    \mathbf{1}~.
\end{equation}
\begin{remark}
The number of elements in a Pimsner--Popa frame is generally \emph{not} equal to
the Jones index. Such an equality requires additional assumptions.
\end{remark}

To study the localization of Pimsner--Popa bases, we consider the square
\begin{equation}
\begin{matrix}
\mathcal A_B
&\subset&
\mathcal A_{BH}
\\
\cup&&\cup
\\
\mathcal A_{-k}'
\cap
\mathcal A_B
&\subset&
\mathcal A_{-k}'
\cap
\mathcal A_{BH}~.
\end{matrix}
\label{eq:localized-commuting-square}
\end{equation}
The vertical inclusions are natural inclusions, and the lower expectation
is obtained by restricting \(E_I\) to the respective relative commutant. For a depth-\(k\) inclusion, the index equality
\eqref{eq:index_equality} of Proposition \ref{prop:depth-k-tunnel-basic-construction} implies that the square is non-degenerate. We now prove the following:

 \begin{proposition}[Localization of Pimsner--Popa bases]
\label{prop:localized-Pimsner-Popa-frame}
Assume that the square
\begin{equation}
\begin{matrix}
\mathcal A_B
&\subset&
\mathcal A_{BH}
\\
\cup&&\cup
\\
\mathcal A_{-k}'
\cap
\mathcal A_B
&\subset&
\mathcal A_{-k}'
\cap
\mathcal A_{BH}
\end{matrix}
\label{eq:localized-square}
\end{equation}
is a non-degenerate commuting square with respect to the compatible conditional expectations. Then a Pimsner--Popa basis for the lower inclusion
\begin{equation}
    \mathcal A_{-k}'
    \cap
    \mathcal A_B
    \subset
    \mathcal A_{-k}'
    \cap
    \mathcal A_{BH}
\end{equation}
is also a Pimsner--Popa basis for the original inclusion
\begin{equation}
    \mathcal A_B
    \subset
    \mathcal A_{BH}~.
\end{equation}
In particular, the basis elements may be chosen such that
\begin{equation}
    a_i
    \in
    \mathcal A_{-k}'
    \cap
    \mathcal A_{BH}~.
    \label{eq:localized-frame-elements}
\end{equation}
\end{proposition}

\begin{proof}
Let
\begin{equation}
    \{a_i\}_{i=1}^r
    \subset
    \mathcal A_{-k}'
    \cap
    \mathcal A_{BH}
\end{equation}
be a Pimsner--Popa frame for the lower inclusion. Then, for every
\begin{equation}
    x
    \in
    \mathcal A_{-k}'
    \cap
    \mathcal A_{BH}~,
\end{equation}
the reconstruction formula reads
\begin{equation}
    x
    =
    \sum_i
    a_i
    E_I|_{\mathcal A_{-k}'\cap\mathcal A_{BH}}
    \left(
        a_i^* x
    \right)=\sum_i
    a_i
    E_I
    \left(
        a_i^* x
    \right)~,
    \label{eq:lower-frame-reconstruction}
\end{equation}
because the
conditional expectation for the lower inclusion is the restriction of
the original conditional expectation
\begin{equation}
    E_I:
    \mathcal A_{BH}
    \longrightarrow
    \mathcal A_B~.
\end{equation}
Non-degeneracy of the commuting square implies that
\(\mathcal A_{BH}\) is generated as a right \(\mathcal A_B\)-module by
the algebra $\mathcal A_{-k}'
        \cap
        \mathcal A_{BH}$. Equivalently, the linear span
\begin{equation}
    \operatorname{span}
    \left\{
        x b
        \,\middle|\,
        x\in
        \mathcal A_{-k}'
        \cap
        \mathcal A_{BH},
        \ 
        b\in\mathcal A_B
    \right\}
    \label{eq:module-dense-subspace}
\end{equation}
is weakly dense in \(\mathcal A_{BH}\). Consider an element of this dense subspace. Using the right \(\mathcal A_B\)-module property of \(E_I\), we obtain
\begin{align}
    \sum_i
    a_i
    E_I
    \left(
        a_i^\dagger x b
    \right)
    &=
    \sum_i
    a_i
    E_I
    \left(
        a_i^\dagger x
    \right)b
    \nonumber\\
    &=
    x b~.
    \label{eq:frame-reconstruction-module-span}
\end{align}
By linearity, the same reconstruction formula holds on the entire
linear span in \eqref{eq:module-dense-subspace}. Finally, since \(E_I\) is normal and the sum over \(i\) is finite, the
map
\begin{equation}
    a
    \longmapsto
    \sum_i
    a_i
    E_I
    \left(
        a_i^\dagger a
    \right)
\end{equation}
is normal. The reconstruction formula therefore extends by weak
continuity to every \(a\in\mathcal A_{BH}\). Thus, \(\{a_i\}_{i=1}^r\) is a Pimsner--Popa frame for
\(\mathcal A_B\subset\mathcal A_{BH}\).
\end{proof}

\begin{figure}[t]
    \centering    \includegraphics[width=0.7\linewidth]{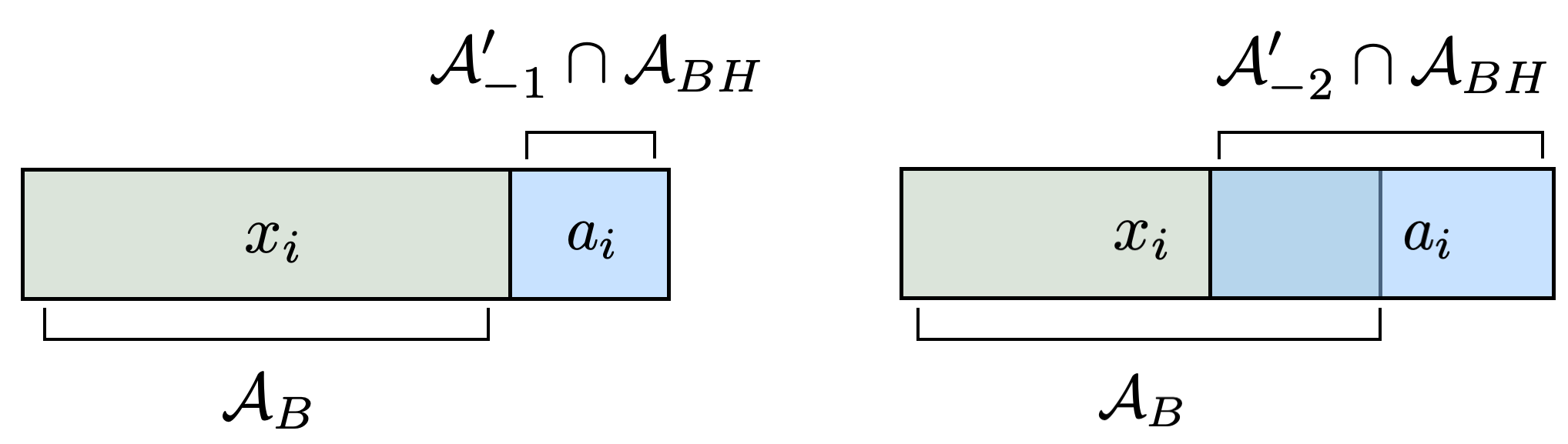}
    \caption{Pimsner--Popa bases for the inclusion $\mathcal{A}_B\subset \mathcal{A}_{BH}$. Left: A depiction of a depth one inclusion, where one can choose a Pimsner--Popa basis $\{a_i
    \}$ to be localized in the first relative commutant $\mathcal{A}_H=\mathcal{A}_{-1}\cap \mathcal{A}_{BH}$. Right: A depiction of a depth two inclusion, where one can choose a Pimsner--Popa basis $\{a_i\}$ to be localized in the second relative commutant $\mathcal{A}_{-2}'\cap \mathcal{A}_{BH}$.}
    \label{fig:pimsner-popa}
\end{figure}

We will study the depth one and two examples in more detail:

\subsubsection*{Depth one inclusion}

For a depth one inclusion, Proposition \ref{prop:localized-Pimsner-Popa-frame} implies that the Pimsner--Popa basis \(\{a_i\}_{i=1}^r\) for $\mathcal{A}_{B}\subset \mathcal{A}_{BH}$ may be
chosen inside the first relative commutant, i.e.
\begin{equation}
    a_i \in
    \mathcal A_B'
    \cap
    \mathcal A_{BH}
    =
    \mathcal A_H~.
    \label{eq:depth-one-localized-basis}
\end{equation}

Since the Pimsner--Popa basis is assumed to be orthonormal, the elements \(\{a_i\}\) form an
ordinary vector space basis of \(\mathcal A_H\). In that case, the number of basis elements is given by the index
\begin{equation}
   r= \dim_{\mathbb C}(\mathcal A_H)
    =
    [
        \mathcal A_{BH}
        :
        \mathcal A_B
    ]~.
    \label{eq:depth-one-number-elements}
\end{equation}

\subsubsection* {Depth two inclusion}

For a depth two inclusion, the localization shows that a Pimsner--Popa basis $\{a_i\}_{i=1}^r$ for the inclusion $\mathcal{A}_B\subset \mathcal{A}_{BH}$ can be chosen such that:
\begin{equation}
    a_i
    \in
    \mathcal A_{-2}'\cap \mathcal{A}_{BH}=\gamma(\mathcal A_{BH})'
    \cap
    \mathcal A_{BH}~.
    \label{eq:depth-two-localized-basis}
\end{equation}
The index of this inclusion is
\begin{equation}
    [
        \gamma(\mathcal A_{BH})'
        \cap
        \mathcal A_{BH}
        :
        \gamma(\mathcal A_{BH})'
        \cap
        \mathcal A_B
    ]
    =
    [
        \mathcal A_{BH}
        :
        \mathcal A_B
    ]~,
    \label{eq:depth-two-localized-index}
\end{equation}
from which it follows that 
\begin{equation}
    \dim_{\mathbb C}
    \left(
        \gamma(\mathcal A_{BH})'
        \cap
        \mathcal A_{BH}
    \right)
    =
    [
        \mathcal A_{BH}
        :
        \mathcal A_B
    ]
    \dim_{\mathbb C}
    \left(
        \gamma(\mathcal A_{BH})'
        \cap
        \mathcal A_B
    \right)~.
    \label{eq:depth-two-localized-dimension}
\end{equation}

If the inclusion is furthermore irreducible, the lower relative commutant is trivial, i.e. 
\begin{equation}\gamma(\mathcal A_{BH})'
    \cap
    \mathcal A_B
    =  J_{B}^{\Phi}(\mathcal{A}_B'\cap \mathcal{A}_{BH})J_{B}^{\Phi}=
    \mathbb C\mathbf{1}~.
\end{equation}
In that case,
\begin{equation}
    r= \dim_{\mathbb C}
    \left(
        \gamma(\mathcal A_{BH})'
        \cap
        \mathcal A_{BH}
    \right)
    =
    [
        \mathcal A_{BH}
        :
        \mathcal A_B
    ]~,
    \label{eq:depth-two-irreducible-dimension}
\end{equation}
so the number of Pimsner--Popa basis elements is given by the index $[
        \mathcal A_{BH}
        :
        \mathcal A_B
    ]$.

\subsection{The algebraic Fourier transform}
\label{app:canonical-intertwiners}

We now introduce the algebraic Fourier transform associated
with a finite index inclusion
\begin{equation}
    \mathcal A_B
    \subset
    \mathcal A_{BH}~,
\end{equation}
of depth two. Our definition will be intrinsic using the Jones projections and the canonical conditional expectation. Recall that the basic constructions
\begin{equation}
    \mathcal B_{BH}
    =
    \langle
        \mathcal A_{BH},
        e_I
    \rangle~,
    \qquad
    \mathcal B_{RH}
    =
    \langle
        \mathcal A_{RH},
        e_H
    \rangle
\end{equation}
contain two distinguished relative commutants:
\begin{equation}
    \mathcal B_I
    :=
    \mathcal B_{RH}
    \cap
    \mathcal A_R'~,
    \qquad
    \mathcal B_H
    :=
    \mathcal B_{BH}
    \cap
    \mathcal A_B'~.
    \label{eq:BIBH}
\end{equation}
These algebras play symmetric roles and will be identified through the algebraic Fourier transform. To define the map, we recall that for a depth two inclusion the following sequence
\begin{equation}
\mathcal{A}_H \subset \mathcal{B}_{H} \subset \mathcal{B}_{HI}~,
\end{equation}
is an instance of Jones basic construction, i.e. 
\begin{equation}
    \mathcal B_{HI}
=
    \langle \mathcal B_{H},e_H \rangle~.
    \label{eq:BIH}
\end{equation}
The conditional expectation $\mathcal{E}_I:\langle \mathcal{B}_{BH},e_H\rangle \to \mathcal{B}_{BH}$, which is related to $E_I:\mathcal{A}_{BH}\to \mathcal{A}_B$ by an application of the canonical endomorphism,
\begin{equation} \label{eq:cond_curly_EI}
\mathcal{E}_I := \gamma^{-1} \circ E_I \circ \gamma~,
\end{equation}
restricts to a conditional expectation on the corresponding relative
commutant. Without loss of generality, we denote its restriction by
\begin{equation}
    \mathcal E_I:
    \mathcal B_{HI}
    \longrightarrow
    \mathcal B_H~.
\end{equation}
We can now introduce the following map:

\begin{definition}[algebraic Fourier transform]
\label{def:quantum-fourier-transform}
The \emph{algebraic Fourier transform} is the linear map
\begin{equation}
    \mathcal F:
    \mathcal B_I
    \longrightarrow
    \mathcal B_H
\end{equation}
defined by
\begin{equation}
    \mathcal F(a)
    :=
    d_I^{3}
    \,
    \mathcal E_I
    \!\left(
        e_He_Ia
    \right)~,
    \label{eq:Fourier-transform}
\end{equation}
\end{definition}
where
\begin{equation}
d_I = [\mathcal{A}_{BH}:\mathcal{A}_B]^{1/2}~.
\end{equation}

The algebraic Fourier transform is invertible. To show this, we will give an explicit expression for the inverse map. Introducing the relative commutant
\begin{equation}
\mathcal A_I
    :=
    \mathcal A_{RH}
    \cap
    \mathcal A_R'~,
\end{equation}
the sequence 
\begin{equation}
\mathcal{A}_I \subset \mathcal{B}_{I} \subset \mathcal{B}_{HI}~,
\end{equation}
is part of the first derived tower associated to the inclusion $\mathcal{A}_R\subset \mathcal{A}_{RH}$. For a depth two inclusion, the algebra $\mathcal{B}_{HI}$ is given by Jones basic construction, i.e.
\begin{equation}
\mathcal{B}_{HI}=\langle \mathcal{B}_I,e_I\rangle~.
\end{equation}
We denote by 
\begin{equation}
    \mathcal E_H':
    \mathcal B_{HI}
    \longrightarrow
    \mathcal B_I
\end{equation}
the conditional expectation \(\mathcal{E}_H': \langle \mathcal{B}_{RH},e_I\rangle \to \mathcal B_{RH}\), restricted to
the appropriate relative commutant. 

We then have the following result:

\begin{proposition}[Inverse algebraic Fourier transform]
\label{prop:inverse-quantum-Fourier-transform}
The inverse algebraic Fourier transform
\begin{equation}
    \mathcal F^{-1}:
    \mathcal B_H
    \longrightarrow
    \mathcal B_I
\end{equation}
is given by
\begin{equation}
    \mathcal F^{-1}(b)
    =
    d_I^3\,
    \mathcal E_H'
    \!\left(
        e_Ie_Hb
    \right)~.
    \label{eq:inverse-Fourier-transform}
\end{equation}
It satisfies
\begin{equation}
    \mathcal F^{-1}\circ\mathcal F
    =
    \operatorname{id}_{\mathcal B_I}~,
    \qquad
    \mathcal F\circ\mathcal F^{-1}
    =
    \operatorname{id}_{\mathcal B_H}~.
\end{equation}
\end{proposition}

\begin{proof}
For \(a\in\mathcal B_I\), substituting
\eqref{eq:Fourier-transform} into
\eqref{eq:inverse-Fourier-transform} gives
\begin{align}
    \mathcal F^{-1}\!\left(\mathcal F(a)\right)
    &=
    d_I^6\,
    \mathcal E_H'
    \!\left[
        e_Ie_H\,
        \mathcal E_I
        \!\left(
            e_He_Ia
        \right)
    \right]~.
    \label{eq:inverse-Fourier-composition-I}
\end{align}
We need to simplify the above expression. First note that $\mathcal{B}_I $ is contained in the basic extension algebra 
\begin{equation}
\mathcal{B}_{HI} =\langle\mathcal{B}_H,e_H \rangle ~,
\end{equation}
The set of elements $xe_Hy$ with $x,y\in \mathcal{B}_{H}$ is weakly dense in $\mathcal{B}_{HI}$ so it suffices to prove our statement for elements $a$ of the form $a=xe_Hy$. Using that the conditional expectation 
\begin{equation}
\mathcal{E}_I: \mathcal{B}_{HI}\to \mathcal{B}_H~,
\end{equation}
satisfies $\mathcal{E}_I(e_H)=d_I^{-2}$, we find 
\begin{equation}
\mathcal{E}_I(e_He_Ixe_Hy)=\mathcal{E}_I(E_H(e_Ix)e_Hy)=d_I^{-2}E_H(e_Ix)y~.
\end{equation}
At the first equality, we have used the identity 
\begin{equation}
e_Hxe_H=e_HE_H(x)
\end{equation}
with $E_H:\mathcal{B}_H\to \mathcal{A}_H$. Plugging this back into \eqref{eq:inverse-Fourier-composition-I} and using the Jones relation \(e_Ie_He_I = d_I^{-2} e_I\) we find that 
\begin{equation}
    \mathcal F^{-1}\!\left(\mathcal F(a)\right)
    =
    d_I^6\,
    \mathcal E_H'
    \!\left[
        e_Ie_H d_I^{-2}E_H(e_Ix)y
    \right] =  d_I^4\,
    \mathcal E_H'
    \!\left[
        e_Ie_He_Ixe_H y
    \right]= d_I^2\,
    \mathcal E_H'
    \!\left[
        e_I a
    \right]~.
\end{equation}
Since $ae_H\in \mathcal{B}_I$, we can use the bi-module property for the conditional expectation
\begin{equation}
\mathcal{E}_H': \langle\mathcal{B}_I,e_I\rangle \to \mathcal{B}_I~,
\end{equation}
combined with $\mathcal{E}_H'(e_I)=d_I^{-2}$ to obtain our final result
\begin{equation}
    \mathcal F^{-1}\!\left(\mathcal F(a)\right)
    =d_I^2\,
    \mathcal E_H'
    \!\left[
        e_I a
    \right] = d_I^2\,d_I^{-2} a
    =a~.
\end{equation}

By exchanging the roles of \(I\) and \(H\) in the above argument, one can prove $\mathcal F\circ\mathcal F^{-1}(b)=b$ for \(b\in\mathcal B_H\) in a similar fashion.
\end{proof}

We have thus shown the following result:
\begin{proposition}
\label{prop:Fourier-isomorphism}
The algebraic Fourier transform defines a vector space isomorphism
\begin{equation}
    \mathcal F:
    \mathcal B_I
    \xrightarrow{\;\cong\;}
    \mathcal B_H~.
\end{equation}
\end{proposition}

\begin{remark}
Note that the algebraic Fourier transform does not preserve the naive algebraic structure associated to the relative commutants. Instead, one has to introduce a modified product on $\mathcal{B}_H$ given by
\begin{equation}
\mathcal{F}(a_1) * \mathcal{F}(a_2) :=\mathcal{F}(a_1a_2)~.
\end{equation}

\end{remark}

Note that there is a natural inner product on the relative commutant $\mathcal{B}_I$ given by
\begin{equation}
\langle a_1,a_2 \rangle_I :=\mathcal{E}_I(a_1^* a_2)~, 
\end{equation}
where we restrict $\mathcal{E}_I:\langle \mathcal{B}_{BH},e_H \rangle \to \mathcal{B}_{BH}$ to the relative commutant $\mathcal{B}_I$. Similarly, we have an inner product on $\mathcal{B}_H$ defined in terms of $\mathcal{E}_H'$. The algebraic Fourier transform has the property that it relates 
orthonormal bases within the two relative commutants:

\begin{proposition}[Fourier transform of orthonormal bases]
\label{prop:Fourier-localized-bases}
Given an orthonormal basis $\{ \hat{a}_i \}_{i=1}^r$ for the inclusion $\mathcal{B}_{I}$ with respect to the conditional expectation $\mathcal{E}_I$, the elements
\begin{equation}
    \hat b_i
    :=
    \mathcal F(\hat a_i)~,
    \qquad
    i=1,\ldots,r~
    \label{eq:basis-Fourier-forward}
\end{equation}
constitute an orthonormal basis for $\mathcal{B}_H$ with respect to the conditional expectation $\mathcal{E}_H'$.
\end{proposition}

\begin{proof}
Since
\begin{equation}
    \mathcal F:
    \mathcal B_I
    \xrightarrow{\;\cong\;}
    \mathcal B_H
\end{equation}
is an isomorphism, the elements $\hat b_i
    =
    \mathcal F(\hat a_i)$ form a basis for \(\mathcal B_H\).

What remains is to evaluate the inner product
\begin{equation}
\mathcal{E}_H'(\hat{b}_i^*\hat{b}_j) =\mathcal{E}_H'(\mathcal{F}(\hat{a}_i)^*\mathcal{F}(\hat{a}_j))= d_I^6\,\mathcal{E}_H'(\mathcal{E}_I(\hat{a}_i^{*}e_Ie_H )\mathcal{E}_I(e_He_I\hat{a}_j)) ~.
\end{equation}
From \eqref{eq:cond_curly_EI} we deduce that the conditional expectation $\mathcal{E}_I$ can be implemented by
\begin{equation}
\mathcal{E}_I(x) = \mathcal{V}_I^{\dagger} x \mathcal{V}_I~,
\end{equation}
where the isometry $\mathcal{V}_I$ is related to $V_I$ by $\mathcal{V}_I = \Gamma V_I \Gamma^{\dagger}=\gamma^{-1}(V_I) \in \gamma^{-1}(\mathcal{A}_B)'= \mathcal{B}_{R}$. Similarly, 
\begin{equation}
\mathcal{E}_H'(x) =\mathcal{V}_H^{\dagger}x\mathcal{V}_H~,
\end{equation}
where $\mathcal{V}_H=\Gamma^{\dagger}V_H\Gamma = \gamma(V_H) \in \gamma(\mathcal{A}_R)'=\mathcal{B}_{B}$. We have the projectors
\begin{equation}
\mathcal{V}_I\mathcal{V}_I^{\dagger} = \gamma^{-1}(e_I) =: \tilde{e}_I~, \qquad \mathcal{V}_H\mathcal{V}_H^{\dagger} = \gamma(e_H) =: \tilde{e}_H
\end{equation}
which satisfies the Jones relation 
\begin{equation}
e_H\tilde{e}_I e_H =\frac{1}{d_I^2}e_H~, \qquad e_I\tilde{e}_H e_I =\frac{1}{d_I^2}e_I~.
\end{equation}
Using the above formulas we can rewrite:
\begin{align}
d_I^6\,\mathcal{E}_H'(\mathcal{E}_I(\hat{a}_i^{*}e_Ie_H )\mathcal{E}_I(e_He_I\hat{a}_j)) & =d_I^6\,\mathcal{V}_H^{\dagger}\mathcal{V}_I^{\dagger}\hat{a}_i^{*}e_Ie_H\mathcal{V}_I \mathcal{V}_I^{\dagger}e_He_I\hat{a}_j\mathcal{V}_I\mathcal{V}_H \nonumber \\
&=d_I^6\,\mathcal{V}_H^{\dagger}\mathcal{V}_I^{\dagger}\hat{a}_i^{*}e_Ie_H \tilde{e}_Ie_He_I\hat{a}_j\mathcal{V}_I\mathcal{V}_H  \nonumber \\
&=d_I^2\,\mathcal{V}_H^{\dagger}\mathcal{V}_I^{\dagger}\hat{a}_i^{*}e_I\hat{a}_j\mathcal{V}_I\mathcal{V}_H \nonumber \\
&=d_I^2\,\mathcal{V}_I^{\dagger}\hat{a}_i^{*}\mathcal{V}_H^{\dagger}e_I\mathcal{V}_H\hat{a}_j\mathcal{V}_I \nonumber \\
&=d_I^2\,\mathcal{V}_I^{\dagger}\hat{a}_i^{*}\mathcal{E}_H'(e_I)\hat{a}_j\mathcal{V}_I\nonumber \\
&=\mathcal{V}_I^{\dagger}(\hat{a}_i^{*}\hat{a}_j)\mathcal{V}_I \nonumber \\
&=\mathcal{E}_I(\hat{a}_i^{*}\hat{a}_j)=\delta_{ij}\mathbf{1}~, 
\end{align}
where we have used that $\mathcal{V}_H\in \mathcal{B}_B$ commutes with $\hat{a}_j\mathcal{V}_I \in \mathcal{B}_{RH}$. At the last step, we have used that $\{\hat{a}_i\}$ is orthonormal with respect to the conditional expectation $\mathcal{E}_I:\mathcal{B}_{HI}\to \mathcal{B}_H$. This proves the statement.
\end{proof}

Let $\{a_i\}_{i=1}^{r}$ be a Pimsner--Popa basis for $\mathcal{A}_B\subset \mathcal{A}_{BH}$ localized in the corresponding relative commutant 
\begin{equation}
a_i \in \gamma(\mathcal{A}_{BH})'\cap \mathcal{A}_{BH}~.
\end{equation}
According to Proposition \ref{prop:tunnel-derived-tower-identification}, the operators defined by
\begin{equation}
\hat a_i := \Gamma a_i \Gamma^{\dagger}
\end{equation}
are contained in
\begin{equation}
\hat a_i\in 
    \mathcal{A}_{BH}' \cap \langle \mathcal{B}_{BH},e_H\rangle = \mathcal{B}_I~,
\end{equation}
and they form a Pimsner--Popa basis for the inclusion $\mathcal{B}_{BH}\subset \langle \mathcal{B}_{BH},e_H\rangle$ with the property that
\begin{equation}
\mathcal{E}_I(\hat{a}_i^{*}\hat{a}_j)=\delta_{ij}\mathbf{1}~.
\end{equation}
The previous proposition shows that the elements 
\begin{equation}
\hat b_i = \mathcal{F}(\hat a_i)
\end{equation}
form a basis for $\mathcal{B}_H$ with the property that  \begin{equation}\mathcal{E}_H'(\hat{b}_i^*\hat{b}_j) = \delta_{ij}\mathbf{1}~.\end{equation}
Using that the relevant square of inclusions is a non-degenerate commuting square, we can lift the basis to a Pimsner--Popa basis for the inclusion $\mathcal{B}_{RH}\subset \langle \mathcal{B}_{RH},e_I\rangle$. Consequently, the shifted operators
\begin{equation}
b_i := \Gamma \hat{b}_i \Gamma^{\dagger} \in \mathcal{\gamma}^{-1}(\mathcal{A}_{RH})\cap \mathcal{A}_{RH}~,
\end{equation}
form a Pimsner--Popa basis for the inclusion $\mathcal{A}_R\subset \mathcal{A}_{RH}$ with respect to the conditional expectation $E_H'$.

We end this section with an equivalent expression for the algebraic Fourier transform in terms of the isometric embeddings. Let us write $\mathcal{H}_I$ and $\mathcal{H}_H$ for the GNS representation of the algebras $\mathcal{B}_I$ and $\mathcal{B}_H$ respectively, and $\mathcal{H}_{HI}$ for the one associated to $\mathcal{B}_{HI}$. As before, we introduce the isometries \begin{equation}
\mathcal{V}_I:\mathcal{H}_{H}\to \mathcal{H}_{HI}~,\qquad \mathcal{V}_H:\mathcal{H}_{I}\to \mathcal{H}_{HI}~,
\end{equation}
where $\mathcal{V}_I = \gamma^{-1}(V_I)$, $\mathcal{V}_H =\gamma(V_H)$, so that the conditional expectations
\begin{equation}
\mathcal{E}_I:\mathcal{B}_{HI} \to \mathcal{B}_H~, \qquad \mathcal{E}_H':\mathcal{B}_{HI} \to \mathcal{B}_I
\end{equation}
can be represented as 
\begin{equation}
\mathcal{E}_I(x)=\mathcal{V}_I^{\dagger}x\mathcal{V}_I~, \qquad \mathcal{E}_H'(x)=\mathcal{V}_H^{\dagger}x\mathcal{V}_H
\end{equation}
respectively. 

Given the density operators $\Phi_I\in \mathcal{B}_I, \Phi_H\in \mathcal{B}_H$, we write their associated GNS states by $|\Phi_I\rangle \in \mathcal{H}_I$ and $|\Phi_H\rangle \in \mathcal{H}_H$, and we will use the same notation for their embeddings in the Hilbert space $\mathcal{H}_{HI}$ through $\mathcal{V}_H$ and $\mathcal{V}_I$. 

We now introduce the following isometric maps: 
\begin{equation}
V_I|\Phi_I\rangle = d_I|e_I\Phi_I\rangle~, \qquad V_H|\Phi_H\rangle  =d_I|e_H\Phi_H\rangle~,
\end{equation}
which are the restrictions of the isometries $V_I$ and $V_H$ to the relevant subspaces. This assignment is consistent with the Jones relations, i.e.
\begin{equation}
V_I = d_I V_I V_I^{\dagger}\mathcal{V}_H =d_I e_I \mathcal{V}_H~, \qquad V_H = d_I V_H V_H^{\dagger}\mathcal{V}_I =d_I e_H \mathcal{V}_I~.
\end{equation}

Given a state $|\Psi_{HI}\rangle \in \mathcal{H}_{HI}$, the conjugate maps 
\begin{equation}
V_I^{\dagger}:\mathcal{H}_{HI}\to \mathcal{H}_I~, \qquad V_H^{\dagger}: \mathcal{H}_{HI}\to \mathcal{H}_H
\end{equation}
can now be expressed as:
\begin{equation} 
V_I^{\dagger}|\Psi_{HI}\rangle =d_I|\mathcal{E}_H'(e_I\Psi_{HI})\rangle~, \qquad 
V_H^{\dagger}|\Psi_{HI}\rangle=d_I|\mathcal{E}_I(e_H\Psi_{HI})\rangle~. 
\end{equation}
 Indeed, for $a\in \mathcal{B}_I$ we have 
\begin{align}
\langle \Phi_I|a^*V_I^{\dagger}|\Psi_{HI}\rangle = d_I\langle \Phi_I|a^*\mathcal{V}_H^{\dagger} e_I|\Psi_{HI}\rangle = d_I\langle \Phi_I|a^*|\mathcal{E}_H'(e_I\Psi_{HI})\rangle~,
\end{align}
which shows the required identity. A similar argument can be used for $V_H^{\dagger}$. Note that this assignment is consistent with the Jones relations, i.e.
\begin{equation}
V_I^{\dagger}= d_I \mathcal{V}_H^{\dagger}V_IV_I^{\dagger} = d_I\mathcal{V}_H^{\dagger}e_I~, \qquad V_H^{\dagger}= d_I \mathcal{V}_I^{\dagger}V_HV_H^{\dagger} = d_I\mathcal{V}_I^{\dagger}e_H~.
\end{equation}
As an additional check, one can show that 
\begin{equation}
V_I^{\dagger}V_I|\Phi_I\rangle =d_I^2|\mathcal{E}_H'(e_I\Phi_I)\rangle =|\Phi_I\rangle~,
\end{equation}
and 
\begin{equation}
V_IV_I^{\dagger}|\Psi_{HI}\rangle = d_I^2|e_I\mathcal{E}_H'(e_I\Psi_{HI}))=d_I^2 e_I\tilde{e}_He_I|\Psi_{HI}\rangle = e_I|\Psi_{HI}\rangle~.
\end{equation}

We are now ready to give the expression for the Fourier transform and its inverse in terms of the relevant isometries and their conjugates:

\begin{proposition}
The algebraic Fourier transform satisfies
\begin{equation}
|\mathcal{F}(\Phi_I)\rangle = (V_H^{\dagger}\circ S_H\circ V_I)|\Phi_I\rangle ~,
\end{equation}
where we have introduced the entanglement swap operator $S_H:=d_I e_H$. Similarly, the inverse Fourier transform satisfies
\begin{equation}
|\mathcal{F}^{-1}(\Phi_H)\rangle = (V_I^{\dagger} \circ S_I\circ V_H)|\Phi_H\rangle~,
\end{equation}
where we have introduced the entanglement swap operator $S_I:=d_I e_I$. 
\end{proposition}

\begin{proof}
Using the representation of the conditional expectation $\mathcal{E}_I$ in terms of the isometry $\mathcal{V}_I$, we find that:
\begin{equation}
\mathcal{F}(\Phi_I) = d_I^3\, \mathcal{V}_I^{\dagger} e_H e_I \Phi_I\mathcal{V}_I~.
\end{equation}
Using the expressions for $V_I,S_H$ and $V_I^{\dagger}$ as defined above, we thus find that
\begin{equation}
|\mathcal{F}(\Phi_I)\rangle = d_I^3\mathcal{V}_I^{\dagger}e_He_I|\Phi_I\rangle=V_H^{\dagger}\circ S_H\circ V_I|\Phi_I\rangle~.
\end{equation}
By a similar argument we have 
\begin{equation}
\mathcal{F}^{-1}(\Phi_H) = d_I^3\, \mathcal{V}_H^{\dagger} e_I e_H \Phi_H\mathcal{V}_H~,
\end{equation}
which shows the result for the inverse Fourier transform. 
\end{proof}

\bibliographystyle{JHEP}
\bibliography{references}

@article{Kang:2018xqy,
    author = "Kang, Monica Jinwoo and Kolchmeyer, David K.",
    title = "{Holographic Relative Entropy in Infinite-Dimensional Hilbert Spaces}",
    eprint = "1811.05482",
    archivePrefix = "arXiv",
    primaryClass = "hep-th",
    doi = "10.1007/s00220-022-04627-z",
    journal = "Commun. Math. Phys.",
    volume = "400",
    number = "3",
    pages = "1665--1695",
    year = "2023"
}

@article{Kang:2019dfi,
    author = "Kang, Monica Jinwoo and Kolchmeyer, David K.",
    title = "{Entanglement wedge reconstruction of infinite-dimensional von Neumann algebras using tensor networks}",
    eprint = "1910.06328",
    archivePrefix = "arXiv",
    primaryClass = "hep-th",
    reportNumber = "CALT-TH-2019-042",
    doi = "10.1103/PhysRevD.103.126018",
    journal = "Phys. Rev. D",
    volume = "103",
    number = "12",
    pages = "126018",
    year = "2021"
}

@article{Gesteau:2021jzp,
    author = "Gesteau, Elliott and Kang, Monica Jinwoo",
    title = "{Nonperturbative gravity corrections to bulk reconstruction}",
    eprint = "2112.12789",
    archivePrefix = "arXiv",
    primaryClass = "hep-th",
    reportNumber = "CALT-TH-2021-042",
    doi = "10.1088/1751-8121/acef7d",
    journal = "J. Phys. A",
    volume = "56",
    number = "38",
    pages = "385401",
    year = "2023"
}

@article{Leutheusser:2025zvp,
    author = "Leutheusser, Samuel and Liu, Hong",
    title = "{Volume as an index of a subalgebra}",
    eprint = "2508.00056",
    archivePrefix = "arXiv",
    primaryClass = "hep-th",
    month = "7",
    year = "2025"
}

@Article{antonini2023cosmology,
author={Antonini, Stefano
and Sasieta, Martin
and Swingle, Brian},
title={Cosmology from random entanglement},
journal={Journal of High Energy Physics},
year={2023},
month={Nov},
day={27},
volume={2023},
number={11},
pages={188},
issn={1029-8479},
doi={10.1007/JHEP11(2023)188},
url={https://doi.org/10.1007/JHEP11(2023)188}
}

@article{Sahu:2025upe,
    author = "Sahu, Abhisek and van der Heijden, Jeremy and Van Raamsdonk, Mark and Zibakhsh, Rana",
    title = "{Algebras for generalized entanglement wedges}",
    eprint = "2511.21852",
    archivePrefix = "arXiv",
    primaryClass = "hep-th",
    doi = "10.1007/JHEP07(2026)192",
    journal = "JHEP",
    volume = "07",
    pages = "192",
    year = "2026"
}

@article{DeBoer:2019kdj,
    author = "de Boer, Jan and Lamprou, Lampros",
    title = "{Holographic Order from Modular Chaos}",
    eprint = "1912.02810",
    archivePrefix = "arXiv",
    primaryClass = "hep-th",
    doi = "10.1007/JHEP06(2020)024",
    journal = "JHEP",
    volume = "06",
    pages = "024",
    year = "2020"
}

@article{Furuya:2023fei,
    author = "Furuya, Keiichiro and Lashkari, Nima and Moosa, Mudassir and Ouseph, Shoy",
    title = "{Information loss, mixing and emergent type III$_{1}$ factors}",
    eprint = "2305.16028",
    archivePrefix = "arXiv",
    primaryClass = "hep-th",
    doi = "10.1007/JHEP08(2023)111",
    journal = "JHEP",
    volume = "08",
    pages = "111",
    year = "2023"
}

@article{Ouseph:2023juq,
    author        = "Ouseph, Shoy and Furuya, Keiichiro and Lashkari, Nima and Leung, Kwing Lam and Moosa, Mudassir",
    title         = "{Local Poincar\'e algebra from quantum chaos}",
    eprint        = "2310.13736",
    archivePrefix = "arXiv",
    primaryClass  = "hep-th",
    doi           = "10.1007/JHEP01(2024)112",
    journal       = "JHEP",
    volume        = "01",
    pages         = "112",
    year          = "2024"
}

@article{deBoer:2025rxx,
    author = "de Boer, Jan and Najian, Bahman and van der Heijden, Jeremy and Zukowski, Claire",
    title = "{Modular chaos, operator algebras, and the Berry phase}",
    eprint = "2505.04682",
    archivePrefix = "arXiv",
    primaryClass = "hep-th",
    doi = "10.1007/JHEP09(2025)086",
    journal = "JHEP",
    volume = "09",
    pages = "086",
    year = "2025"
}

@article{Geng:2025bcb,
    author = "Geng, Hao and Jiang, Yikun and Xu, Jiuci",
    title = "{Algebras, entanglement islands, and observers}",
    eprint = "2506.12127",
    archivePrefix = "arXiv",
    primaryClass = "hep-th",
    doi = "10.1088/1361-6382/ae811e",
    journal = "Class. Quant. Grav.",
    volume = "43",
    number = "13",
    pages = "135012",
    year = "2026"
}

@article{Klinger:2026kqj,
    author = "Klinger, Marc S.",
    title = "{How to have your wormholes and factorize, too}",
    eprint = "2602.15120",
    archivePrefix = "arXiv",
    primaryClass = "hep-th",
    month = "2",
    year = "2026"
}

@article{AliAhmad:2025oli,
    author = "Ali Ahmad, Shadi and Klinger, Marc S.",
    title = "{Extensions from within}",
    eprint = "2503.02944",
    archivePrefix = "arXiv",
    primaryClass = "hep-th",
    month = "3",
    year = "2025"
}

@article{AliAhmad:2024saq,
    author = "Ali Ahmad, Shadi and Klinger, Marc S.",
    title = "{Emergent geometry from quantum probability}",
    eprint = "2411.07288",
    archivePrefix = "arXiv",
    primaryClass = "hep-th",
    doi = "10.1103/PhysRevD.111.105015",
    journal = "Phys. Rev. D",
    volume = "111",
    number = "10",
    pages = "105015",
    year = "2025"
}

@article{Faulkner:2020iou,
    author = "Faulkner, Thomas and Hollands, Stefan and Swingle, Brian and Wang, Yixu",
    title = "{Approximate Recovery and Relative Entropy I: General von Neumann Subalgebras}",
    eprint = "2006.08002",
    archivePrefix = "arXiv",
    primaryClass = "quant-ph",
    doi = "10.1007/s00220-021-04143-6",
    journal = "Commun. Math. Phys.",
    volume = "389",
    number = "1",
    pages = "349--397",
    year = "2022"
}

@article{Faulkner:2020hzi,
    author = "Faulkner, Thomas",
    title = "{The holographic map as a conditional expectation}",
    eprint = "2008.04810",
    archivePrefix = "arXiv",
    primaryClass = "hep-th",
    month = "8",
    year = "2020"
}

@article{Beny:2007ewj,
    author = "B\'eny, C\'edric and Kempf, Achim and Kribs, David W.",
    title = "{Generalization of Quantum Error Correction via the Heisenberg Picture}",
    doi = "10.1103/PhysRevLett.98.100502",
    journal = "Phys. Rev. Lett.",
    volume = "98",
    number = "10",
    pages = "100502",
    year = "2007"
}

@article{Beny:2008,
author = {Bény, Cédric and Kempf, Achim and Kribs, David},
year = {2008},
month = {11},
pages = {},
title = "{Quantum error correction on infinite-dimensional Hilbert spaces}",
volume = {50},
journal = {Journal of Mathematical Physics},
doi = {10.1063/1.3155783}
}

@article{Furuya:2020tzv,
    author = "Furuya, Keiichiro and Lashkari, Nima and Ouseph, Shoy",
    title = "{Real-space RG, error correction and Petz map}",
    eprint = "2012.14001",
    archivePrefix = "arXiv",
    primaryClass = "hep-th",
    doi = "10.1007/JHEP01(2022)170",
    journal = "JHEP",
    volume = "01",
    pages = "170",
    year = "2022"
}

@article{hayden2007blackholes,
    author = "Hayden, Patrick and Preskill, John",
    title = "{Black holes as mirrors: Quantum information in random subsystems}",
    eprint = "0708.4025",
    archivePrefix = "arXiv",
    primaryClass = "hep-th",
    reportNumber = "CALT-68-2659",
    doi = "10.1088/1126-6708/2007/09/120",
    journal = "JHEP",
    volume = "09",
    pages = "120",
    year = "2007"
}

@article{engelhardt2025observer,
    author = "Engelhardt, Netta and Gesteau, Elliott and Harlow, Daniel",
    title = "{Observer complementarity for black holes and holography}",
    eprint = "2507.06046",
    archivePrefix = "arXiv",
    primaryClass = "hep-th",
    reportNumber = "MIT-CTP/5884",
    month = "7",
    year = "2025"
}

@article{harlow2013quantum,
    author = "Harlow, Daniel and Hayden, Patrick",
    title = "{Quantum Computation vs. Firewalls}",
    eprint = "1301.4504",
    archivePrefix = "arXiv",
    primaryClass = "hep-th",
    doi = "10.1007/JHEP06(2013)085",
    journal = "JHEP",
    volume = "06",
    pages = "085",
    year = "2013"
}

@article{faulkner2013quantum,
  title={Quantum corrections to holographic entanglement entropy},
  author={Faulkner, Thomas and Lewkowycz, Aitor and Maldacena, Juan},
  journal={Journal of High Energy Physics},
  volume={2013},
  number={11},
  pages={1--18},
  year={2013},
  publisher={Springer}
}

@article{lewkowycz2013generalised,
   title={Generalized gravitational entropy},
   volume={2013},
   ISSN={1029-8479},
   url={http://dx.doi.org/10.1007/JHEP08(2013)090},
   DOI={10.1007/jhep08(2013)090},
   number={8},
   journal={Journal of High Energy Physics},
   publisher={Springer Science and Business Media LLC},
   author={Lewkowycz, Aitor and Maldacena, Juan},
   year={2013},
   month=aug}

@article{dong2018entropy,
   title={Entropy, extremality, euclidean variations, and the equations of motion},
   volume={2018},
   ISSN={1029-8479},
   url={http://dx.doi.org/10.1007/JHEP01(2018)081},
   DOI={10.1007/jhep01(2018)081},
   number={1},
   journal={Journal of High Energy Physics},
   publisher={Springer Science and Business Media LLC},
   author={Dong, Xi and Lewkowycz, Aitor},
   year={2018},
   month=jan }

@misc{yoshida2017efficient,
      title="{Efficient decoding for the Hayden-Preskill protocol}", 
      author={Beni Yoshida and Alexei Kitaev},
      year={2017},
      eprint={1710.03363},
      archivePrefix={arXiv},
      primaryClass={hep-th}
}

@article{akers2022black,
  title={The black hole interior from non-isometric codes and complexity},
  author={Akers, Chris and Engelhardt, Netta and Harlow, Daniel and Penington, Geoff and Vardhan, Shreya},
  journal={arXiv preprint arXiv:2207.06536},
  year={2022}
}

@article{Hollands_2021,
   title="{Variational approach to relative entropies with an application to QFT}",
   volume={111},
   ISSN={1573-0530},
   url={http://dx.doi.org/10.1007/s11005-021-01474-2},
   DOI={10.1007/s11005-021-01474-2},
   number={6},
   journal={Letters in Mathematical Physics},
   publisher={Springer Science and Business Media LLC},
   author={Hollands, Stefan},
   year={2021},
   month=oct }

@book{kosaki1998type,
  title={Type III Factors and Index Theory},
  author={Kosaki, H.},
  series={Lecture Notes Series},
  url={https://books.google.nl/books?id=1W_vAAAAMAAJ},
  year={1998},
  publisher={Research Institute of Mathematics, Global Analysis Research Center, Seoul National University}
}

@article{kosaki1991index,
  title={Index theory for operator algebras},
  author={Kosaki, H},
  journal={Sugaku Expositions},
  volume={4},
  pages={177--197},
  year={1991}
}

@article{longo1989index,
  title={Index of subfactors and statistics of quantum fields. I},
  author={Longo, Roberto},
  journal={Communications in mathematical physics},
  volume={126},
  number={2},
  pages={217--247},
  year={1989},
  publisher={Springer}
}

@article{longo1990index,
  title="{Index of subfactors and statistics of quantum fields: II. Correspondences, braid group statistics and Jones polynomial}",
  author={Longo, Roberto},
  journal={Communications in mathematical physics},
  volume={130},
  number={2},
  pages={285--309},
  year={1990},
  publisher={Springer}
}

@inproceedings{pimsner1986entropy,
  title={Entropy and index for subfactors},
  author={Pimsner, Mihai and Popa, Sorin},
  booktitle={Annales scientifiques de l'Ecole normale sup{\'e}rieure},
  volume={19},
  pages={57--106},
  year={1986}
}

@article{almheiri2013black,
  title={Black holes: complementarity or firewalls?},
  author={Almheiri, Ahmed and Marolf, Donald and Polchinski, Joseph and Sully, James},
  journal={Journal of High Energy Physics},
  volume={2013},
  number={2},
  pages={1--20},
  year={2013},
  publisher={Springer}
}

@article{page1993information,
  title={Information in black hole radiation},
  author={Page, Don},
  journal={Physical review letters},
  volume={71},
  number={23},
  pages={3743},
  year={1993},
  publisher={APS}
}

@article{almheiri2020page,
   title="{The Page curve of Hawking radiation from semiclassical geometry}",
   volume={2020},
   ISSN={1029-8479},
   url={http://dx.doi.org/10.1007/JHEP03(2020)149},
   DOI={10.1007/jhep03(2020)149},
   number={3},
   journal={Journal of High Energy Physics},
   publisher={Springer Science and Business Media LLC},
   author={Almheiri, Ahmed and Mahajan, Raghu and Maldacena, Juan and Zhao, Ying},
   year={2020},
   month=mar }

@article{almheiri2013apologia,
  title={An apologia for firewalls},
  author={Almheiri, Ahmed and Marolf, Donald and Polchinski, Joseph and Stanford, Douglas and Sully, James},
  journal={Journal of High Energy Physics},
  volume={2013},
  number={9},
  pages={1--32},
  year={2013},
  publisher={Springer}
}

@article{susskind1993strechted,
	author = {Susskind, Leonard and Thorlacius, L{\'a}rus and Uglum, John},
	doi = {10.1103/physrevd.48.3743},
	issn = {0556-2821},
	journal = {Physical Review D},
	month = {Oct},
	number = {8},
	pages = {3743--3761},
	publisher = {American Physical Society (APS)},
	title = {The stretched horizon and black hole complementarity},
	url = {http://dx.doi.org/10.1103/PhysRevD.48.3743},
	volume = {48},
	year = {1993}}

@article{leutheusser2023emergent,
    author = "Leutheusser, Samuel Aaron Wehlau and Liu, Hong",
    title = "{Emergent Times in Holographic Duality}",
    eprint = "2112.12156",
    archivePrefix = "arXiv",
    primaryClass = "hep-th",
    reportNumber = "MIT-CTP/5382",
    doi = "10.1103/PhysRevD.108.086020",
    journal = "Phys. Rev. D",
    volume = "108",
    number = "8",
    pages = "086020",
    year = "2023"
}

@article{leutheusser2023causal,
    author = "Leutheusser, Samuel and Liu, Hong",
    title = "{Causal connectability between quantum systems and the black hole interior in holographic duality}",
    eprint = "2110.05497",
    archivePrefix = "arXiv",
    primaryClass = "hep-th",
    reportNumber = "MIT-CTP/5335",
    doi = "10.1103/PhysRevD.108.086019",
    journal = "Phys. Rev. D",
    volume = "108",
    number = "8",
    pages = "086019",
    year = "2023"
}

@article{vanderHeijden:2024tdk,
    author = "van der Heijden, Jeremy and Verlinde, Erik",
    title = "{An operator algebraic approach to black hole information}",
    eprint = "2408.00071",
    archivePrefix = "arXiv",
    primaryClass = "hep-th",
    doi = "10.1007/JHEP02(2025)207",
    journal = "JHEP",
    volume = "02",
    pages = "207",
    year = "2025"
}

@article{HermanOcneanu1989,
  author    = {Herman, Richard H. and Ocneanu, Adrian},
  title     = {Index theory and Galois theory for infinite index inclusions of factors},
  journal   = {Comptes Rendus de l'Académie des Sciences, Série I, Mathématique},
  volume    = {309},
  number    = {17},
  pages     = {923--927},
  year      = {1989},
  publisher = {Gauthier-Villars}
}

@article{Nill:1994be,
    author = "Nill, Florian and Wiesbrock, Hans-Werner",
    title = "{A Comment on Jones inclusions with infinite index}",
    eprint = "hep-th/9411084",
    archivePrefix = "arXiv",
    reportNumber = "SFB-288-144",
    doi = "10.1142/S0129055X95000244",
    journal = "Rev. Math. Phys.",
    volume = "7",
    pages = "599--630",
    year = "1995"
}

@article{EnockNest1996,
  author  = {Enock, Michel and Nest, Ryszard},
  title   = {Irreducible inclusions of factors, multiplicative unitaries and {Kac} algebras},
  journal = {Journal of Functional Analysis},
  volume  = {137},
  number  = {2},
  pages   = {466--543},
  year    = {1996},
  doi     = {10.1006/jfan.1996.0053},
  url     = {https://www.sciencedirect.com/science/article/pii/S0022123696900531}
}

@article{ocneanu1991quantum,
  title="{Quantum symmetry, differential geometry of finite graphs and classification of subfactors, University of Tokyo Seminary Notes 45}",
  author={Ocneanu, Adrian},
  journal={Notes recorded by Y. Kawahigashi},
  year={1991}
}

@article{2Tu2025,
    author = "Tu, Jingxin and van den Heuvel, Pim and van der Heijden, Jeremy and Verlinde, Erik",
    title = "{Black Holes, Crossed Products and the Page Curve}",
    eprint = "To appear",
    archivePrefix = "arXiv",
    primaryClass = "hep-th",
    month = "9",
    year = "2026"
}

@inbook{Ocneanu1989operator, place={Cambridge}, series={London Mathematical Society Lecture Note Series}, title={Quantized groups, string algebras, and Galois theory for algebras}, booktitle={Operator Algebras and Applications}, publisher={Cambridge University Press}, author={Ocneanu, A.}, editor={Evans, David E. and Takesaki, MasamichiEditors}, year={1989}, pages={119–172}, collection={London Mathematical Society Lecture Note Series}}

@article{jones1983index,
  title={Index for subfactors},
  author={Jones, Vaughan FR},
  journal={Inventiones mathematicae},
  volume={72},
  number={1},
  pages={1--25},
  year={1983},
  publisher={Springer}
}

@article{brown2020python,
  title="{The Python’s Lunch: geometric obstructions to decoding Hawking radiation}",
  author={Brown, Adam R and Gharibyan, Hrant and Penington, Geoff and Susskind, Leonard},
  journal={Journal of High Energy Physics},
  volume={2020},
  number={8},
  pages={1--53},
  year={2020},
  publisher={Springer}
}

@article{horowitz2004blackhole,
   title={The black hole final state},
   volume={2004},
   ISSN={1029-8479},
   url={http://dx.doi.org/10.1088/1126-6708/2004/02/008},
   DOI={10.1088/1126-6708/2004/02/008},
   number={02},
   journal={Journal of High Energy Physics},
   publisher={Springer Science and Business Media LLC},
   author={Horowitz, Gary T and Maldacena, Juan},
   year={2004},
   month=feb, pages={008–008} }

@article{stanford2014complexity,
    author = "Stanford, Douglas and Susskind, Leonard",
    title = "{Complexity and Shock Wave Geometries}",
    eprint = "1406.2678",
    archivePrefix = "arXiv",
    primaryClass = "hep-th",
    doi = "10.1103/PhysRevD.90.126007",
    journal = "Phys. Rev. D",
    volume = "90",
    number = "12",
    pages = "126007",
    year = "2014"
}

@article{almheiri2019entropy,
  title={The entropy of bulk quantum fields and the entanglement wedge of an evaporating black hole},
  author={Almheiri, Ahmed and Engelhardt, Netta and Marolf, Donald and Maxfield, Henry},
  journal={Journal of High Energy Physics},
  volume={2019},
  number={12},
  pages={1--47},
  year={2019},
  publisher={Springer}
}

@article{dong2016deriving,
   title={Deriving covariant holographic entanglement},
   volume={2016},
   ISSN={1029-8479},
   url={http://dx.doi.org/10.1007/JHEP11(2016)028},
   DOI={10.1007/jhep11(2016)028},
   number={11},
   journal={Journal of High Energy Physics},
   publisher={Springer Science and Business Media LLC},
   author={Dong, Xi and Lewkowycz, Aitor and Rangamani, Mukund},
   year={2016},
   month=nov }

@article{penington2020entanglement,
  title={Entanglement wedge reconstruction and the information paradox},
  author={Penington, Geoffrey},
  journal={Journal of High Energy Physics},
  volume={2020},
  number={9},
  pages={1--84},
  year={2020},
  publisher={Springer}
}

@article{engelhardt2015quantum,
   title={Quantum extremal surfaces: holographic entanglement entropy beyond the classical regime},
   volume={2015},
   ISSN={1029-8479},
   url={http://dx.doi.org/10.1007/JHEP01(2015)073},
   DOI={10.1007/jhep01(2015)073},
   number={1},
   journal={Journal of High Energy Physics},
   publisher={Springer Science and Business Media LLC},
   author={Engelhardt, Netta and Wall, Aron C.},
   year={2015},
   month=jan }

@article{engelhardt2021world,
  title={A world without pythons would be so simple},
  author={Engelhardt, Netta and Penington, Geoff and Shahbazi-Moghaddam, Arvin},
  journal={Classical and Quantum Gravity},
  volume={38},
  number={23},
  pages={234001},
  year={2021},
  publisher={IOP Publishing}
}

@article{engelhardt2022finding,
  title={Finding pythons in unexpected places},
  author={Engelhardt, Netta and Penington, Geoff and Shahbazi-Moghaddam, Arvin},
  journal={Classical and Quantum Gravity},
  volume={39},
  number={9},
  pages={094002},
  year={2022},
  publisher={IOP Publishing}
}

@article{Longo:1990zp,
    author = "Longo, R.",
    title = "{Index of subfactors and statistics of quantum fields. 2: Correspondences, braid group statistics and Jones polynomial}",
    doi = "10.1007/BF02473354",
    journal = "Commun. Math. Phys.",
    volume = "130",
    pages = "285--309",
    year = "1990"
}

@article{Chandrasekaran:2022eqq,
    author = "Chandrasekaran, Venkatesa and Penington, Geoff and Witten, Edward",
    title = "{Large N algebras and generalized entropy}",
    eprint = "2209.10454",
    archivePrefix = "arXiv",
    primaryClass = "hep-th",
    doi = "10.1007/JHEP04(2023)009",
    journal = "JHEP",
    volume = "04",
    pages = "009",
    year = "2023"
}

@article{Chandrasekaran:2022cip,
    author = "Chandrasekaran, Venkatesa and Longo, Roberto and Penington, Geoff and Witten, Edward",
    title = "{An algebra of observables for de Sitter space}",
    eprint = "2206.10780",
    archivePrefix = "arXiv",
    primaryClass = "hep-th",
    doi = "10.1007/JHEP02(2023)082",
    journal = "JHEP",
    volume = "02",
    pages = "082",
    year = "2023"
}

\end{document}